\documentclass{lmcs} %
\pdfoutput=1
\usepackage[utf8]{inputenc}

\usepackage{lastpage}
\lmcsdoi{21}{1}{21}
\lmcsheading{}{\pageref{LastPage}}{}{}%
{Nov.~17,~2023}{Mar.~05,~2025}{}


\usepackage{tikz}
\usetikzlibrary{cd}

\usetikzlibrary{calc, positioning}
\usetikzlibrary{automata}
\tikzset{
    initial text={},
    accepting/.style=accepting by arrow
}

\tikzcdset{
	arrow style=tikz,
	diagrams={>={Straight Barb[scale=0.8]}}
}

\usetikzlibrary{fit}
\usetikzlibrary{arrows.meta}
\tikzset{
    inode/.style={
        inner xsep=0pt,
	}
}

\usepackage{hyperref}
\usepackage{cleveref}
\usepackage{thm-restate}
\usepackage{stmaryrd}

\usepackage[xcolor, cleveref, notion, quotation, electronic]{knowledge}
\usepackage{mathcommand}
\knowledgeconfigure{protect quotation={tikzcd}}
\knowledgeconfigure{diagnose line=true, diagnose bar=true}

\definecolor{Dark Ruby Red}{HTML}{580507}
\definecolor{Dark Blue Sapphire}{HTML}{053641}
\definecolor{Dark Gamboge}{HTML}{be7c00}

\IfKnowledgePaperModeTF{
}{
    \knowledgestyle{intro notion}{color={Dark Ruby Red}, emphasize}
    \knowledgestyle{notion}{color={Dark Blue Sapphire}}
    \hypersetup{
        colorlinks=true,
        breaklinks=true,
        linkcolor={cGreen}, %
        citecolor={cGreen}, %
        filecolor={cGreen}, %
        urlcolor={cGreen},
    }
    \IfKnowledgeElectronicModeTF{
    }{
        \knowledgeconfigure{anchor point color={Dark Ruby Red}, anchor point shape=corner}
        \knowledgestyle{intro unknown}{color={Dark Gamboge}, emphasize}
        \knowledgestyle{intro unknown cont}{color={Dark Gamboge}, emphasize}
        \knowledgestyle{kl unknown}{color={Dark Gamboge}}
        \knowledgestyle{kl unknown cont}{color={Dark Gamboge}}
    }
}

\input{kl/abbreviations.kl}
\input{kl/complexity.kl}
\input{kl/queries.kl}
\input{kl/treewidth.kl}
\input{kl/semantic-treewidth.kl}
\input{kl/simple-queries.kl}
\input{kl/misc.kl}

\usepackage{amssymb} %
\usepackage{dutchcal} %
\usepackage{xspace} %
\usepackage{booktabs} %
\usepackage{bbding} %

\usepackage{caption}
\usepackage{subcaption}
\usepackage{adforn}

\definecolor{Desire}{HTML}{eb3b5a} %
\definecolor{Boyzone}{HTML}{2d98da} %
\definecolor{Royal Blue}{HTML}{3867d6} %
\definecolor{NYC Taxi}{HTML}{f7b731} %
\definecolor{Beniukon Orange}{HTML}{fa8231}
\definecolor{Algal Fuel}{HTML}{20bf6b} %
\definecolor{Innuendo}{HTML}{a5b1c2} %
\definecolor{Twinkle Blue}{HTML}{d1d8e0} %
\definecolor{Blue Horizon}{HTML}{4b6584} %
\definecolor{Gloomy Purple}{HTML}{8854d0} %

\colorlet{cBlue}{Royal Blue}
\colorlet{cYellow}{NYC Taxi}
\colorlet{cOrange}{Beniukon Orange}
\colorlet{cGreen}{Algal Fuel}
\colorlet{cRed}{Desire}
\colorlet{cGrey}{Innuendo}
\colorlet{cDarkGrey}{Blue Horizon}
\colorlet{cLightGrey}{Twinkle Blue}
\colorlet{cPurple}{Gloomy Purple}

\renewcommand{\epsilon}{\varepsilon}

\newif\ifproofappendix

\newrobustcmd\introinrestatable[1]{%
\ifproofappendix%
\kl{#1}%
\else%
\intro{#1}%
\fi%
}

\newrobustcmd\introstarinrestatable[1]{%
\ifproofappendix%
\kl{#1}%
\else%
\intro*#1%
\fi%
}

\newrobustcmd\introinrestatableopt[1]{%
\ifproofappendix%
\kl[#1]{#1}%
\else%
\intro[#1]{#1}%
\fi%
}

\usepackage{xpatch}
\usepackage{apptools}
\makeatletter
\xpatchcmd{\thmt@restatable}%
{\csname #2\@xa\endcsname\ifx\@nx#1\@nx\else[{#1}]\fi}%
{\IfAppendix{\csname #2\@xa\endcsname}{\csname #2\@xa\endcsname\ifx\@nx#1\@nx\else[{#1}]\fi}}%
{}{} %
\makeatother

\newrobustcmd\recall[1]{
  \proofappendixtrue%
    #1*
  \proofappendixfalse%
}

\usepackage[backgroundcolor=orange!20, textcolor={Dark Ruby Red}, textsize=tiny,disable]{todonotes}
\definecolor{green}{RGB}{0,120,0}
\definecolor{hlyellow}{RGB}{250, 250, 190}
\definecolor{diegoeditcolor}{RGB}{210,210,255}
\definecolor{remieditcolor}{RGB}{210,255,210}
\newcommand{\sidediego}[1]{}
\newcommand{\sideremi}[1]{}
\newcommand{\remi}[1]{}

\newcommand{\diego}[1]{}
\definecolor{light-gray}{gray}{0.9}

\newcommand{\proofcase}[1]{\textbf{#1}~}

\colorlet{wrote}{cRed}
\colorlet{advised}{cBlue}
\newrobustcmd{\wrote}{\color{wrote}\scriptsize\text{wrote}}
\newrobustcmd{\advised}{\color{advised}\scriptsize\text{advised}}

\renewcommand{\phi}{\varphi}
\renewcommand{\leq}{\leqslant}
\renewcommand{\geq}{\geqslant}
\renewcommand{\emptyset}{\varnothing}
\knowledgenewrobustcmd{\lBrack}{\cmdkl{\llbracket}}
\knowledgenewrobustcmd{\rBrack}{\cmdkl{\rrbracket}}

\newcommand{\set}[1]{\{#1\}}
\newrobustcmd{\defeq}{\mathrel{\hat{=}}}
\newrobustcmd{\id}{\mathrm{id}}

\newrobustcmd{\Nat}{\mathbb{N}}
\newrobustcmd\pset[1]{\+P(#1)} %
\newrobustcmd{\pto}{\rightharpoonup} %
\newrobustcmd{\dom}{\mathop{\textrm{dom}}} %
\newrobustcmd{\poly}{\mathop{\textrm{poly}}} %

\knowledgenewrobustcmd{\A}{\mathbb{A}} %
\knowledgenewrobustcmd{\Aext}{\cmdkl{\mathbb{A}^\pm}} %
\knowledgenewrobustcmd\vertex[1]{\cmdkl{V}(#1)}
\knowledgenewrobustcmd\edges[1]{\cmdkl{E}(#1)}

\knowledgenewrobustcmd\qvar{\footnotesize\bullet} %

\knowledgenewrobustcmd{\pathl}{\cmdkl{\mathbf{P}_{\!l}}} %

\knowledgenewrobustcmd\subaut[3]{#1\cmdkl{[#2,#3]}}

\knowledgenewrobustcmd\bagmap{\cmdkl{\mathbf{v}}}
\knowledgenewrobustcmd\tagmap{\cmdkl{\mathbf{t}}}
\knowledgenewrobustcmd\tagmappath[1]{\cmdkl{\mathbf{t}[#1]}}
\newrobustcmd\tagmappathprime[1]{%
  \withkl{\kl[\tagmappath]}{%
    \cmdkl{\mathbf{t}'[#1]}%
  }%
}

\knowledgenewrobustcmd{\atom}[1]{\,\raisebox{-.09em}{$\xrightarrow{\smash{#1}}$}\,}
\knowledgenewrobustcmd{\coatom}[1]{\,\raisebox{-.09em}{$\xleftarrow{\smash{#1}}$}\,}
\knowledgenewrobustcmd{\symatom}[1]{\,\raisebox{-.09em}{$\xrightleftharpoons{\smash{#1}}$}\,}
\knowledgenewrobustcmd{\atoms}[1]{\cmdkl{\textnormal{Atoms}}(#1)}
\knowledgenewrobustcmd{\contained}{\mathrel{\cmdkl{\subseteqq}}}
\newrobustcmd{\strcontained}{
  \mathrel{\withkl{\kl[\contained]}{\cmdkl{%
    \subsetneqq
  }}}
}
\disablecommand\equiv
\suggestcommand\equiv{Use instead \semequiv for semantical equivalence.}
\knowledgenewrobustcmd{\semequiv}{\mathrel{\cmdkl{\LaTeXequiv}}}
\knowledgenewrobustcmd{\nbatoms}[2][]{\cmdkl{\|}#2\cmdkl{\|^{#1}_{\textrm{at}}}}
\knowledgenewrobustcmd{\size}[1]{\cmdkl{\|}#1\cmdkl{\|}}
\knowledgenewrobustcmd{\vars}{\cmdkl{\textit{vars}}} %
\newrobustcmd{\collapse}{\approx}
\knowledgenewrobustcmd{\Exp}{\cmdkl{\textnormal{Exp}}} %

\knowledgenewrobustcmd{\UCtwoRPQ}{\cmdkl{\textnormal{UC2RPQ}}}
\newrobustcmd{\UCRPQ}{\kl[\UCtwoRPQ]{\textnormal{UCRPQ}}}
\newrobustcmd{\CtwoRPQ}{%
  \withkl{\kl[\UCtwoRPQ]}{\cmdkl{%
    \textnormal{C2RPQ}
  }}
}

\knowledgenewrobustcmd{\UCRPQSRE}{\ensuremath{\cmdkl{\textup{UCRPQ}(\textup{SRE})}}}
\newrobustcmd{\CRPQSRE}{%
  \withkl{\kl[\UCRPQSRE]}{\cmdkl{%
    \textup{CRPQ}(\textup{SRE})%
  }}%
}%

\knowledgenewrobustcmd{\homto}{\mathrel{\cmdkl{\raisebox{-.09em}{$\xrightarrow{\smash{\textit{\tiny hom}}}$}}}}
\newrobustcmd{\nothomto}{\mathrel{\withkl{\kl[\homto]}{\cmdkl{\raisebox{-.09em}{$\not\xrightarrow{\smash{\textit{\tiny hom}}}$}}}}}
\newcommand{\xrightarrowdbl}[2][]{%
  \xrightarrow[#1]{#2}\hspace{-.8em}\xrightarrow{}
}
\knowledgenewrobustcmd\surj{\mathrel{\cmdkl{\raisebox{-.09em}{$\xrightarrowdbl{\smash{\textit{\tiny hom}}}$}}}}
\knowledgenewrobustcmd{\fun}{f}

\knowledgenewrobustcmd{\class}{\mathcal{C}}
\knowledgenewrobustcmd{\Tw}[1][k]{\cmdkl{\mathcal{T\hspace{-.15em}w}_{#1\!}}}
\knowledgenewrobustcmd{\Pw}[1][k]{\cmdkl{\mathcal{P\hspace{-.05em}w}_{#1\!}}}
\knowledgenewrobustcmd{\ContrPw}[1][k]{\cmdkl{\mathcal{C\hspace{-.05em}p\hspace{-.05em}w}_{#1\!}}}
\knowledgenewrobustcmd{\PwOneWay}[1][k]{\cmdkl{\mathcal{1P\hspace{-.05em}w}_{#1\!}}}
\knowledgenewrobustcmd{\ContrPwOneWay}[1][k]{\cmdkl{\mathcal{1\hspace{-.08em}C\hspace{-.05em}p\hspace{-.05em}w}_{#1\!}}}
\knowledgenewrobustcmd{\ContrTw}[1][k]{\cmdkl{\mathcal{C\hspace{-.05em}t\hspace{-.05em}w}_{#1\!}}}
\knowledgenewrobustcmd{\TwOneWay}[1][k]{\cmdkl{\mathcal{1T\hspace{-.15em}w}_{#1\!}}}
\knowledgenewrobustcmd{\ContrTwOneWay}[1][k]{\cmdkl{\mathcal{1\hspace{-.08em}C\hspace{-.05em}t\hspace{-.05em}w}_{#1\!}}}
\newrobustcmd\bw{\textrm{bw}}

\newrobustcmd{\core}{\mathop{\textrm{core}}}

\newcommand{\anexpansion}{\xi}
\knowledgenewrobustcmd{\Refin}[1][]{\cmdkl{\textnormal{Ref}^{\smash{#1}}}}
\knowledgenewrobustcmd{\MUA}[2]{\cmdkl{\ensuremath{\textnormal{App}_{#2}(#1)}}}
\knowledgenewrobustcmd{\MUAHom}[2]{\cmdkl{\ensuremath{\textnormal{App}_{#2}^{\smash{\star}}(#1)}}}
\knowledgenewrobustcmd{\MUAHomBounded}[3]{\cmdkl{\ensuremath{\textnormal{App}_{#2}^{\smash{\star,#3}}(#1)}}}
\knowledgenewrobustcmd{\type}{\cmdkl{\textnormal{type}}}
\knowledgenewrobustcmd{\Qapp}[1][k]{\cmdkl{\ensuremath{\textnormal{App}_{\Tw[#1]}^{\smash{\textup{zip}}}(\gamma)}}}
\knowledgenewrobustcmd{\contract}[1]{\cmdkl{[}#1\cmdkl{]}}

\knowledgenewrobustcmd\upquery[1]{#1^{\cmdkl{\uparrow}}}

\disablecommand\ell
\suggestcommand\ell{Use instead \l for bound on the size of refinements.}
\knowledgerenewcommand{\l}{\cmdkl{\LaTeXell}}
\knowledgenewrobustcmd{\lOne}{\cmdkl{\LaTeXell_1}}
\knowledgenewrobustcmd{\avoid}{\textsf{avoid}}
\knowledgenewrobustcmd{\trap}{\textsf{trap}}

\newcommand{\pspace}{\textup{\textsf{PSpace}}\xspace}

\newcommand{\expspace}{{\sf ExpSpace}\xspace}

\newcommand{\conp}{{\sf coNP}\xspace}

\newcommand{\wone}{\textup{W[1]}\xspace}

\newrobustcmd\pitwo{\ensuremath{\Pi^p_2}}

\newrobustcmd\sigmatwo{\ensuremath{\Sigma^p_2}}

\knowledgenewrobustcmd\itemClosureInfCQ{\cmdkl{\textup{\textrm{(1)}}}}
\knowledgenewrobustcmd\itemClosureUCRPQ{\cmdkl{\textup{\textrm{(2)}}}}
\knowledgenewrobustcmd\itemClosureUCRPQSimple{\cmdkl{\textup{\textrm{(3)}}}}

\knowledgenewrobustcmd{\disjointUnion}{\mathrel{\cmdkl{+}}}

\graphicspath{{./fig/}}

\NewCommandCopy{\proofqedsymbol}{\qedsymbol}%
\newcommand{\exampleqedsymbol}{{$\triangle$}}%
\AtBeginEnvironment{proof}{\renewcommand{\qedsymbol}{\proofqedsymbol}}
\AtBeginEnvironment{example}{%
  \pushQED{\qed}\renewcommand{\qedsymbol}{\exampleqedsymbol}%
}
\AtEndEnvironment{example}{\popQED\endexample}
\AtBeginEnvironment{exa}{%
  \pushQED{\qed}\renewcommand{\qedsymbol}{\exampleqedsymbol}%
}
\AtEndEnvironment{exa}{\popQED\endexample}

\crefname{thm}{Theorem}{Theorems}

\crefname{defi}{Definition}{Definitions}

\let\endexample\endexa
\crefname{exa}{Example}{Examples}

\crefname{rem}{Remark}{Remarks}

\crefname{obs}{Observation}{Observations}

\crefname{cor}{Corollary}{Corollaries}

\crefname{lem}{Lemma}{Lemmata}

\crefname{prop}{Proposition}{Propositions}

\crefname{clm}{Claim}{Claims}
\crefname{fact}{Fact}{Facts}
\crefname{nota}{Notation}{Notations}
\crefname{qu}{Question}{Questions}

\begin{document}

\title[Semantic Tree-Width and Path-Width of CRPQs]{Semantic Tree-Width and Path-Width of\\ Conjunctive Regular Path Queries\rsuper*}
\titlecomment{{\lsuper*}This is the journal version of an ICDT'23 paper \cite{thispaperICDT}, see \Cref{sec:conf-paper-diff} for a summary of the added material. Submitted in November 2023.}

\author[D.~Figueira]{Diego Figueira\lmcsorcid{0000-0003-0114-2257}}
\author[R.~Morvan]{Rémi Morvan\lmcsorcid{0000-0002-1418-3405}}

\address{Univ. Bordeaux, CNRS, Bordeaux INP, LaBRI, UMR5800, F-33400 Talence, France}
\email{diego.figueira@cnrs.fr, remi.morvan@u-bordeaux.fr}

\begin{abstract}
  
We show that the problem of whether a query is equivalent to a query of tree-width $k$ is decidable, for the class of Unions of Conjunctive Regular Path Queries with two-way navigation (UC2RPQs). 
A previous result by Barceló, Romero, and Vardi \cite{BarceloRV16} has shown decidability for the case $k=1$, and here we extend this result showing that decidability in fact holds for any arbitrary $k\geq 1$. 
The algorithm is in 2\expspace, but for the restricted but practically relevant case where all regular expressions of the query are of the form $a^*$ or $(a_1 + \dotsb + a_n)$ we show that the complexity of the problem drops to \pitwo.

We also investigate the related problem of approximating a UC2RPQ by queries of small tree-width. We exhibit an algorithm which, for any fixed number $k$, builds the maximal under-approximation of tree-width $k$ of a UC2RPQ.
The maximal under-approximation of tree-width $k$ of a query $q$ is a query $q'$ of tree-width $k$ which is contained in $q$ in a maximal and unique way, that is, such that for every query $q''$ of tree-width $k$, if $q''$ is contained in $q$ then $q''$ is also contained in $q'$.

Our approach is shown to be robust, in the sense that it allows also to test equivalence with queries of a given path-width, it also covers the previously known result for $k=1$, and it allows to test for equivalence of whether a (one-way) UCRPQ is equivalent to a UCRPQ of a given tree-width (or path-width).

\end{abstract}

\maketitle
\bigskip

\noindent
\raisebox{-.4ex}{\HandRight}\ \ This pdf contains internal links: clicking on a "notion@@notice" leads to its \AP ""definition@@notice"".\footnote{This result was achieved by using the "knowledge" package and its companion tool "knowledge-clustering".}

\clearpage
\setcounter{tocdepth}{2}
\tableofcontents
\clearpage
\section{\AP{}Introduction}
\label{sec:intro}

\subsection{\AP{}Graph Databases}

"Graph databases" have gained significant attention due to their ability to efficiently model and manage complex, interconnected data. Unlike traditional relational databases, they model data as entities connected by edges that represent relationships. This structure facilitates the analysis of highly interconnected data, where the topology of the connections is as crucial as the data itself, making them particularly well-suited for use cases like biology, social networks, banking, recommendation systems, and fraud detection. We refer the reader to \cite{Barcelo-pods13,DBLP:journals/sigmod/Wood12,survey-graphdbs2017} for surveys on the foundations and applications of graph databases.

\AP ""Graph databases"" are abstracted as edge-labelled directed graphs
$G = \langle \vertex{G}, \edges{G} \rangle$, 
where nodes of $\intro*\vertex{G}$ represent entities and labelled edges $\intro*\edges{G} \subseteq \vertex G \times \A \times \vertex G$
represent relations between these entities, with $\A$ being a fixed finite alphabet.
For instance, \Cref{fig:example-graph-database} depicts a "graph database",
whose nodes are authors and papers, on the alphabet
$\A = \{\text{\color{wrote}wrote},\,\text{\color{advised}advised}\}$.
Edges $x \atom{\wrote} y$ indicate that the person $x$ wrote the paper $y$,
while edges $x \atom{\advised} y$ indicate that person $x$ was
the Ph.D. advisor of person $y$.

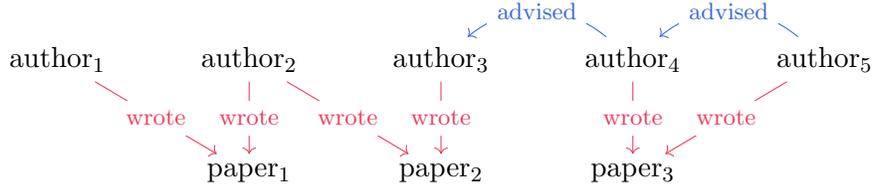
\begin{figure}[ht]
    \centering
    \begin{tikzpicture}
        \node (a1) {author$_1$};
        \node (a2) [right = of a1] {author$_2$};
        \node (a3) [right = of a2] {author$_3$};
        \node (a4) [right = of a3] {author$_4$};
        \node (a5) [right = of a4] {author$_5$};
        \node (p1) [below = 2.5em of a2] {paper$_1$};
        \node (p2) [below = 2.5em of a3] {paper$_2$};
        \node (p3) [below = 2.5em of a4] {paper$_3$};

        \draw[wrote] (a1) edge[->] node[fill=white] {\footnotesize wrote} (p1)
            (a2) edge[->] node[fill=white] {\footnotesize wrote} (p1)
            (a2) edge[->] node[fill=white] {\footnotesize wrote} (p2)
            (a3) edge[->] node[fill=white] {\footnotesize wrote} (p2)
            (a4) edge[->] node[fill=white] {\footnotesize wrote} (p3)
            (a5) edge[->] node[fill=white] {\footnotesize wrote} (p3);
        \draw[advised] (a4) edge[->, bend right=40] node[fill=white] {\footnotesize advised} (a3)
            (a5) edge[->, bend right=40] node[fill=white] {\footnotesize advised} (a4);
    \end{tikzpicture}
    \caption{%
        \AP\label{fig:example-graph-database}%
        A "graph database" with eight nodes and eight edges on a two-letter alphabet.
    }
\end{figure}

\AP Being a subclass of relational databases, "graph databases" can be queried by the
predominant query language of ""conjunctive queries"", "aka" "CQs", 
which consists of the closure under projection---"aka" existential quantification---of conjunctions of atoms of the form $x \atom{a} y$
for some letter $a \in \A$. For instance, the "conjunctive query"
\[
    \gamma_1(x, y) = x \atom{\wrote} z
        \land y \atom{\wrote} z    
\]
returns, when "evaluated" on the "graph database" $G$
defined in \Cref{fig:example-graph-database}, all pairs of nodes $(u, v)$ such that $u$ is a co-author
of $v$. Each variable not appearing in the left-hand side of 
the definition of a "conjunctive query" (in this example, $z$) is implicitly 
existentially quantified. 
Note that, to the cost of losing the information of which variable is existentially quantified, every "CQ"
can be seen as a "graph database", where each variable is a node, and each atom is an edge; 
hence, we sometimes use "graph database" terminology for "CQs".

The expressive power of "CQs" is somewhat limited, since
"CQs" cannot express, for example, transitive closure.
Since the ability to navigate paths is of importance in many "graph database" 
scenarios, most modern graph query languages support, as a central querying mechanism,
"conjunctive regular path queries", or "CRPQs" for short. 
In particular, "CRPQs" form the core navigational mechanism of the new ISO standard Graph Query Language (GQL) \cite{isoGQL} and the SQL extension for querying graph-structured data SQL/PGQ \cite{isoPGQ} (see also \cite{DBLP:conf/icdt/FrancisGGLMMMPR23,DBLP:conf/pods/FrancisGGLMMMPR23}).

"CRPQs" are 
defined analogously to "conjunctive queries", except that their atoms are now of the form 
$x \atom{L} y$ where $L$ is an arbitrary regular language over the alphabet $\A$. For 
instance the "evaluation" of the "CRPQ"
\[
    \gamma_2(x, y) = x \atom{\wrote} z
        \land z' \atom{\wrote} z 
        \land y \atom{({\advised})^*} z'
\]
on $G$ yields every pair of persons $(u,v)$ such that $u$ is a co-author of a
``scientific descendant'' of $v$. 

\AP Formally, a ""CRPQ"" $\gamma$ is defined as a tuple $\bar z = (z_1,\hdots,z_n)$
of ""output variables"", "aka" \reintro{free variables},\footnote{For technical reasons (see the definition of "equality atoms") we allow for a variable to appear multiple times.}
together with a conjunction of ""atoms"" of the form
$\bigwedge_{j=1}^m x_j \atom{L_j} y_j$, where each $L_j$ is a regular language and where $m \geq 0$.
The set of all variables occurring in $\gamma$, namely\footnote{We neither assume 
disjointness nor inclusion between $\{z_1,\hdots,z_n\}$ and $\{x_1,y_1,\hdots,x_m,y_m\}$}
$\{z_1,\hdots,z_n\}\cup\{x_1,y_1,\hdots,x_m,y_m\}$, is denoted by
$\intro*\vars(\gamma)$.
Given a "database" $G$, we say that a tuple of nodes $\bar u = (u_1,\hdots,u_n)$
\AP""satisfies"" $\gamma$ 
on $G$ if there is a mapping
$\fun\colon \vars(\gamma) \to \vertex{G}$ such that $u_i = \fun(z_i)$ for all
$1 \leq i \leq n$, and for each $1 \leq j \leq m$,
there exists a path from $\fun(x_i)$ to $\fun(y_i)$ in $G$, labelled by
a word from $L_i$ (if the path is empty, the label is $\epsilon$). The \AP""evaluation"" of $\gamma$ on $G$ is then the set of all tuples that "satisfy" $\gamma$.
For example, $(\text{author}_2, \text{author}_5)$ "satisfies" $\gamma_2$ 
on the "graph database" $G$ of \Cref{fig:example-graph-database} via
the function that maps $x$ to $\text{author}_2$, $y$ to $\text{author}_5$,
$z$ to $\text{paper}_2$, and $z'$ to $\text{author}_3$.

\AP
The language of "CRPQ" can be extended to navigate edges in both directions. 
\knowledgenewrobustcmd{\Gpm}{\cmdkl{G^\pm}}
Consider the expanded database $\intro*\Gpm$ obtained from $G$ by 
adding, for every edge $x \atom{a} y$ in $G$, an extra edge $y \atom{a^-} x$.
We obtain a graph database on the alphabet $\intro*\Aext = \A \cup \A^-$ where
$\A^- = \set{a^- \mid a \in \A}$. We then define the syntax of
a \AP""CRPQ with two-way navigation"", or \reintro{C2RPQ}, as a "CRPQ" on the alphabet $\Aext$.
Its \reintro{evaluation} is defined as the "evaluation" of the "CRPQ" on $\Gpm$.
For instance, the "evaluation" of the "C2RPQ"
\[
    \gamma_3(x, y) = x \atom{({\wrote}\cdot{\wrote}^-)^*} y
\]
on the "graph database" of \Cref{fig:example-graph-database} returns all pairs of
individuals linked by a chain of co-authorship.
It includes $(\text{author}_1, \text{author}_3)$ or $(\text{author}_1, \text{author}_1)$
but not $(\text{author}_1, \text{author}_4)$.
\AP If a query has no "output variables" we call it ""Boolean"", and
its "evaluation" can either be the set $\set{()}$, in which case we say that $G$
\reintro{satisfies} the query, or the empty set $\set{}$. For example, $G$ "satisfies" the
"Boolean CRPQ"
\[\gamma_4() = x \atom{\wrote} y\]
if, and only if, the database contains one author together with one paper they wrote.

To simplify proofs, we assume that all the regular languages are described via 
non-deterministic  finite automata (NFA) instead of regular expressions,
which does not affect any of our complexity bounds.
However, for readability all our examples will be given in terms of regular expressions.
We denote the set of "atoms" of a "C2RPQ" $\gamma$ by \AP$\intro*\atoms\gamma$, and by 
$\intro*\nbatoms{\gamma}$ we denote its number of "atoms", "ie", $|\atoms\gamma|$.
Moreover, we denote by $\intro*\size{\gamma}$ the sum of its number of "atoms" with
the sum of the size of NFAs used to describe $\gamma$.

\AP Finally, a ""union of CQs"" (\reintro{UCQs}) (resp.\ ""union of CRPQs"" (\reintro{UCRPQs}), resp.\ ""union of C2RPQs"" (\reintro{UC2RPQs}))  
is defined as a finite set of "CQs" (resp.\ "CRPQs", resp.\ "C2RPQs"), whose
tuples of "output variables" have all the same arity. 
\AP
A ""subquery"" of a "C2RPQ" $\gamma$ is any "C2RPQ" resulting from removing some "atoms" (possibly none) from $\gamma$. A \reintro{subquery} of  "UC2RPQ" is a union of "subqueries" of the "C2RPQs" therein.
The "evaluation" of a union is defined as the union of its "evaluations", for instance the following "UCQ"
\begin{align*}
    \Gamma_5 & = \gamma_{5}^1(x, y) \lor \gamma_{5}^2(x, y) \\
    & \text{ where }
    \gamma_{5}^1(x, y) = x \atom{\wrote} y %
    \text{ and }
    \gamma_{5}^2(x, y) = x \coatom{\advised} z \land
        z \atom{\wrote} y
\end{align*}
"evaluates" to the set of pairs $(x,y)$ such that $y$ is a paper written by either $x$
or their advisor.
We naturally extend the notations $\nbatoms{-}$ and $\size{-}$ to "unions@UC2RPQs".
\AP ""Infinitary unions"" are defined analogously, except
that we allow for potentially infinite unions. We often use a set notation to denote the union, especially for "infinitary unions".

For a more detailed introduction to "CRPQs", we refer the reader to \cite{DBLP:conf/rweb/Figueira21}.
For a more general introduction to different query languages for "graph databases"---including "CRPQs"---see \cite{barcelo2013querying}, and for a more practical approach,
see \cite{survey-graphdbs2017}.

\smallskip

\AP %
The ""evaluation problem"" for "UC2RPQ" is the problem of, given
a "UC2RPQ" $\Gamma$, a "graph database" $G$ and a tuple $\bar u$ of elements of $G$,
whether $\bar u$ "satisfies" $\Gamma$ on $G$. 
Given two "UC2RPQ" $\Gamma$
and $\Gamma'$ whose "output variables" have the same arity,
we say that $\Gamma$ is \AP""contained"" in $\Gamma'$,
denoted by $\Gamma \intro*\contained \Gamma'$ if
for every "graph database" $G$, for every tuple $\bar u$ of $G$,
if $\bar u$ "satisfies" $\Gamma$ on $G$, then so does $\Gamma'$ (we will hence reserve the symbol `$\subseteq$' for set inclusion---note in particular that inclusion (of the "UC2RPQs", seen as sets of "C2RPQs") implies "containment", but the converse does not hold). 
The \AP""containment problem"" for "UC2RPQs" is the problem of, given
two "UC2RPQs" $\Gamma$ and $\Gamma'$, to decide if $\Gamma \contained \Gamma'$.
When $\Gamma$ is "contained" in $\Gamma'$ and vice versa, we say that
$\Gamma$ and $\Gamma'$ are \AP""semantically equivalent"", denoted by
$\Gamma \intro*\semequiv \Gamma'$.  

\paragraph{Queries of small tree-width}
It is known that the "evaluation problem" for "UC2RPQ" is "NP"-complete, just as for "conjunctive queries"
\cite[Theorem 7]{DBLP:conf/stoc/ChandraM77}.
However, queries whose underlying structure looks like a tree---formally, queries of bounded 
"tree-width"---can be "evaluated" in polynomial time \cite[Theorem 3]{CHEKURI2000211}.\footnote{Theorem 3 talks about query "containment" of "CQs", which is in fact equivalent to the "evaluation problem" for "CQs". Moreover, the theorem deals with ``query width'', but this parameter is equivalent up to a multiplicative constant to the "tree-width" \cite[Lemma 2]{CHEKURI2000211} assuming that the database signature arity is fixed.} 

\AP
"Tree-width" is a measure of how much a
graph differs from a tree, introduced by Arnborg, Corneil and
Proskurowski~\cite{arnborg1987complexity}.
For a gentle but thorough introduction to "tree-width", we refer the reader to
\cite[\S 3.6]{sparsity12}. Formally, a \AP""tree decomposition"" of a multigraph $G$ is a
pair $(T, \intro*\bagmap)$ where $T$ is a tree and $\bagmap: \vertex{T} \to \pset{\vars(G)}$ is a function that associates to each node of $T$, called \AP""bag"",
a set of vertices of $G$. When $x \in \bagmap(b)$ we shall say that the "bag"
$b \in \vertex{T}$ \AP""contains@@tw"" vertex $v$. Further, it must satisfy the following three properties:
\begin{itemize}
    \item each vertex $v$ of $\gamma$ is "contained@@tw" in at least one "bag" of $T$;
    \item for each edge $u \atom{} v$ of $G$, there is at least one "bag" of $T$
        that "contains@@tw" both $u$ and $v$; and 
    \item for each vertex $v$ of $G$, the set of bags of $T$ "containing@@tw" $v$ is a 
        connected subset of $\vertex{T}$.
\end{itemize}
The \AP""width"" of $(T, \bagmap)$ is the maximum of $|\bagmap(b)|-1$ when $b$ ranges over
$\vertex{T}$.

\begin{figure}
	\centering
	\begin{subfigure}{.6\linewidth}
		\centering
		\includegraphics[scale=.8]{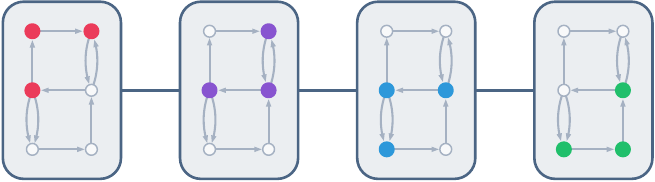}
		\caption{
			\AP\label{fig:ex-tree-dec-full}
			``Full'' representation of $(T, \bagmap)$.
		}
	\end{subfigure}
	\hfill
	\begin{subfigure}{.38\linewidth}
		\centering
		\includegraphics[scale=1.1]{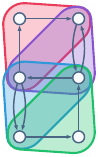}
		\caption{
			\AP\label{fig:ex-tree-dec-concise}
			``Concise'' representation of $(T, \bagmap)$
		}
	\end{subfigure}
	\caption{
		\AP\label{fig:ex-tree-dec}
		Two different representations of the same "tree decomposition" $(T, \bagmap)$ of a 
		multigraph $G$ with six 
		vertices. The underlying tree is a path with four nodes and each "bag" contains 3 
		vertices---hence the "decomposition@tree decomposition" has "width" 2.
	}
\end{figure}
We give an example of "tree decomposition" in \Cref{fig:ex-tree-dec}:
\begin{itemize}
	\item In \Cref{fig:ex-tree-dec-full}, we give the ``full'' representation of
	the "decomposition@tree decomposition": we draw $T$, and inside each "bag" $b$ of $T$
	we represent a copy of $G$. Nodes of $G$ belonging to $b$ are highlighted, while the others are dimmed. Sometimes, we will only write the
	name of the nodes contained in the "bag", instead of drawing the graph.
	\item In \Cref{fig:ex-tree-dec-concise}, we give a ``concise'' representation:
	we draw over $G$ a coloured shape for each "bag" of $T$. This representation is ambiguous---the structure of $T$ is not made explicit---and will only be used when no ambiguity can arise.
\end{itemize}

The \AP""tree-width"" of $G$ is the minimum of the "width" of all "tree decompositions" of $G$. The \reintro{tree-width} of a "C2RPQ" is the "tree-width" of its underlying multigraph. 
We denote by $\intro*\Tw$ (resp.\ $\intro*\TwOneWay$) the set of all "C2RPQs" (resp.\ "CRPQs") of "tree-width" at most $k$. The \reintro{tree-width} of a "UC2RPQ" is simply the maximum of the "tree-width" of its "C2RPQs".
\AP
A ""path decomposition"" is a "tree decomposition" $(T,\bagmap)$ in which $T$ is a path. The \AP""path-width"" of $\gamma$ is the minimum of the "width" among all "path decompositions" of $\gamma$. The \reintro{path-width} of a "C2RPQ" and "UC2RPQ" are defined analogously. We denote by \AP$\intro*\Pw$ (resp.\ $\intro*\PwOneWay$) the set of all "C2RPQs" (resp.\ "CRPQs") of "path-width" at most $k$. The relationship between these classes is depicted in \Cref{fig:taonomy-syntactic}:
note that $\TwOneWay$ and $\PwOneWay$ are not explicitly drawn, but correspond to the
intersection of $\Tw$ (resp.\ $\Pw$) with the class of "CRPQs".

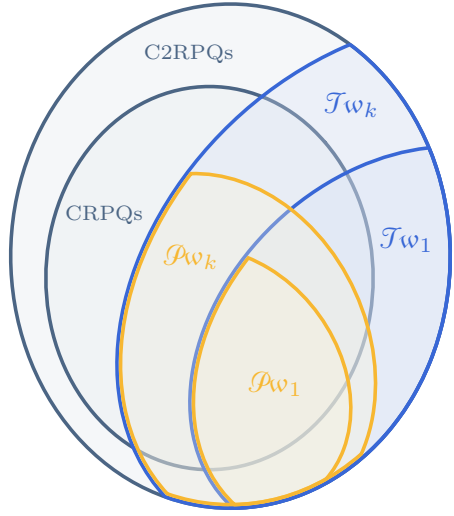
\begin{figure}
	\centering
	\scalebox{1.125}{
	\begin{tikzpicture}
		\node at (0,0) {\includegraphics[scale=.8]{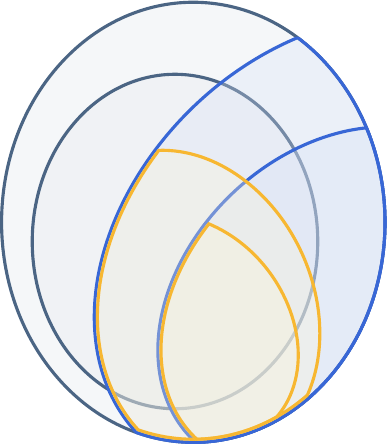}};
		\node[font=\small] at (.5,-1.5)
			{$\withkl{\kl[\Pw]}{\cmdkl{\color{cYellow}\mathcal{P\hspace{-.15em}w}_{1\!}}}$};
		\node[font=\small] at (-.5,0) 
			{$\withkl{\kl[\Pw]}{\cmdkl{\color{cYellow}\mathcal{P\hspace{-.15em}w}_{k\!}}}$};
		\node[font=\small] at (2.05,.2)
			{$\withkl{\kl[\Tw]}{\cmdkl{\color{cBlue}\mathcal{T\hspace{-.15em}w}_{1\!}}}$};
		\node[font=\small] at (1.4,1.8)
			{$\withkl{\kl[\Tw]}{\cmdkl{\color{cBlue}\mathcal{T\hspace{-.15em}w}_{k\!}}}$};
		\node[font=\tiny] at (-1.5,.5) {\kl[CRPQ]{\color{cDarkGrey}CRPQs}};
		\node[font=\tiny] at (-.5,2.4) {\kl[C2RPQ]{\color{cDarkGrey}C2RPQs}};
	\end{tikzpicture}
	}
	\caption{
		\AP\label{fig:taonomy-syntactic}
		Clickable taxonomy of syntactic classes studied in this paper.
	}
\end{figure}

Similar statements of the following proposition can be considered Folklore (see "eg" {\cite[Theorem IV.3]{DBLP:conf/lics/0001BV17}}); however, our inability to find a proof for it with sharp bounds invites us to include a proof.
\begin{restatable}[Proof in \Cref{apdx-sec:prop:crpq-bound-tree-width-upper-bound}]{prop}{crpqboundtwupperbound}
	\AP\label{prop:crpq-bound-tree-width-upper-bound}
    For each $k \geq 1$, the "evaluation problem" for "UC2RPQs" of "tree-width" at
    most $k$ can be solved in time $\+O(\size{\Gamma} \cdot |G|^{k+1} \cdot \log{|G|})$ on a Turing machine,
	or $\+O(\size{\Gamma} \cdot |G|^{k+1})$ under a RAM model, where $\Gamma$ and $G$ are the input "UC2RPQ" and "graph database", respectively.
\end{restatable}

In practice, "graph databases" tend to be huge and often changing, while queries
are in comparison very small.
This motivates the following question, given some natural $k \geq 1$: 

\begin{center}
    \AP 
    Given a "UC2RPQ" $\Gamma$, is it "equivalent" to a "UC2RPQ" $\Gamma'$ of "tree-width" at most $k$?\\
    That is, does it have ""semantic tree-width"" at most $k$?
\end{center}
This problem is called the ""semantic tree-width $k$ problem"".
Should it be decidable in a constructive way---that is, decidable, and if the answer is positive, we can compute a witnessing $\Gamma'$ from $\Gamma$---, then one could, once and for all,
compute $\Gamma'$ from $\Gamma$ and, whenever one wants to "evaluate" $\Gamma$ on a
database, "evaluate" $\Gamma'$ instead.

We will also study the restriction of these notions to one-way queries: a "UCRPQ" has \AP""one-way semantic tree-width"" at most $k$ if it is equivalent to a "UCRPQ" of "tree-width" at most $k$. The \AP""one-way semantic tree-width $k$ problem"" is the problem of, given a "UCRPQ" $\Gamma$, whether it has "one-way semantic tree-width" at most $k$.

\begin{example}
    \AP\label{ex:CRPQ-tw3-stw2}
    Consider the following "CRPQs",\footnote{In this graphical representation,
	we interpret a labelled graph as the "CRPQ" defined as
	the conjunction of the "atoms" induced by the labelled edges of the graph.
	For instance, $\gamma(\bar x)$ is a conjunction of six "atoms".}
    where $\bar x = (x_0,x_1,y,z)$:\leavevmode
    \begin{center}
        \small
        \begin{tikzcd}[column sep=small, row sep=small]
            &[-.5em] x_0 \ar[dr, "a"] \ar[rr, "c"] \ar[ddr, "a(bb)^+" swap, pos=.6, bend right] & &
            x_1 \ar[dl, "a" swap] \ar[ddl, "ab(bb)^*", pos=.7, bend left]
                &[1.5em] &[-.5em] x_0 \ar[dr, "a"] \ar[rr, "c"] \ar[ddr, "a(bb)^+" swap, pos=.6, bend right] & &
                x_1 \ar[dl, "a" swap] 
                    &[1.5em] &[-.5em] x_0 \ar[dr, "a"] \ar[rr, "c"]  & &
                    x_1 \ar[dl, "a" swap] \ar[ddl, "ab(bb)^*", pos=.6, bend left] \\
            \gamma(\bar x) \defeq & & y \ar[d, "b^+", pos=.35] & 
                & \delta(\bar x) \defeq & & y \ar[d, "b(bb)^*" pos=.35] & 
                    & \delta'(\bar x) \defeq & & y \ar[d, "(bb)^+" swap, pos=.35] & \\
            & & z & 
                & & & z & 
                    & & & z &
        \end{tikzcd}
    \end{center}
    \noindent
    The underlying graph of $\gamma(\bar x)$ being the directed 4-clique, $\gamma 
    (\bar x)$ has "tree-width" 3. We claim that $\gamma(\bar x)$ is equivalent to the "UCRPQ"
    $\delta(\bar x) \lor \delta'(\bar x)$, and hence has "one-way semantic tree-width" at most 2.

    Indeed, given a "graph database" satisfying $\gamma(\bar x)$ via some mapping $\mu$, 
    it suffices to make a case disjunction on whether the number of $b$-labelled "atoms"
    in the path from
    $\mu(y)$ to $\mu(z)$ is even or odd. In the first case, the "atom" $x_0\atom{a(bb)^+} z$ becomes
    redundant since we can deduce the existence of such a path from the conjunction
    $x \atom{a} y \atom{(bb)^+} z$, and hence the "database" "satisfies" $\delta(\bar x)$ via $\mu$.
    Symmetrically, in the second case, the "atom" $x_1 \atom{b(bb)^*} z$ becomes redundant,
    and the "database" "satisfies" $\delta'(\bar x)$ via $\mu$. 
    Thus, $\gamma(\bar x)$
    is "contained", and hence "equivalent" (the other "containment" being trivial), to
    the "UCRPQ" $\delta(\bar x) \lor \delta'(\bar x)$ of \linebreak[5] "tree-width" 2.
\end{example}

\subsection{\AP{}Related Work}
\label{sec:relwork}
On the class "conjunctive queries", the "semantic tree-width $k$ problem" becomes the \conp-complete problem of finding out whether the retraction of a query has "tree-width"
at most $k$. In fact, "CQs" enjoy the effective existence of unique minimal queries \cite[Theorem 12]{DBLP:conf/stoc/ChandraM77}, which happen to also minimize the tree-width. For "CRPQs" and "UC2RPQs", the question is far more challenging, and it has only been solved for the case $k = 1$ 
by Barceló, Romero, and Vardi \cite[Theorem 6.1]{BarceloRV16}; the case $k>1$ was left widely open
\cite[\S 7]{BarceloRV16}.

Furthermore, classes of "CQs" of bounded "semantic tree-width" precisely characterize tractable (and "FPT") "evaluation problem" \cite[Theorem~1.1]{Grohe07}.
This result is on bounded-arity schemas, which was later generalized \cite[Theorem~1]{DBLP:conf/ijcai/ChenGLP20} for characterizing "FPT" "evaluation" on arbitrary schemas---by replacing "semantic tree-width" with semantic ``submodular width'' \cite{Marx13}.

The problem of computing "maximal under-approximations" of "CQs" of a given "tree-width" has been explored in \cite{DBLP:journals/siamcomp/BarceloL014}.
A "maximal under-approximations" of tree-width at most $k$ of a "CQ" $\gamma$ consists of a
"CQ" $\delta_k$ of "tree-width" at most $k$, which under-approximates it,
"ie" $\delta_k$ is "contained" in $\gamma$, and which is maximal, in the sense that for every "CQ" 
$\delta'$, if $\delta'$ has "tree-width" at most $k$ and is "contained" in $\gamma$, then 
$\delta'$ is "contained" in $\delta_k$. 
"Maximal under-approximations" of a given "tree-width" for "CQs" always exist \cite{DBLP:journals/siamcomp/BarceloL014} and thus, a "CQ" is "semantically equivalent"
to a "CQ" of "tree-width" at most $k$ if, and only if, it is equivalent to its maximal under-approximation of "tree-width" at most $k$. Our solution to decide the
"semantic tree-width $k$ problem" for "UC2RPQs" is based on this idea.

While "maximal under-approximations" always exist for "CQs", this is not the case for the dual notion of ``minimal over-approximations''. The problem of when these exist is still unknown to be decidable, aside for some the special cases of acyclic "CQs" and Boolean "CQs" over binary schemas \cite{DBLP:journals/mst/BarceloRZ20}.

\subsection{\AP{}Contributions}
Here we solve both the "semantic tree-width $k$ problem" and "one-way semantic tree-width $k$ problem" for every $k$ with one unifying approach.
\begin{theorem}
    \AP\label{thm:decidability-semtw}
    For each $k \geq 1$, the "semantic tree-width $k$ problem" and the "one-way semantic tree-width $k$ problem" are decidable. Moreover, these problems are in "2ExpSpace" and are "ExpSpace"-hard.
	When $k=1$, the problems are in fact "ExpSpace"-complete.
\end{theorem}
In \Cref{sec:maximal-under-approximations} (\Cref{lem:sem-tw-in-twoexp}),
we prove the upper bound for $k\geq 2$, by relying on the so-called ``"Key Lemma"'', which is our main technical result, and is proven in \Cref{sec:treedec,sec:proof-key-lemma}.
The upper bound for the case $k=1$---which was already proven in \cite{BarceloRV16} for the (two-way) "semantic tree-width $1$ problem"---is shown in \Cref{sec:acyclic-queries} (\Cref{cor:sem-tw-1-pb-exp-c}). The lower bound is shown in \Cref{sec:lowerbound} (\Cref{lemma:lowerbound}).

The "Key Lemma" (\Cref{lemma:bound_size_refinements}) essentially states that
every "UC2RPQ" has a computable ``maximal under approximation'' by a "UC2RPQ" of "tree-width" $k$ and that this approximation is well-behaved with respect to the class of languages used to label the queries under some mild assumptions on it (being ``"closed under sublanguages"''). Let us first
explain this assumption before formalizing the statement above (stated as \Cref{cor:mua-exists-effective}).

For a class $\+L$ of languages, let $\intro*\UCtwoRPQ(\+L)$ denote the class of all "UC2RPQs" whose atoms are all labelled by languages from $\+L$.
\AP For an NFA $\+A$ and two states $q,q'$ thereof, we denote by $\intro*\subaut{\+A}{q}{q'}$ the ""sublanguage"" of $\+A$ recognized  when considering $\set{q}$ as the set of initial states and $\set{q'}$ as the set of final states.
\AP We say that $\+L$ is ""closed under sublanguages"" if
(i) it contains every language of the form $\{a\}$,
where $a \in \A$ is any (positive) letter such that either $a$ or $a^-$ occur in a word of a
language of $\+L$, and (ii) for every language $L \in \+L$ there exists an NFA $\+A_L$ accepting $L$ such that every "sublanguage" $\subaut{\+A_L}{q}{q'}$ distinct from $\emptyset$ and
$\{\varepsilon\}$ belongs to $\+L$.

To the best of our knowledge,
all classes of regular expressions that have been considered in the realm of regular path queries (see, "eg", \cite[\S1]{FigueiraGKMNT20}) are "closed under sublanguages". In particular, this is
the case for the class $\bigl\{ \{a_1+\hdots+a_n\} \mid a_1,\hdots,a_n
\in \A \bigr\} \cup \bigl\{ a^* \mid a \in \A \bigr\}$, which will be our focus of study in \Cref{sec:sre}. Moreover, even if some class $\+L$
is not "closed under sublanguages", such as $\{(aa)^*\}$,
then it is contained in a minimal class "closed under sublanguages"---$\{a, a(aa)^*, (aa)^*\}$ in 
this example.

We can now state the main implication of the "Key Lemma" (whose formal statement requires some extra definitions).
\begin{restatable*}[Existence of the "maximal under-approximation"]{cor}{muaexistseffective}
    \AP\label{cor:mua-exists-effective}
    For each $k \geq 2$, for each class $\+L$ "closed under sublanguages",
    and for each query $\Gamma \in \UCtwoRPQ(\+L)$,
    there exists $\Gamma' \in \UCtwoRPQ(\+L)$ of "tree-width" at most $k$ 
    such that
    $\Gamma' \contained \Gamma$, and for every $\Delta \in \UCtwoRPQ$, if $\Delta$ has
    "tree-width" at most $k$ and $\Delta \contained \Gamma$, then $\Delta \contained \Gamma'$.
    Moreover, $\Gamma'$ is computable from $\Gamma$ in "ExpSpace".
\end{restatable*}

As a consequence of \Cref{cor:mua-exists-effective,prop:crpq-bound-tree-width-upper-bound}, we have that queries of bounded "semantic tree-width" have tractable evaluation.
\begin{restatable*}[FPT evaluation for bounded "semantic tree-width"]{cor}{fptEvalBoundedSemTreeWidth}
	\AP\label{coro:fpt-eval-bounded-semtreewidth}
	For each $k\geq 1$, the "evaluation problem" for "C2RPQs" of "semantic tree-width"
	at most $k$ is fixed-parameter tractable---"FPT"---when parametrized in the size of
	the query. More precisely on input $\langle \Gamma, G \rangle$,
	the algorithm runs in time
	$\+O(f(\size{\Gamma})\cdot |G|^{k+1} \cdot \log{|G|})$ on a Turing machine, where $f$ is a doubly-exponential function---or $\+O(f(\size{\Gamma})\cdot |G|^{k+1})$ under a RAM model.
\end{restatable*}
Note that \cite[Theorem~22]{FGM24} shows that the statement above can be improved to have a single-exponential function $f$.

Moreover, we also show that for any class $\+L$ of regular languages "closed under sublanguages", if $\Gamma \in \UCtwoRPQ(\+L)$ has "semantic tree-width" $k > 1$,
then $\Gamma$ is equivalent to a $\UCtwoRPQ(\+L)$ of "tree-width" at most $k$.
Analogous characterizations hold for $k=1$ and/or "path-width", see \Cref{coro:charact-semantic-treewidth-1,coro:charact-semantic-pathwidth-k}.
\proofappendixtrue
\begin{restatable*}{thm}{closureundersublanguages}
    \AP\label{thm:closure-under-sublanguages}
    \AP Assume that $\+L$ is "closed under sublanguages". 
    For any $k > 1$ and any query $\Gamma \in \UCtwoRPQ(\+L)$, the following are equivalent:
    \begin{enumerate}
        \itemAP[\introstarinrestatable{\itemClosureInfCQ}] $\Gamma$ is "equivalent" to an "infinitary union" of "conjunctive queries"
            of "tree-width" at most $k$; \label{thm:closure-under-sublanguages:1}
        \itemAP[\introstarinrestatable{\itemClosureUCRPQ}] $\Gamma$ has "semantic tree-width" at most $k$; \label{thm:closure-under-sublanguages:2}
        \itemAP[\introstarinrestatable{\itemClosureUCRPQSimple}] $\Gamma$ is "equivalent" to a $\UCtwoRPQ(\+L)$ of "tree-width" at most $k$. \label{thm:closure-under-sublanguages:3}
    \end{enumerate}
\end{restatable*} 
\proofappendixfalse

The implications $\itemClosureUCRPQSimple \Rightarrow \itemClosureUCRPQ \Rightarrow \itemClosureInfCQ$ immediately follow
from the definition of the "semantic tree-width".
On the other hand, the
implications $\itemClosureInfCQ \Rightarrow \itemClosureUCRPQ$ and $\itemClosureUCRPQ \Rightarrow \itemClosureUCRPQSimple$ are surprising,
since they are both trivially false when $k=1$. We defer the proof of this last claim
to \Cref{rk:closure-under-sublanguages-k1} as we first need a few tools to manipulate "CRPQs".

The previous theorem, together with the high complexity of "semantic tree-width $k$ problem",
motivates us to focus on the case of "CRPQs" using some "simple regular expressions" ("SRE") in \Cref{sec:sre}, where we show that the complexity of this problem is much lower.

\begin{restatable*}{thm}{thmSemTwSREpitwo}
    \AP\label{thm:semtw-sre-pitwo}
    For $k\geq 2$, the "semantic tree-width $k$ problem" for {\UCRPQSRE} is in "Pi2".
\end{restatable*}

We then study the problem of $k=1$: at first glance, our proof for $k\geq 2$ of
\Cref{thm:decidability-semtw} does not capture this case, for a technical---yet crucial---reason. In
\Cref{sec:acyclic-queries}, we explain how to adapt our proof to capture it: 
and show the decidability the "semantic tree-width 1 problem"---which was already studied by 
Barceló, Romero and Vardi \cite{BarceloRV16}---and of the "one-way semantic tree-width 1 problem". 

Building on the same idea, we show in \Cref{sec:semantic-path-width} that our results extend to "path-width".
\begin{restatable*}{thm}{decidabilitySemPw}
    \AP\label{thm:decidability-sempw}
    For each $k \geq 1$,
    the "semantic path-width $k$ problems" are decidable. Moreover, they lie in "2ExpSpace"
    and are "ExpSpace"-hard. Moreover, if $k=1$, these problems are in fact "ExpSpace"-complete.
\end{restatable*}
In turn, this leads to an evaluation algorithm with a remarkably low complexity.
\begin{restatable*}{thm}{paraNLEvalBoundedSemPathWidth}
	\AP\label{thm:evaluation-bounded-pathwidth}
	For each $k \geq 1$, the "evaluation problem", restricted to "UC2RPQs" of
	"semantic path-width" at most $k$ is in "para-NL" when parametrized in the size of the query.
	More precisely, the problem, on input $\langle \Gamma, G \rangle$, can be solved in
	non-deterministic space $f(|\Gamma|) + \log(|G|)$, where $f$ is a single exponential
	function.
\end{restatable*}

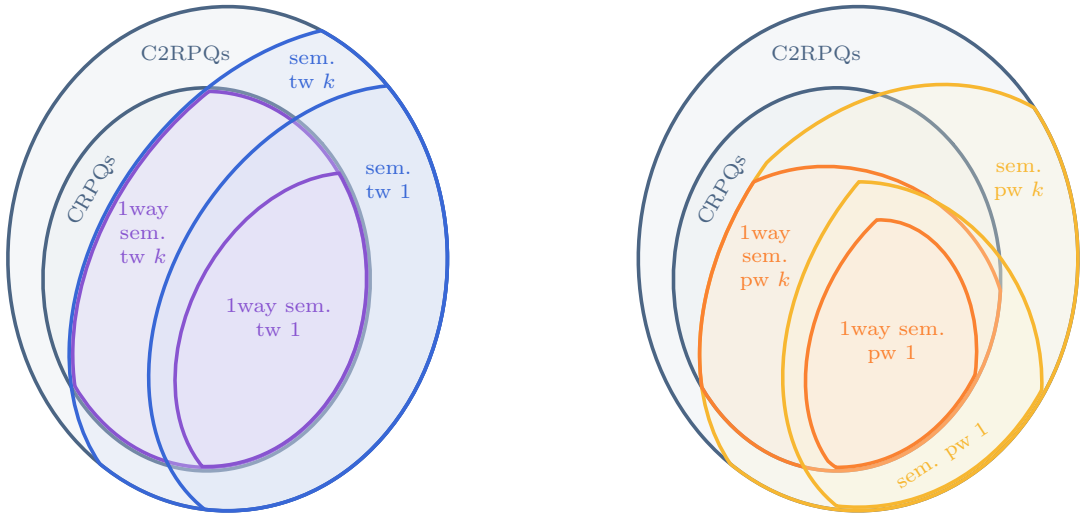
\begin{figure}
	\centering
	\begin{subfigure}{.45\linewidth}
		\centering 
		\scalebox{1.125}{
		\begin{tikzpicture}
			\node at (0,0) {\includegraphics[scale=.8]{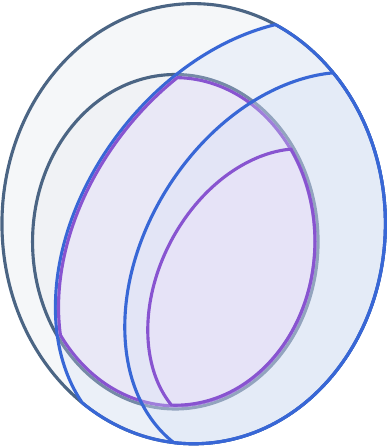}};
			\node[font=\tiny, align=center] at (.6,-.7)
				{\kl[one-way semantic tree-width]{\color{cPurple}1way sem.}\\ \kl[one-way semantic tree-width]{\color{cPurple}tw 1}};
			\node[font=\tiny, align=center] at (-1,.3)
				{\kl[one-way semantic tree-width]{\color{cPurple}1way}\\ \kl[one-way semantic tree-width]{\color{cPurple}sem.}\\ \kl[one-way semantic tree-width]{\color{cPurple}tw $k$}};
			\node[font=\tiny, align=center] at (1.9,.9)
				{\kl[semantic tree-width]{\color{cBlue}sem.}\\ \kl[semantic tree-width]{\color{cBlue}tw 1}};
			\node[font=\tiny, align=center] at (1,2.2)
				{\kl[semantic tree-width]{\color{cBlue}sem.}\\ \kl[semantic tree-width]{\color{cBlue}tw $k$}};
			\node[font=\tiny] at (-1.6,.8) {\kl[CRPQ]{\rotatebox{60}{\color{cDarkGrey}CRPQs}}};
			\node[font=\tiny] at (-.5,2.4) {\kl[C2RPQ]{\color{cDarkGrey}C2RPQs}};
		\end{tikzpicture}
		}
		\caption{
			\AP\label{fig:taxonomy-semantic-tw}
			Semantic classes of "C2RPQs" related to "tree-width".
		}
	\end{subfigure}
	\hfill
	\begin{subfigure}{.45\linewidth}
		\centering 
		\scalebox{1.125}{
		\begin{tikzpicture}
			\node at (0,0) {\includegraphics[scale=.8]{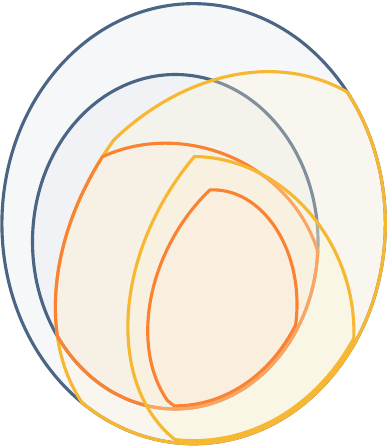}};
			\node[font=\tiny, align=center] at (.4,-1)
				{\kl[one-way semantic path-width]{\color{cOrange}1way sem.}\\ \kl[one-way semantic path-width]{\color{cOrange}pw 1}};
			\node[font=\tiny, align=center] at (-1.1,0)
				{\kl[one-way semantic path-width]{\color{cOrange}1way}\\ \kl[one-way semantic path-width]{\color{cOrange}sem.}\\ \kl[one-way semantic path-width]{\color{cOrange}pw $k$}};
			\node[font=\tiny] at (1,-2.35) {\kl[semantic path-width]{\rotatebox{33}{\color{cYellow}sem. pw 1}}};
			\node[font=\tiny, align=center] at (1.9,.9)
				{\kl[semantic path-width]{\color{cYellow}sem.}\\ \kl[semantic path-width]{\color{cYellow}pw $k$}};
			\node[font=\tiny] at (-1.6,.8) {\kl[CRPQ]{\rotatebox{60}{\color{cDarkGrey}CRPQs}}};
			\node[font=\tiny] at (-.5,2.4) {\kl[C2RPQ]{\color{cDarkGrey}C2RPQs}};
		\end{tikzpicture}
		}
		\caption{
			\AP\label{fig:taxonomy-semantic-pw}
			Semantic classes of "C2RPQs" related to "path-width".
		}
	\end{subfigure}
	\caption{
		\AP\label{fig:taxonomy-semantic}
		Clickable taxonomy of semantic classes studied in this paper, where $k \geq 2$.
	}
\end{figure}
Interestingly, the proof for "tree-width" 1 and "path-width" $k$ ($k \geq 1$)
can be derived from the proof from "tree-width" $k\geq 2$ but necessitates an
additional technical trick which yields different closure properties (or lack thereof).
We show that a "UCRPQ" has "semantic tree-width" at most $k$ if, and only if, it
has "one-way semantic tree-width" at most $k$ whenever $k \geq 2$
(\Cref{coro:collapse-twoway-oneway-semtw}). In other words, if the original query does not
use "two-way navigation", then considering "UC2RPQs" does not help to further minimize the "tree-width". Interestingly, this is false for $k=1$ ("cf" \Cref{rk:closure-under-sublanguages-k1},
also \cite[Proposition 6.4]{BarceloRV16}) and for "path-width", no matter the value of $k\geq 1$ (see \Cref{rk:path-width:oneway-vs-twoway}). Overall, this leads to
the landscape depicted in \Cref{fig:taxonomy-semantic}.

Finally, we conclude in \Cref{sec:discussion}.
We provide a \emph{partial} characterization \emph{à la} Grohe of classes of "UC2RPQs" which admit a tractable evaluation in \Cref{sec:charact-tractability}.
\begin{restatable*}{thm}{thmtractabilityfinred}
    \AP\label{thm:tractability-finred}
    Assuming $\wone \neq $ "FPT", for any recursively enumerable class $\+C$ of "finitely-redundant" Boolean "UC2RPQs", the "evaluation problem" for $\+C$ is "FPT" if, and only if, $\+C$ has bounded "semantic tree-width".
\end{restatable*}
We also discuss open questions, ranging from complexity
questions (\Cref{sec:discussion-complexity}) to extensions of our results to bigger classes
or larger settings (\Cref{sec:discussion-larger-classes,sec:discussion-different-notions}).

\subsection{\AP{}Conference Paper}
\label{sec:conf-paper-diff}
The current article is based on the conference paper \cite{thispaperICDT}. The main results for "tree-width" $k>1$ are essentially the same---though with improved explanations and figures, and we fixed some minor bugs in the proof of the "Key Lemma". Here we also show how to extend our techniques to tackle the "semantic tree-width $1$ problem" (\Cref{sec:acyclic-queries}) and we introduce and study the "semantic path-width $k$ problems" (\Cref{sec:semantic-path-width}). Our very partial lift of Grohe's characterization of "FPT" classes of queries (\Cref{thm:tractability-finred}) is also new.

\section{\AP{}Preliminaries}
\label{sec:prelim}

Before attacking the statement of our "Key Lemma" in \Cref{sec:maximal-under-approximations},
we first give a few elementary definitions on "C2RPQs" in this section.
\AP
We write $\Nat$ to denote $\set{0,1,\dotsc}$ and $\intro*\lBrack i,j \intro*\rBrack$ to denote $\set{ n \in \Nat : i \leq n \leq j}$.
\AP
A ""homomorphism"" $\fun$ from a "C2RPQ" $\gamma(x_1, \dotsc, x_m)$ to a "C2RPQ" $\gamma'(y_1, \dotsc, y_m)$ is a mapping from $\vars(\gamma)$ to $\vars(\gamma')$ such that $\fun(x) \atom{L} \fun(y)$ is an "atom" of $\gamma'$ for every "atom" $x \atom{L} y$ of $\gamma$, and further $\fun(x_i)=y_i$ for every $i$.
Such a "homomorphism" $\fun$ is \AP""strong onto"" if for every "atom" $x' \atom{L} y'$ of $\gamma'$ there is an "atom" $x \atom{L} y$ of $\gamma$ such that $\fun(x)=x'$ and $\fun(y)=y'$.
An example of "homomorphism" is provided in \Cref{fig:basic-hom}.
We write $\gamma \intro*\homto \gamma'$ if there is a "homomorphism" from $\gamma$ to $\gamma'$, and $\gamma \intro*\surj \gamma'$ if there is a "strong onto homomorphism".
In the latter case, we say that $\gamma'$ is a \AP""homomorphic image"" of $\gamma$.
It is easy to see that if $\gamma \homto \gamma'$ then $\gamma' \contained \gamma$, and in the case where $\gamma,\gamma'$ are "CQs" this is an ``if and only if'' \cite[Lemma 13]{DBLP:conf/stoc/ChandraM77}.

\paragraph*{Some intuitions on maximal under-approximations}
Given a "conjunctive query" $\gamma$,
the union of all "conjunctive queries"
that are "contained" in $\gamma$ is "semantically equivalent" to the union
$\bigvee \{ \gamma' \mid \gamma \surj \gamma' \}$. Naturally, this statement borders on the trivial since $\gamma'$ belongs to this union. It becomes interesting when we add a restriction:
given a class $\class$ of "CQs" (to which $\gamma$ may not belong) closed under "subqueries", then $\Gamma' \defeq \bigvee \{ \gamma' \in \class \mid \gamma \surj \gamma' \}$ is the maximal under-approximations
of $\gamma$ by finite unions of "conjunctive queries" of $\class$, in the following sense:
\begin{enumerate}[i.]
	\item (finite) $\Gamma'$ is a finite union of "CQs" of $\class$,
	\item (under-approximation) $\Gamma' \contained \gamma$, and
	\item (maximality) for any finite union $\Delta$ of "CQs" of $\class$, if $\Delta \contained \gamma$, then $\Delta \contained \Gamma'$.
\end{enumerate}

\begin{proof}
Only the last point is non-trivial, and follows from the fact that if
$\Delta \contained \gamma$, then for each $\delta \in \Delta$, $\delta \contained \gamma$,
so there is a "homomorphism" $f\colon \gamma \to \delta$. The image $\delta'$
of $f$ is a "subquery" of $\delta$, and $\+C$ is closed under "subqueries",
so it belongs to $\+C$, and hence to $\Gamma'$. Since there is a trivial homomorphism
from $\delta'$ to $\delta$, we moreover have that $\delta \contained \delta'$.
Hence, for each "CQ" $\delta \in \Delta$, there is a CQ $\delta' \in \Gamma'$ such
that $\delta \contained \delta'$, and hence $\Delta \contained \Gamma'$.
\end{proof}

As a consequence, we deduce that for each $k \geq 1$,
the "maximal under-approximation" of a "CQ" by
a finite union of "CQs" of "tree-width" at most $k$ is computable, and hence
we can effectively decide if some "CQ" is "equivalent" to a query of "tree-width" at
most $k$ by testing the equivalence with this maximal under-approximation.
For more details on approximations of "CQs", see \cite{DBLP:journals/siamcomp/BarceloL014}.
Note that interestingly, changing $\Gamma'$ from
$\bigvee \{ \gamma' \in \class \mid \gamma \surj \gamma' \}$
to $\bigvee \{ \gamma' \in \class \mid \gamma' \contained \gamma \}$
preserves both under-approximation and maximality, but $\Gamma'$ is now an infinite
union of "CQs" of $\+C$.

Unfortunately, these results cannot be straightforwardly extended to "conjunctive regular
path queries" since the previous proof implicitly relied on two points:
\begin{enumerate}
	\item the equivalence between the
	"containment" $\gamma' \contained \gamma$ and the existence of a "homomorphism"
	$\gamma \homto \gamma'$, and
	\item the possibility to restrict $\gamma'$ to its image $\gamma \homto \gamma'$ while 
	obtaining a semantically bigger query.
\end{enumerate}
These two crucial ingredients is what allows us to build a finite set $\Gamma'$ from $\gamma$.
For "CRPQs", the second point still holds, but not the first one.
For instance, the "CQ" $\gamma(x,y) = x \atom{a} z \atom{b} y$ is
contained in (in fact "equivalent" to) the "CRPQ" $\gamma'(x,y) = x \atom{ab} y$,
but there is no "homomorphism" from $\gamma'(x,y)$ to $\gamma(x,y)$.
Our main result shows that to find "maximal under-approximations" of "C2RPQs",
it suffices to take "homomorphic images" of so-called ``"refinements"'' of $\gamma$,
instead of "homomorphic images" of $\gamma$ itself. The next paragraphs are devoted to
introducing "refinements" and tools related to them.

\paragraph*{Equality Atoms}
\AP"C2RPQs" with ""equality atoms"" are queries of the form $\gamma(\bar{x}) = \delta \land I$, 
where $\delta$ is a "C2RPQ" (without equality atoms) and $I$ is a conjunction of "equality atoms" of the form $x=y$. 
Again, we denote by $\vars(\gamma)$ the set of variables appearing in the (equality and non-equality) atoms of $\gamma$. 
We define the binary relation $=_\gamma$ over $\vars(\gamma)$ to be the reflexive-symmetric-transitive closure of the binary relation $\{(x, y) \mid \text{$x=y$ is an "equality atom" in $\gamma$}\}$. 
In other words, we have $x=_\gamma y$ if the equality $x=y$ is forced by the "equality atoms" of $\gamma$. 
Note that every "C2RPQ" with "equality atoms" $\gamma(\bar{x}) = \delta \land I$ is equivalent to a "C2RPQ" without "equality atoms"  $\gamma^{\collapse}$, 
which is obtained from $\gamma$ by collapsing each equivalence class of the relation $=_\gamma$ into a single variable. 
This transformation gives us a \emph{canonical} renaming from $\vars(\gamma)$ to $\vars(\gamma^{\collapse})$. For instance, $\gamma(x,y) \defeq x \atom{K} y \land y \atom{L} z \land x = y$
collapses to $\gamma^{\collapse}(x,x) \defeq x \atom{K} x \land x \atom{L} z$.

\paragraph*{Refinements}
\AP An ""atom $m$-refinement"" of a "C2RPQ" "atom" $\gamma(x,y) = x \atom{L} y$ where $L$ is given by the NFA $\+A_L$ is any "C2RPQ" of the form 
\begin{equation}
    \AP\label{eq:refinement}
    \rho(x,y) = x \atom{L_1} t_1 \atom{L_2} \hdots \atom{L_{n-1}} t_{n-1} \atom{L_n} y
\end{equation}
where $1 \leq n \leq m$, $t_1,\hdots,t_{n-1}$ are fresh (existentially quantified) variables,
and $L_1,\hdots,L_n$ are such that there exists a sequence $(q_0,\dotsc,q_n)$ of states of $\+A_L$
such that $q_0$ is initial, $q_n$ is final, and for each $i$, $L_i$ is either of the form
\begin{enumerate}[(i)]
	\item $\subaut{\+A_L}{q_i}{q_{i+1}}$,
	\item $\{a\}$ if the letter $a\in \A$ belongs to $\subaut{\+A_L}{q_i}{q_{i+1}}$, or 
	\item $\{a^{-}\}$ if $a^{-} \in \A^{-}$ belongs to $\subaut{\+A}{q_i}{q_{i+1}}$.
\end{enumerate}
Additionally, if $\epsilon \in L$, the "equality atom" ``$x = y$'' is also an \reintro{atom $m$-refinement}. Thus, an \reintro{atom $m$-refinement} can be either of the form \eqref{eq:refinement} or ``$x=y$''.
By convention, $t \atom{a^{-}} t'$ is a shorthand for $t' \atom{a} t$. As a consequence,
the underlying graph of an "atom $m$-refinement" of the form \eqref{eq:refinement} is not necessarily a directed path.
By definition, note that
$L_1\cdots L_n \subseteq L$ and hence $\rho \contained \gamma$ for any "atom $m$-refinement" $\rho$ of $\gamma$.
An \AP""atom refinement"" is an "atom $m$-refinement" for some $m$.
An example is provided in \Cref{fig:basic-refinement}.
\begin{figure}
	\centering
	\begin{subfigure}{.45\linewidth}
		\centering
		\includegraphics[scale=.85]{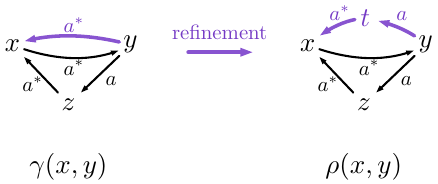}
		\caption{\AP\label{fig:basic-refinement}A "refinement".}
	\end{subfigure}
	\hfill
	\begin{subfigure}{.45\linewidth}
		\centering
		\includegraphics[scale=.85]{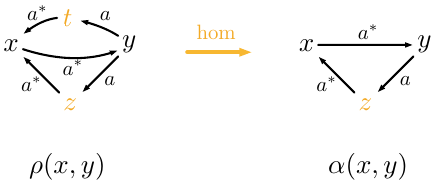}
		\caption{\AP\label{fig:basic-hom}A "strong onto homomorphism".}
	\end{subfigure}
	\caption{
		\AP\label{fig:basic-approx}
		"Refinements" and "homomorphisms" of "C2RPQs".
	}
\end{figure}

\begin{definition}
    \AP\label{def:atom-contraction}
    \AP Given an "atom refinement" $\rho = x \atom{L_1} t_1 \atom{L_2} \hdots \atom{L_{n-1}} t_{n-1} \atom{L_n} y$ of $\gamma = x \atom{L} y$ as in \eqref{eq:refinement}, define
    a ""condensation"" of $\rho$ between $t_i$ and $t_j$, where $0 \leq i,j \leq n$ and $j > i+1$, as any "C2RPQ" of the form:
    \[
        \rho' = x \atom{L_1} t_1 \atom{L_2} \hdots \atom{L_i} \textcolor{cPurple}{t_i \atom{K} t_j} \atom{L_{j+1}} \hdots
        \atom{L_{n-1}} t_{n-1} \atom{L_n} y
    \]
    such that $\textcolor{cPurple}{K = \+A[q_i,q_j]}$.
	\begin{fact}
		\AP\label{fact:refinement-contained}
		Every "condensation" $\rho'$ of $\rho$ is a "refinement" of $\gamma$, and $\rho \contained \rho' \contained \gamma$.
	\end{fact}
    \AP Informally, we will abuse the notation and
    write $\intro*\contract{L_{i}\cdots L_{j}}$ to denote the language $K$---even if this language
    does not only depend on $L_{i}\cdots L_{j}$.
\end{definition}

\begin{example}
    \AP\label{ex:atom-refinement-twoway}
    Let $\gamma(x,y) = x \atom{(aa^-)^*} y$ be a "C2RPQ" "atom", where
    $(aa^-)^*$ is implicitly represented by its minimal automaton.
    Then $\rho(x,y)$ is a "refinement" of "refinement length" seven of $\gamma(x,y)$
    and $\rho'(x,y)$ is a "condensation" of $\rho(x,y)$, where:
    \begin{align*}
        \rho(x,y) & = x \atom{a} t_1 \atom{(a^-a)^*} t_2 \atom{(a^-a)^*} t_3
            \coatom{a} t_4 \atom{(aa^-)^*} t_5 \atom{(aa^-)^*a} t_6 \coatom{a} y, \\
        \rho'(x,y) & = x \atom{a} t_1 \atom{(a^-a)^*} t_2 \atom{(a^-a)^*} t_3
		\coatom{a} t_4 \atom{(aa^-)^*} y. 
    \end{align*}
    On the other hand, $\rho''(x,y) = x \atom{a} t_1 \coatom{a} y$ is not
    a "condensation" of $\rho(x,y)$.
\end{example}

Given a natural number $m$, an \AP""$m$-refinement"" of a "C2RPQ" $\gamma(\bar x) = \bigwedge_{i} x_i \atom{L_i} y_i$ is any query resulting from: 1) replacing every "atom" by one of its "$m$-refinements@@atom", and 2)
should some "$m$-refinements@@atom" have "equality atoms",
collapsing the variables.
\AP A ""refinement"" is an "$m$-refinement" for some $m$.
Note that any "atom $m$-refinements" is, by definition, also an
"atom $m'$-refinements" when $m \leq m'$: as a consequence, in the "refinement" of a "C2RPQ"
the "atom refinements" need not have the same length.
For instance, both $\rho(x,x) = x \atom{c} x$ and $\rho'(x,y) = x \atom{a} t_1 \atom{a} y \coatom{c} y$ are "refinements" of $\gamma(x,y) = x \atom{a^*} y \coatom{c} x$.

For a given "C2RPQ" $\gamma$, let $\AP\intro*\Refin[\leq m](\gamma)$ be the set of all "$m$-refinements" of $\gamma$, and $\reintro*\Refin(\gamma)$ be the set of all its "refinements".
Given a "refinement" $\rho(\bar x)$ of $\gamma(\bar x)$,
its ""refinement length"" is the least natural number
$m$ such that $\rho(\bar x) \in \Refin[\leq m](\gamma)$.
Note that if the automaton representing a language $L$ has more than one final state, for instance the minimal automaton for $L = a^+ + b^+$,
then $x \atom{L} y$ is not a "refinement" of itself.
However, it will always be "equivalent" to a union of refinements: in
this example, $x \atom{a^+ + b^+} y$ is "equivalent" to the union of
$x \atom{a^+} y$ and $x \atom{b^+} y$, which are both "refinements"
of the original "C2RPQ".

\paragraph*{Expansions}
Remember that a "C2RPQ" whose languages are
of the form $\set{a}$ or $\set{a^-}$ for $a \in \A$ is in effect a "CQ".
The \AP""expansions"" of a "C2RPQ" $\gamma$ is the set $\intro*\Exp(\gamma)$ of all "CQs" which are "refinements" of $\gamma$.
In other words, an "expansion" of $\gamma$ is any "CQ" obtained from $\gamma$
by replacing each "atom" $x \atom{L} y$ by a path $x \atom{w} y$ for some
word $w \in L$.
For instance, $\xi(x,y) = x \atom{a} t_1 \coatom{a} t_2 \atom{a} t_3 \coatom{a} y$
is an "expansion" of $\rho(x,y) = x \atom{(aa^-)^*} y$.

Any "C2RPQ" is equivalent to the infinitary union of its "expansions". In light of this, the semantics for "UC2RPQ" can be rephrased as follows. 
Given a "UC2RPQ" $\Gamma(\bar x)$ and a graph database $G$, 
the "evaluation" of $\Gamma(\bar x)$ over $G$, denoted by $\Gamma(G)$, is the set of tuples 
$\bar{v}$ of nodes for which there is $\anexpansion \in \Exp(\Gamma)$ such that there is a "homomorphism" $\anexpansion \homto G$ that sends $\bar x$ onto $\bar v$.  
Similarly, "containment" of "UC2RPQs" can also be characterized in terms of expansions.

\begin{proposition}[Folklore, see e.g. {\cite[Proposition 3.2]{Florescu:CRPQ}} or
    {\cite[Theorem 2]{CGLV00}}]
    \AP\label{prop:cont-char-exp-st} 
    Let $\Gamma_1$ and $\Gamma_2$ be "UC2RPQs". Then the following are equivalent
    \begin{itemize}
        \item $\Gamma_1 \contained \Gamma_2$;
        \item for every $\anexpansion_1\in \Exp(\Gamma_1)$, $\anexpansion_1 \contained \Gamma_2$;
        \item for every $\anexpansion_1\in \Exp(\Gamma_1)$ there is $\anexpansion_2\in \Exp(\Gamma_2)$ such that $\anexpansion_2\homto \anexpansion_1$. 
    \end{itemize}
\end{proposition}

Note that since an "expansion" of $\gamma$ is also a "refinement" of $\gamma$, it also
holds that $\gamma$ is "semantically equivalent" to the infinitary union of its "refinements".

Our approach to proving \Cref{thm:decidability-semtw,thm:closure-under-sublanguages}
and the "Key Lemma" heavily rely on "refinements". One crucial property
that these objects satisfy is that they preserve "tree-width" $k$, unless $k=1$,
as illustrated in \Cref{fig:tree-decompositon-expansion}.

\begin{restatable}{fact}{refinementtw}
    \AP\label{fact:refinement-tw}
    Let $k \geq 2$ and let $\gamma$ be a "C2RPQ" of "tree-width" at most $k$.
    Then any "refinement" of $\gamma$ has "tree-width" at most $k$.
\end{restatable}

\begin{figure}
    \centering
	\begin{subfigure}[t]{.4207\textwidth}
		\centering
		\includegraphics*[width=.98\textwidth]{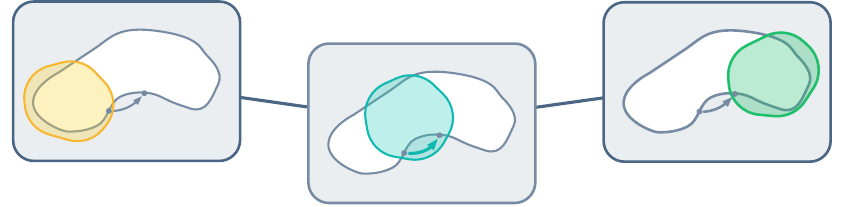}
		\caption{
			\AP\label{subfig:tree-decompositon-before-expansion}
			A multigraph together with a "tree decomposition" of "width" $k$.
		}
	\end{subfigure}
	\hfill
	\begin{subfigure}[t]{.5593\textwidth}
		\centering
		\includegraphics*[width=.98\textwidth]{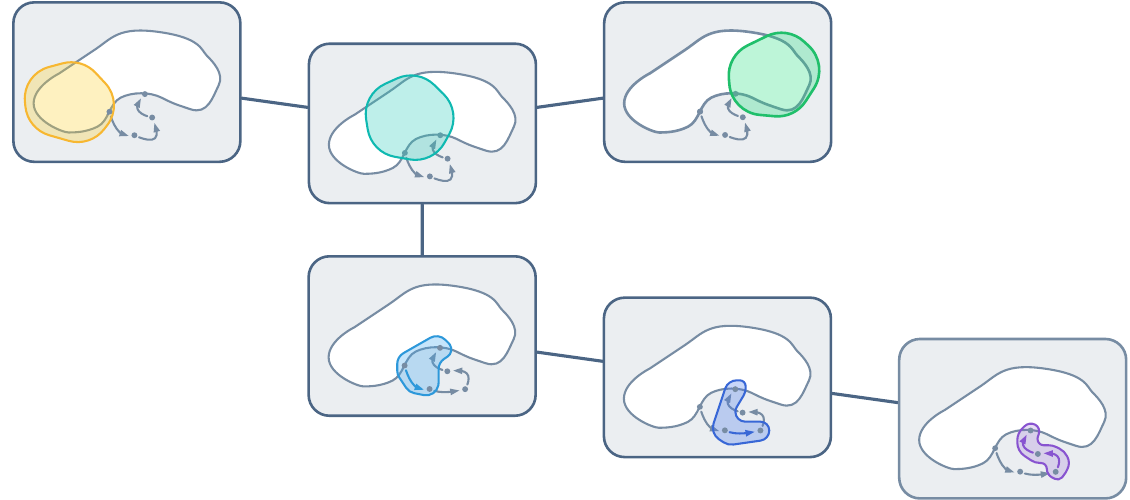}
		\caption{
			\AP\label{subfig:tree-decompositon-after-expansion}
			A "refinement" of the multigraph of \Cref{subfig:tree-decompositon-before-expansion} together with a "tree decomposition" of "width" $\max(k,2)$.
		}
	\end{subfigure}
    \caption{\AP\label{fig:tree-decompositon-expansion} "Refinements" and "expansions" preserve "tree-width" at
	most $k\geq 2$.
    }
\end{figure}
\begin{proof}
    The underlying graph of a "refinement" of $\gamma$ is obtained from the underlying graph
    of $\gamma$ by either contracting some edges (when dealing with "equality atoms"),
	or by replacing
    a single edge by a path of edges (where the non-extremal nodes are new nodes).

	This first operation preserves "tree-width" at most $k$ (even if $k = 1$),
    see "eg" \cite[Lemma 16]{bodlaender1998arboretum}. The second operation
    preserves "tree-width" at most $k$, assuming $k > 1$: if a graph $G'$
    is obtained from a graph $G$ by replacing an edge $x_0 \atom{} x_n$
    by a path $x_0 \atom{} x_1 \atom{} \cdots \atom{} x_n$, then
    from a "tree decomposition" of $G$ it suffices to pick a "bag" containing
    both $x_0$ and $x_n$, and add a branch to the tree, rooted at this "bag",
    and containing bags with nodes
    \begin{align*}
        \{x_0,x_1,x_n\},\,
        \{x_1,x_2,x_n\},\,
		\hdots,\,
        \{x_i,x_{i+1},x_n\},\,
		\hdots,\,
        \{x_{n-2},x_{n-1},x_n\},
    \end{align*}
	as depicted in \Cref{fig:tree-decompositon-expansion}.
    All bags contain exactly three nodes, so we obtain "tree decomposition" of
    $G'$ whose "width" is the maximum between 2 and the "width" of the original
    "tree decomposition" of $G$.
\end{proof}

For $k=1$, the property fails: for instance the "CRPQ" $\gamma(x) = x \atom{a^*} x$
has "tree-width" at most 1 (in fact it has "tree-width" 0), but its "refinement"
$\rho(x) = x \atom{a^*} t_1 \atom{a^*} t_2 \atom{a^*} x$ has "tree-width" 2.

\paragraph*{Fine tree decompositions}
    \AP
	For technical reasons---the proof of \Cref{lemma:shape-decomposition}---, we will use a restrictive class of "tree decompositions" which we call ``"fine@fine tree decomposition"''\footnote{This is similar---but orthogonal---to the classical notion of
	``nice tree decomposition'', see "eg" \cite[Definition 13.1.4, page 149]{kloks1994treewidth}.}. A \AP""fine tree decomposition"" is a "tree decomposition" $(T, \bagmap)$ in which:
	\begin{equation}
		\AP\label{eq:fine-tree-dec}
		\parbox{.65\linewidth}{every non-root "bag" can be obtained from its parent "bag" by
		either adding or removing a non-empty set of vertices.}%
	\end{equation}
	In the context of a "fine tree decomposition" of "width" $k$, a \AP""full bag"" is any "bag" of size $k+1$.
	
    A "C2RPQ" has "tree-width" $k$ if and only if it has a "fine tree decomposition" of "width" at most $k$. Indeed, from a "tree decomposition", it suffices to:
	\begin{enumerate}
		\item first merge every consecutive pair of "bags" that contain exactly the same variables;
		\item between every pair of "bags"
		that does not satisfy \eqref{eq:fine-tree-dec}, add a "bag" whose set of vertices
		correspond to the intersection of the two adjacent "bags".
	\end{enumerate}

\section{\AP{}Maximal Under-Approximations}
\label{sec:maximal-under-approximations}

In this section, we state our key technical result, \Cref{lemma:bound_size_refinements}, which we will refer to as the ``"Key Lemma"''.
Essentially, we follow the same structure as \Cref{thm:closure-under-sublanguages}:
given a "C2RPQ" $\gamma$ and a natural number $k>1$, we start by
considering its "maximal under-approximation" by "infinitary unions"
of "conjunctive queries" of "tree-width" $k$ (\Cref{def:max-under-approx}),
and then show that this query can in fact be expressed as a
"UC2RPQ" of "tree-width" $k$ whose "atoms" contain "sublanguages" of those in $\gamma$ ("Key Lemma"~\ref{lemma:bound_size_refinements}).

\AP For the first definitions of this section, let us fix any class $\class$ of "C2RPQs"---we will later apply these results to the class $\Tw$ of "C2RPQs" of "tree-width" at most $k$.

\begin{definition}[Maximal under-approximation]
    \AP\label{def:max-under-approx}
    \AP Let $\gamma$ be a "C2RPQ". The \AP""maximal under-approximation"" of $\gamma$ by "infinitary unions" of $\class$-queries is $\intro*\MUA{\gamma}{\class} \defeq
    \{ \alpha \in \class \mid \alpha \contained\gamma \}$.
\end{definition}
For intuition, we refer the reader back to paragraph ``Some intuitions on maximal under-approximations'' at the beginning of \Cref{sec:prelim}.

\begin{remark}
    \AP\label{rk:uctworpq}
    Observe that $\MUA{\gamma}{\class}$ is an "infinitary union" of $\class$-queries,
    that $\MUA{\gamma}{\class} \contained \gamma$, and that for every "infinitary union" of
    $\mathcal{C}$-queries $\Delta$, if $\Delta \contained \gamma$, then $\Delta \contained \MUA{\gamma}{\class}$ ("ie", it is the unique maximal under-approximation up to "semantical equivalence").
    Similarly, the "maximal under-approximation" of a "UC2RPQ" is simply the union of the "maximal under-approximations" of the "C2RPQs" thereof.
\end{remark}

Unfortunately, the fact that a query $\alpha$ is part of this union, namely $\alpha \in \MUA{\gamma}{\class}$, does not yield any useful information on the \emph{shape} of $\alpha$---we merely know that $\alpha \contained \gamma$. We thus introduce another "infinitary union" of $\class$-queries  of a restricted shape, namely $\MUAHom{\gamma}{\class} \subseteq \MUA{\gamma}{\class}$, in which queries $\alpha \in \MUAHom{\gamma}{\class}$ come together with a witness of their "containment" in $\gamma$.

\begin{definition}
    \AP The "maximal under-approximation" of $\gamma$ by "infinitary unions" of homomorphic\-ally-smaller $\class$-queries is\phantomintro\MUAHom
    \begin{align}
        \reintro*\MUAHom{\gamma}{\class} \defeq
        \{
            \alpha \in \class
            \mid
            \exists \rho \in \Refin(\gamma),\, \exists \fun\colon \rho \surj \alpha
        \}.
        \AP\label{eq:MUAHom}
    \end{align}
\end{definition}

For a basic example of "approximation" (with no constraint on $\class$),
we refer the reader to \Cref{fig:basic-approx}.
The resulting query $\alpha(x,y)$ is the "homomorphic image" of a "refinement" of
$\gamma(x,y)$. Hence, $\alpha(x,y) \in \MUAHom{\gamma}{\class}$ if $\class$ is, for instance, the class of all "C2RPQs"---or more generally, if $\class$ contains $\alpha(x,y)$.

\begin{example}[{\Cref{ex:CRPQ-tw3-stw2}, cont’d}]
    \AP\label{ex:CRPQ-tw3-stw2-contd}
    Both $\delta(\bar x)$ and $\delta'(\bar x)$ are "semantically equivalent" 
    to queries in $\MUAHom{\gamma(\bar x)}{\Tw[2]}$.
    Indeed, starting from $\gamma(\bar x)$,
    we can "refine"
	\[
		x_0 \atom{a(bb)^+} z
		\quad\text{into}\quad
		x_1 \atom{a} t \atom{(bb)^+} z.
	\]
    Denote by $\rho(\bar x)$ the query obtained:
    \begin{center}
        \small
        \begin{tikzcd}[column sep=small, row sep=small]
            &[-.5em] x_0 \ar[dr, "a"] \ar[rr, "c"] \ar[d, "a" left] & &
            x_1 \ar[dl, "a" swap] \ar[ddl, "ab(bb)^*", pos=.7, bend left]
                &[1.5em] &[-.5em] x_0 \ar[dr, "a"] \ar[rr, "c"] & &
                x_1 \ar[dl, "a" swap] \ar[ddl, "ab(bb)^*", pos=.6, bend left]
                    &[1.5em] &[-.5em] x_0 \ar[dr, "a"] \ar[rr, "c"]  & &
                    x_1 \ar[dl, "a" swap] \ar[ddl, "ab(bb)^*", pos=.6, bend left] \\
            \rho(\bar x) \defeq & t \ar[dr, "(bb)^+" below left, pos=.2, bend right=14] & y \ar[d, "b^+", pos=.35] & 
                & \delta'_{\textnormal{app}}(\bar x) \defeq & & y \ar[d, "b^+" pos=.35, bend left=15] \ar[d, "(bb)^+" swap, bend right=15, pos=.35] & 
                    & \delta'(\bar x) = & & y \ar[d, "(bb)^+" swap, pos=.35] & \\
            & & z & 
                & & & z & 
                    & & & z &
        \end{tikzcd}
    \end{center}
    Then merge variables $t$ and $y$: this new query $\delta'_{\textnormal{app}}(\bar x)$
    is "equivalent" to $\delta'(\bar x)$. Moreover, since
	$\delta'_{\textnormal{app}}(\bar x)$ has "tree-width" at most 2 and was obtained as a 
	"homomorphic image" of a "refinement" of $\gamma(\bar x)$, we have that
	$\delta'_{\textnormal{app}}(\bar x) \in \MUAHom{\gamma(\bar x)}{\Tw[2]}$.
	A similar argument applies to $\delta$, by "refining" the "atom" between $x_1$ and $z$ instead.
	\qedhere
\end{example}

\AP Clearly, $\MUAHom{\gamma}{\class}$---whose queries are informally called 
""approximations""---is included, and thus semantically "contained", in $\MUA{\gamma}{\class}$, 
since $\rho \contained \gamma$ and $\alpha \contained \rho$ in \eqref{eq:MUAHom}.
In fact, under some assumptions on $\class$,
the converse "containment" also holds. 

\begin{observation}
    \AP\label{obs:equivalence_under_approx_homomorphism}
    If $\class$ is closed under "expansions" and "subqueries",
	then for any "C2RPQ" $\gamma$, we have $\MUA{\gamma}{\class} \semequiv \MUAHom{\gamma}{\class}$.
\end{observation}
\begin{proof}
	Since $\MUA{\gamma}{\class} \supseteq \MUAHom{\gamma}{\class}$,
	it suffices to show that $\MUA{\gamma}{\class} \contained \MUAHom{\gamma}{\class}$.
	Pick $\alpha \in \MUA{\gamma}{\class}$. Let $\xi$ be an "expansion" of $\alpha$. 
	Since $\alpha \contained \gamma$, there exists by \Cref{prop:cont-char-exp-st}
	an "expansion" $\xi_\gamma$ of $\gamma$ such that $\xi_\gamma \homto \xi$. 
	Consider the restriction $\xi'$ of $\xi$ to its "homomorphic image".
	Since $\alpha \in \class$ and $\class$ is closed both under "expansions" and "subqueries",
	$\xi' \in \class$. Since moreover, by construction, $\xi'$ is the ("strong onto@strong onto homomorphism") "homomorphic image"
	of an "expansion" (hence "refinement") of $\gamma$, then $\xi' \in \MUAHom{\gamma}{\class}$.
	Hence, we have shown that for every "expansion" of $\MUA{\gamma}{\class}$,
	there is an "expansion" of $\MUAHom{\gamma}{\class}$ with a "strong onto homomorphism" from the
	latter to the former, which concludes the proof by \Cref{prop:cont-char-exp-st}.
\end{proof}

Note that
in the definition of $\MUAHom{\gamma}{\class}$ we work with "strong onto homomorphisms":
changing the definition to have any "homomorphism" would yield a slightly bigger but "semantically equivalent" class of queries---though having untamed shapes.

Observe then, by \Cref{fact:refinement-tw}, that the class $\Tw$ of all "C2RPQs" of "tree-width" at 
most $k$ is closed under "refinements" and hence under "expansions", provided that $k$ is greater 
or equal to 2. Moreover, $\Tw$ is always closed under "subqueries" for each $k$.

\begin{corollary}
    \AP\label{coro:equivalence_under_approx_homomorphism_twk}
    For $k \geq 2$, for all "C2RPQ" $\gamma$,
    $\MUA{\gamma}{\Tw} \semequiv \MUAHom{\gamma}{\Tw}$.
\end{corollary}

\begin{example}[counterexample for $k=1$]
    \AP\label{ex:counterex-tw1}
    Consider the following query:
    \begin{center}
    \begin{tikzcd}[row sep=0.2em]
        &[-2em] &[-2em] z \ar[<-, ddr, "b"] &[-2em] & & &[-2em] \\
        \gamma(x) \;\defeq & & & \\
		& x \ar[<-, uur, "c"] & & y. \ar[<-, ll, "a"]
    \end{tikzcd}
    \end{center}
    We claim that $\MUA{\gamma}{\Tw[1]} \not\contained \MUAHom{\gamma}{\Tw[1]}$.
	First, we claim that $\gamma \in \Exp(\MUA{\gamma}{\Tw[1]})$ since
	$\gamma$ is an "expansion" of $\delta(x) = x \atom{abc} x$, which clearly belongs
	to $\MUA{\gamma}{\Tw[1]}$.
	Then, observe that $\gamma(x)$ has a single refinement: itself!
	It follows that $\MUAHom{\gamma}{\Tw[1]}$ is finite, and consists precisely of all
	"homomorphic images" of $\gamma(x)$ of "tree-width" at most 1, which are:
	\[
		\alpha_1(w) \defeq
		\begin{tikzcd}[row sep=0.2em, ampersand replacement=\&]
			w \ar["a", loop left] \rar["b" below, bend right=20] \& z \lar["c" above, bend right=20]
		\end{tikzcd},\qquad
		\alpha_2(w) \defeq
		\begin{tikzcd}[row sep=0.2em, ampersand replacement=\&]
			w \ar["c", loop left] \rar["a" below, bend right=20] \& y \lar["b" above, bend right=20]
		\end{tikzcd}
	\]
	\[
		\alpha_3(x) \defeq
		\begin{tikzcd}[row sep=0.2em, ampersand replacement=\&]
			x \rar["a" below, bend right=20] \& w \lar["c" above, bend right=20] \ar["b", loop right] 
		\end{tikzcd},\,\qquad
		\alpha_4(w) \defeq
		\begin{tikzcd}[row sep=0.2em, ampersand replacement=\&]
			w \ar["a", loop left] \ar["b", loop below] \ar["c", loop right] \& \hphantom{z}
		\end{tikzcd}
	\]
	which correspond to the case when the following variable are merged: $\{x,y\}$,
	$\{x,z\}$, $\{y,z\}$ and $\{x,y,z\}$, respectively. Note that
	all of these queries are "CQs", from which it follows that
	every "expansion" of a query in $\MUAHom{\gamma}{\Tw[1]}$ is one of the $\alpha_i$,
	and has a self-loop. In particular, such an "expansion" cannot have a "homomorphism"
	to $\gamma$. Hence, we showed that there is an "expansion" of $\MUA{\gamma}{\Tw[1]}$
	"st" no "expansion" of $\MUAHom{\gamma}{\Tw[1]}$ can be "homomorphically mapped@homomorphism"
	to it. Hence, by \Cref{prop:cont-char-exp-st},
	$\MUA{\gamma}{\Tw[1]} \not\contained \MUAHom{\gamma}{\Tw[1]}$. \qedhere
\end{example}

In general, by definition,
$\MUAHom{\gamma}{\Tw}$ is an "infinitary union" of "C2RPQs". Our main technical result shows that,
in fact, $\MUAHom{\gamma}{\Tw}$ is always equivalent to a \emph{finite} union of "C2RPQs". This is done by bounding the length of the "refinements" occurring in the definition of $\MUAHom{\gamma}{\Tw}$.
For any $m \geq 1$, we define:
\AP
\[
    \intro*\MUAHomBounded{\gamma}{\class}{\leq m} \defeq 
    \{ 
        \alpha \in \class
        \mid
        \exists \rho \in \Refin[\leq m](\gamma),\, \exists \fun\colon \rho \surj \alpha
    \}.
\]
\newcommand{\lbound}[2]{\Theta(\nbatoms[2]{#2}\cdot ({#1}+1)^{\nbatoms{#2}})}
\begin{lemma}[""Key Lemma""]
    \AP\label{lemma:bound_size_refinements}
    \AP For $k \geq 2$ and "C2RPQ" $\gamma$, we have
    $\MUAHom{\gamma}{\Tw} \semequiv \MUAHomBounded{\gamma}{\Tw}{\leq\l}$, where
    $\intro*\l = \lbound{k}{\gamma}$.
\end{lemma}
By construction, $\MUA{\gamma}{\Tw}$ is the maximal under-approximation of $\gamma$ by
"infinitary unions" of "C2RPQs" of "tree-width" at most $k$. Using the equivalence above and
\Cref{coro:equivalence_under_approx_homomorphism_twk}, it follows that
it is also the maximal under-approximation of $\gamma$ by
a "UC2RPQ" of "tree-width" at most $k$.
\muaexistseffective
\begin{proof}
	The algorithm to compute $\Gamma'$ is straightforward:
	it enumerates $\l$-refinements, enumerates its "homomorphic images",
	and keeps the result only if it has "tree-width" at most $k$---which can
	be done in linear time using Bodlaender's algorithm \cite[Theorem 1.1]{bodlaender1996treewidth}.
\end{proof}

Using the "Key Lemma" as a black box---which will be proven in \Cref{sec:proof-key-lemma}---, we can now give a proof of the upper bound of \Cref{thm:decidability-semtw} for all
cases $k\geq 2$---the case $k=1$ will be the object of \Cref{sec:acyclic-queries}.
\begin{restatable}[Upper bound for \Cref{thm:decidability-semtw} for $k\geq 2$]{lem}{lemsemtwintwoexp}
    \AP\label{lem:sem-tw-in-twoexp}
    For $k \geq 2$, the "semantic tree-width $k$ problem" for "UC2RPQ" is in "2ExpSpace".
\end{restatable}
Note that $\MUAHomBounded{\gamma}{\Tw}{\leq\l}$ has double-exponential size in $\size{\gamma}$,
so testing equivalence of $\gamma$ with this "UC2RPQ" yields an algorithm in triple-exponential 
space in $\size{\gamma}$ since "(U)C2RPQ" equivalence is "ExpSpace" \cite[Theorem 5]{CGLV00}
---see also \cite[§ after Theorem 4.8]{Florescu:CRPQ} for a similar result on "CRPQs" without 
inverses but with an infinite alphabet. To get a better upper bound, we first need
the following proposition:
\begin{proposition}
	\AP\label{prop:bound-containment-pb}
    The "containment problem" $\Gamma \contained \Delta$ between two "UC2RPQs" can be solved in non-deterministic space $\+O(\size{\Gamma} + \size{\Delta}^{c \cdot {n_\Delta}})$, for some constant $c$,
	and where $n_\Delta$ is the maximal number of "atoms" of a disjunct of $\Delta$, namely $n_\Delta = \max{\{\nbatoms{\delta} \mid \delta \in \Delta\}}$. 
\end{proposition}

\begin{proof}
	The proposition follows from the following claim.
    \begin{claim}[implicit in \cite{Figueira20}]
        The "containment problem" $\Gamma \contained \Delta$ between two "UC2RPQs" can be solved in non-deterministic space $\+O(\size{\Gamma} + {\size{\Delta}}^{c \cdot \bw(\Delta)})$, where $\bw(\Delta)$ is the "bridge-width" of $\Delta$ and $c$ is a constant.
    \end{claim}
    \AP
    In the statement above, a ""bridge"" of a "C2RPQ" is a minimal set of "atoms" whose removal increases the number of connected components of the query, and the ""bridge-width"" of a "C2RPQ" is the maximum size of a "bridge" therein. The "bridge-width" of a union of "C2RPQs" is the maximum "bridge-width" among the "C2RPQs" it contains. In particular, the maximal number
	of "atoms" of a disjunct is an upper bound for "bridge-width".
\end{proof}

We provide an alternative upper bound in
\Cref{prop:bound-containment-pb-alt} (\Cref{apdx-sec:alternative-upper-bound-containment}),
which also yields a "2ExpSpace" upper bound for \Cref{lem:sem-tw-in-twoexp}.

\begin{proof}[Proof of~\Cref{lem:sem-tw-in-twoexp}]
To test whether a query $\Gamma$ is of "semantic tree-width" $k$, it suffices to test the "containment" $\Gamma \contained \Gamma'$, where $\Gamma'$ is the "maximal under-approximation" $\bigcup_{\gamma \in \Gamma}\MUAHomBounded{\gamma}{\Tw}{\leq\l}$ given by \Cref{cor:mua-exists-effective}: a double-exponential union of single-exponential sized "C2RPQs". Thus, by the bound of \Cref{prop:bound-containment-pb} (and "Savitch's Theorem"), we obtain a double-exponential space upper bound.
\end{proof}

Moreover, from the "equivalences"
$\MUA{\gamma}{\Tw} \semequiv \MUAHom{\gamma}{\Tw}$ and
$\MUAHom{\gamma}{\Tw} \semequiv \MUAHomBounded{\gamma}{\Tw}{\leq\l}$ of
\Cref{coro:equivalence_under_approx_homomorphism_twk,lemma:bound_size_refinements}, we can derive
new characterizations for queries of bounded "semantic tree-width".
\closureundersublanguages

\begin{proof}[Proof of \Cref{thm:closure-under-sublanguages}]
    The implications
	$\itemClosureUCRPQSimple \Rightarrow \itemClosureUCRPQ \Rightarrow \itemClosureInfCQ$
	are straightforward:
    they follow directly from \Cref{fact:refinement-tw}.
    For $\itemClosureInfCQ \Rightarrow \itemClosureUCRPQSimple$, note that $\itemClosureInfCQ$ implies that
    $\Gamma \semequiv \MUA{\Gamma}{\Tw}$, and by
    \Cref{lemma:bound_size_refinements},
    $\MUA{\Gamma}{\Tw} \semequiv \Delta \defeq
    \bigvee_{\gamma \in \Gamma} \MUAHomBounded{\gamma}{\Tw}{\leq \l}$,
    so $\Gamma$ is "equivalent" to the latter.
    Since queries of $\Delta$ are obtained as "homomorphic" images
    of "refinements" of $\Gamma$, all of which are labelled by "sublanguages" of
    $\+L$, and since $\+L$ is "closed under sublanguages", it follows that
    $\Gamma$ is "equivalent" to a $\UCtwoRPQ(\+L)$ of "tree-width" $k$.
\end{proof}

\begin{remark}
    \AP\label{rk:closure-under-sublanguages-k1}
	The statement of \Cref{thm:closure-under-sublanguages} does not hold for $k=1$.

    $\itemClosureUCRPQ \not\Rightarrow \itemClosureInfCQ$ when $k=1$:
    consider the "CRPQ" $\gamma(x,y) = x \atom{a^*} y \land y \atom{b} x$ of "tree-width" 1,
    and hence of "semantic tree-width" $1$, and observe that it is not equivalent to any
    "infinitary union" of "conjunctive queries" of "tree-width" $1$---this can be proven
    by considering, for example, the "expansion"
    $x \atom{a} z \atom{a} y \land y \atom{b} x$ of $\gamma(x,y)$
    and applying \Cref{prop:cont-char-exp-st}.

    $\itemClosureUCRPQSimple \not\Rightarrow \itemClosureUCRPQ$ when $k=1$:
	by \cite[Proposition 6.4]{BarceloRV16} the "CRPQ"
    of "semantic tree-width" 1
    $\gamma(x) \defeq x \coatom{a} z \atom{a} y \land x \atom{b} y \semequiv x \atom{ba^- a} x$
    is not equivalent to any "UCRPQ" of "tree-width" 1. Hence, the implication is false when $\+L$ is the class of regular languages over $\Aext$
    that do not use any letter of the form $a^-$. \qed
\end{remark}

See \Cref{coro:charact-semantic-treewidth-1} for a similar (but different) characterization
of queries of "semantic tree-width" at most 1.
As an immediate corollary of \Cref{thm:closure-under-sublanguages}, by taking
$\+L$ to be the class of all regular languages over $\A$, we obtain the following result.
\begin{corollary}
	\AP\label{coro:collapse-twoway-oneway-semtw}
	Let $k \geq 2$. A "UCRPQ" has "semantic tree-width" at most $k$
	if and only if it has "one-way semantic tree-width" at most $k$.
\end{corollary}

Lastly, using \Cref{cor:mua-exists-effective} as a black box, we can obtain an "FPT" algorithm
for the "evaluation problem".
\fptEvalBoundedSemTreeWidth
\begin{proof}
	First, compute from $\Gamma$ its "maximal under-approximation" $\Gamma'$ using
	\Cref{cor:mua-exists-effective}
	in single-exponential space, and hence double-exponential time.
	Then, evaluate $G$ on $\Gamma'$ using \Cref{prop:crpq-bound-tree-width-upper-bound}.
\end{proof}
This improves the database-dependency from the previously best (and first) known upper bound, which was $\+O(f'(\size{\Gamma})\cdot |G|^{2k+1})$ for a single-exponential $f'$ \cite[Theorem IV.11 \& Lemma IV.13]{DBLP:conf/lics/0001BV17}.
We discuss open questions related to this
in \Cref{sec:charact-tractability}.

We are left with the proof of the "Key Lemma". But before doing so, we will need to introduce in the next \Cref{sec:treedec} some basic notions that we will need in the proof, which is deferred to \Cref{sec:proof-key-lemma}.

\section{\AP{}Intermezzo: Tagged Tree Decompositions}
\label{sec:treedec}

In this section we introduce some technical tools necessary for the proof of
the "Key Lemma". Remember that its statement deals with
\[
	\MUAHomBounded{\gamma}{\Tw}{\leq m} \defeq 
    \{ 
        \alpha \in \Tw
        \mid
        \exists \rho \in \Refin[\leq m](\gamma),\, \exists \fun\colon \rho \surj \alpha
    \},
\]
and consequently its proof needs to manipulate "homomorphisms"
from "refinements" onto "C2RPQs" of "tree-width" $\leq k$.
The proof will ``massage'' the "homomorphism" $f$ and queries $\alpha,\rho$ in order to reduce the size of $m$, while preserving (a) the existence of a "homomorphism" between the two queries,
(b) the "tree-width" of the right-hand side, (c) the fact that the left-hand side is a "refinement",
and (d) some semantic properties of the queries. Our construction will be guided by
the "tree decomposition" of $\alpha$, and more importantly by how $\rho$ is mapped onto such decomposition.

\begin{figure}[tbp]
	\centering
	\begin{subfigure}[c]{.25\linewidth}
		\centering
		\includegraphics{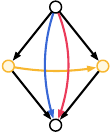}
		\caption{A query $\gamma$ of "tree-width" 3.\\~}
	\end{subfigure}
	\hfill
	\begin{subfigure}[c]{.35\linewidth}
		\centering
		\includegraphics{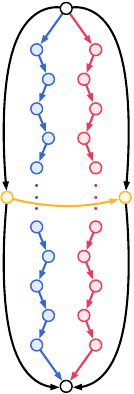}
		\caption{\AP\label{fig:trio-refinement}A "refinement" $\rho$ of $\gamma$.\\~\\~}
	\end{subfigure}
	\hfill
	\begin{subfigure}[c]{.35\linewidth}
		\centering
		\includegraphics{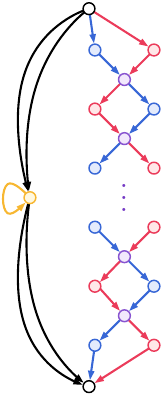}
		\caption{A homomorphic image $\alpha$ of $\rho$ of "tree-width" 2. See \Cref{fig:trio-tree-dec} for a "tree decomposition" of $\alpha$ (ignoring the dashed blue lines).}
	\end{subfigure}
	\caption{
		\AP\label{fig:trio}
		An example of a "homomorphism" $f\colon \rho \surj \alpha$.
		The "strong onto homomorphism" $f$ is implicitly defined: it sends the two yellow
		vertices of $\rho$ on the unique yellow vertex of $\alpha$, and identifies some
		blue and red vertices of $\rho$---thus creating purple vertices in $\alpha$.
	}
\end{figure}

\begin{definition}
    \AP Let $\fun\colon \rho \homto \alpha$ be a "homomorphism" between two "C2RPQs".
    A ""tagged tree decomposition""
    of $\fun$ is a triple $(T, \bagmap, \intro*\tagmap)$ where
    $(T, \bagmap)$ is a "tree decomposition" of $\alpha$,
    and $\tagmap$ is a mapping $\tagmap\colon \atoms{\rho} \to \vertex{T}$, called \AP""tagging"",
    such that for each "atom" $e = x \atom{\lambda} y \in \atoms{\rho}$, we have that $\bagmap(\tagmap(e))$ "contains@@tw" both $\fun(x)$ and
    $\fun(y)$.
\end{definition}

\begin{figure}[tbp]
	\centering
	\includegraphics[width=\linewidth]{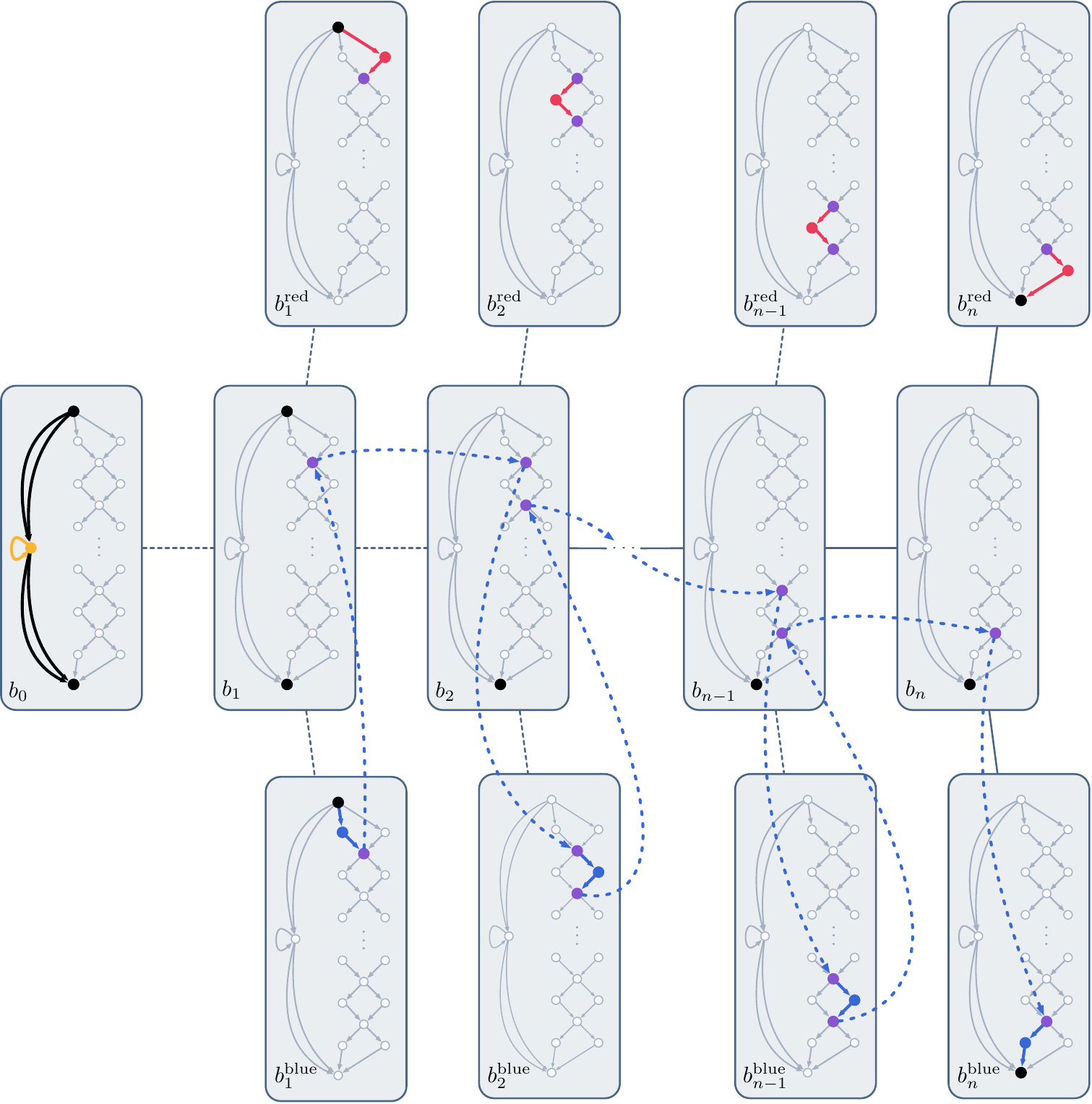}
	\caption{
		\AP\label{fig:trio-tree-dec}
		A "fine tagged tree decomposition" of $\alpha$ (see \Cref{fig:trio}) of "width" 2.
		(Recall that some bags are omitted for the sake of readability. These bags
		are there to make the decomposition "fine".)
	}
\end{figure}

In other words, $\tagmap$ gives,
for each "atom" of $\rho$, a witnessing "bag" that "contains@@tw" it, in the sense that
it contains the image by $\fun$ of the atom's source and target.
By definition, given a "tree decomposition" $(T, \bagmap)$ of $\alpha$ and a "homomorphism"
$\fun\colon \rho \surj \alpha$, there is always one way (usually many) of extending $(T, \bagmap)$
into a "tagged tree decomposition" of $\fun$.

We provide an example of "homomorphism" $\fun\colon \rho \surj \alpha$ in \Cref{fig:trio}. Note 
that in this example, $\rho$ is defined as the "refinement" of a query, and $\fun$ is "strong 
onto"---for now this is innocuous, but we will always work under these
assumptions in \Cref{sec:proof-key-lemma}. In \Cref{fig:trio-tree-dec}, we give a "tagged tree 
decomposition" of this "homomorphism". Each "bag" is given a name, written in the bottom left 
corner. The "tagging" is represented as follows: if an "atom" is "tagged" in a "bag", then it
is drawn as a solid bold arrow in this "bag". Note that by definition, a given "atom" is "tagged" in
exactly one "bag". For now, blue dashed arrow between bags can be ignored---they will illustrate \Cref{def:path-induced}.

\begin{fact}
    \AP\label{fact:restriction_tagged_treedec}
    Let $(T, \bagmap, \tagmap)$ be a "tagged tree decomposition" of some "strong onto homomorphism"
    $\fun\colon \rho \surj \alpha$.
    Let $T'$ be the smallest connected subset of $T$ containing the image of $\tagmap$.
    Then $(T', \bagmap\vert_{T'}, \tagmap)$ is still a "tagged tree decomposition" of $\fun$,
    whose "width" is at most the "width" of $(T, \bagmap, \tagmap)$.
\end{fact}

In the following paragraphs, we extend the notion of "tagging" to paths.
We illustrate this notion in \Cref{fig:trio-tree-dec}, where we describe
the path induced by
the blue path of \Cref{fig:trio-refinement}---which starts at the top-most vertex,
follows the blue atoms, and reaches the bottom-most vertex.
\AP Informally, in the context of a "tagged tree decomposition" $(T, \bagmap, \tagmap)$ of  $\fun\colon \rho \homto \alpha$, given a path $\pi$ of $\rho$, say $x_0 \atom{\lambda_1} x_1 \atom{\lambda_2} \cdots \atom{\lambda_n} x_n$, the "path induced" by $\pi$, denoted by $\tagmappath{\pi}$, is informally defined as the following ``path'' in
$T\times \alpha$, seen as a sequence of pairs of "bags" and variables
from $V(T) \times \vars(\alpha)$:
\begin{itemize}
    \item it starts with the "bag" $\tagmap(x_0 \atom{\lambda_1} x_1)$ of $T$ and the variable $\fun(x_0)$ of $\alpha$; in \Cref{fig:trio-tree-dec}, this corresponds to bag $b^{\text{blue}}_1$;
	\item it then goes to $\langle \tagmap(x_0 \atom{\lambda_1} x_1), \fun(x_1) \rangle$;
    \item it then follows the shortest path in $T$ (unique, since it is a tree) that goes to the "bag" $\tagmap(x_1 \atom{\lambda_2} x_2)$, while staying in $\fun(x_1)$ in $\alpha$---in \Cref{fig:trio-tree-dec}, this bag is the same as before, namely $b^{\text{blue}}_1$, so we do nothing;
    \item then, it goes to $\langle \tagmap(x_1 \atom{\lambda_2} x_2),\fun(x_2) \rangle$ in a single step;
    \item it then follows the shortest path in $T$ (unique, since it is a tree) that goes to the "bag" $\tagmap(x_2 \atom{\lambda_\lambda} x_3)$, while staying in $\fun(x_2)$ in $\alpha$---in our running example, we go from $b^{\text{blue}}_i$
	to $b_i$, and then to $b_{i+1}$ before reaching $b^{\text{blue}}_{i+1}$;
    \item it continues in the same way for all other atoms of the path, ending up with the  "bag" $\tagmap(x_{n-1} \atom{\lambda_n} x_n)$ and the variable $\fun(x_{n})$ of $\alpha$.
\end{itemize}
By construction, note that the constructed sequence $(b_i, z_i)_{i}$,
also denoted by $({b_i \choose z_i})_{i}$, is
such that $z_i \in \bagmap(b_i)$. Moreover, the values taken
by the sequence $(z_i)_{i}$ are $(f(x_j))_{0 \leq j \leq n}$, in the same order
but potentially with repetitions.
Graphically, this sequence corresponds to a path in the "tagged tree decomposition", where one can
not only move along the "bags", but also along the variables they contain.
In our example, the "path induced" by the blue path of \Cref{fig:trio-refinement}
corresponds in \Cref{fig:trio-tree-dec} to the blue path consisting of both solid and dashed edges.
Moreover, note that a single atom $x_0 \atom{\lambda} x_1$ of $\rho$ "induces the path":
\begin{equation}
	\AP\label{eq:path-induced-atom}
	\Bigl\langle
	{\textstyle
		{ \tagmap(x_0 \atom{\lambda} x_1) \choose x_0 },\;
		{ \tagmap(x_0 \atom{\lambda} x_1) \choose x_1 }
	}
	\Bigr\rangle.
\end{equation}

\begin{definition}[Path induced in a tagged tree decomposition---formal definition]
	\AP\label{def:path-induced}
	Given a "homomorphism" $f\colon \rho \homto \alpha$ and a "tagged tree decomposition"
	$(T, \bagmap, \tagmap)$ of $f$,
    the \AP""link"" from an "atom" $A = x \atom{\lambda} y$ to an "atom" $B= y \atom{\lambda'} z$ of $\rho$ is the unique (possibly empty) sequence
    $
            \textstyle{
                {b_{1} \choose \fun(y)}, \hdots,
                {b_{n} \choose \fun(y)},
            }
    $
    where $\tagmap(A), b_{1}, \dotsc, b_n, \tagmap(B)$ is the unique simple path from $\tagmap(A)$ to $\tagmap(B)$ in $T$.

    \AP The \intro{path induced} by a path
	$
		\pi = 
		x_0 \atom{\lambda_1} x_1 \atom{\lambda_2} \cdots \atom{\lambda_n} x_n
	$
	of $\rho$ is the unique sequence
    \[
		\intro*\tagmappath{\pi} \defeq
            \textstyle{
                {b_0 \choose \fun(x_0)} {b_0 \choose \fun(x_1)} \,
                L_1 \,
                {b_1 \choose \fun(x_1)} {b_1 \choose \fun(x_2)} \,
                L_2 
                \dotsb 
				{b_{n-2} \choose \fun(x_{n-2})}
                L_{n-1} \,
                {b_{n-1} \choose \fun(x_{n-1})} {b_{n-1} \choose \fun(x_n)}
            }
    \]
    where $b_i = \tagmap(x_i \atom{\lambda_{i+1}} x_{i+1})$ and $L_i$ is the "link" from $x_{i-1} \atom{\lambda_i} x_i$ to $x_{i} \atom{\lambda_{i+1}} x_{i+1}$, for every $i$.
\end{definition}

\AP Moreover, given a "bag" $b$ of $T$ and a variable $z$ of $\alpha$, we say that
$\tagmappath{\pi}$ ""leaves"" $b$ at $z$ when ${b \choose z}$ belongs to
$\tagmappath{\pi}$, and this is either the last element of the sequence $\tagmappath{\pi}$,
or the next element of the sequence has a "bag" distinct from $b$.

For example, in \Cref{fig:trio-tree-dec},
$\tagmappath{\pi}$ "leaves" $b^{\text{blue}}_1$ at the first purple vertex.
Similarly, it "leaves" $b_1$ and $b_2$ at this same vertex. Moreover,
it also "leaves" $b_2$ at the second purple vertex.

We say that an "induced path" is \AP""cyclic@@path"" if it contains two positions $i, j$
such that $i+2 \leq j$ and $b_i = b_{j}$.
We say that it is \reintro(path){acyclic} otherwise,
meaning that if we visit a "bag" for the first time, we can visit it again at
most once, in which case it must be precisely at the next time step.
For instance, the "path induced" by the blue "atom refinement" in \Cref{fig:trio-tree-dec}
is "cyclic@@path". However, the "path induced" by a single
"atom"---see \eqref{eq:path-induced-atom}--- is always "acyclic@@path".

\begin{fact}
	\AP\label{fact:acyclic-decomposition-leave-forever}
	If an "induced path" $\tagmappath{\pi}$ is "acyclic@@path", for any "bag" $b$,
	there is at most one variable $z$ of $\alpha$ such that $\tagmappath{\pi}$
	"leaves" $b$ at $z$.
\end{fact}

\AP Lastly, we define a ""fine tagged tree decomposition"" of $\fun\colon \rho \homto \alpha$
to be a "tagged tree decomposition" of $\fun$ that is also a "fine tree decomposition" of $\alpha$.
We abuse the notation and talk about the "fine tagged tree decomposition" of a "C2RPQ" $\gamma$
to talk about the "fine tagged tree decomposition" of the identity "homomorphism"
$\id\colon \gamma \surj \gamma$.

\AP One of the key properties of "fine tagged tree decompositions" is that in any of its ""non-branching paths""---"ie" paths in $T$ whose non-extremal "bags" have degree exactly 2---, at least half of the "bags" are "non-full", "ie" they
contain at most $k$ variables\footnote{Recall 
that in a decomposition of "width" $k$, "bags" are allowed to contain at most $k+1$ variables.}.
Such "bags" will prove useful in the next section because of the following property.

\begin{figure}
    \centering
    \includegraphics[width=\linewidth]{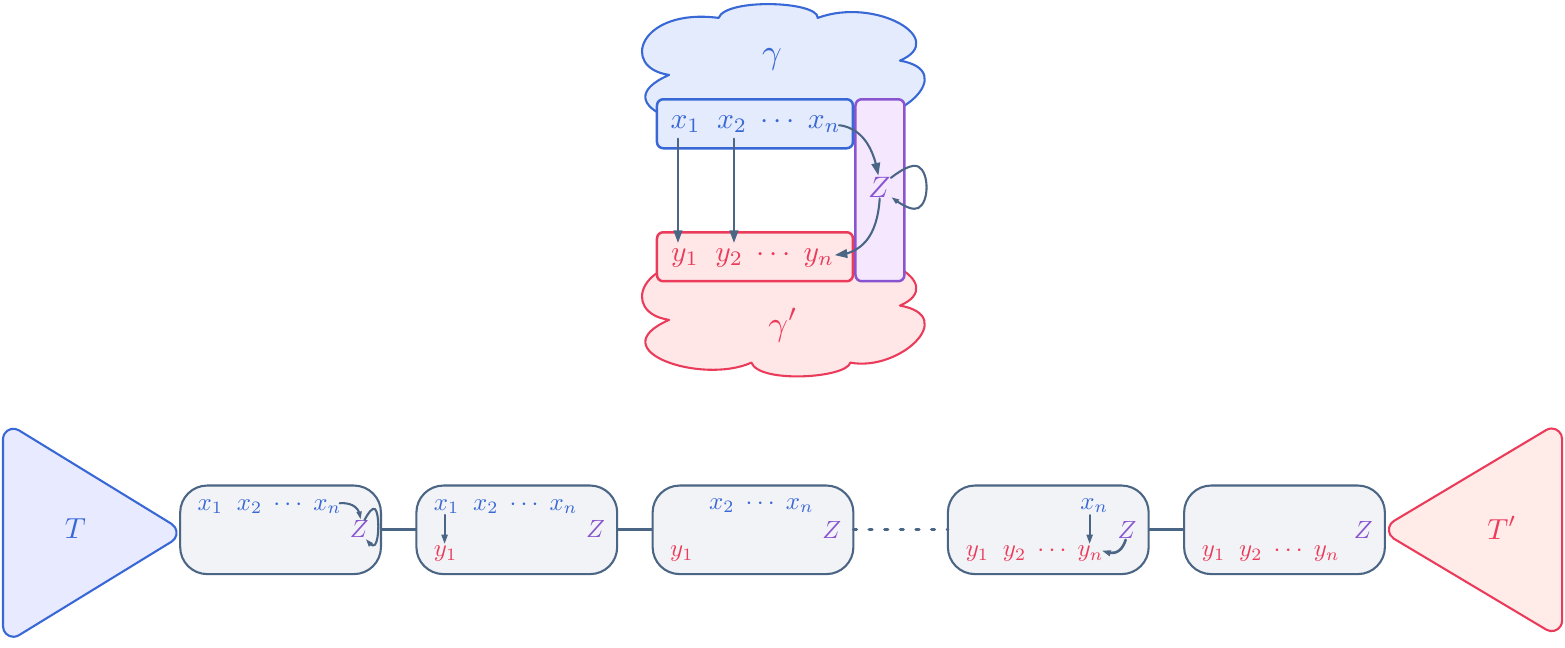}
    \caption{%
        \AP\label{fig:connecting-tree-decompositions} The query
        $\gamma \land \gamma' \land \delta$ (top)
        and one of its "fine tagged tree decomposition"
		of "width" at most $k$ (bottom).
    }
\end{figure}

\begin{proposition}
    \AP\label{prop:connecting-tree-decompositions}
    Let $\gamma,\gamma'$ be "C2RPQs",
    and $(T,\bagmap,\tagmap)$---resp.\ $(T',\bagmap',\tagmap')$---be
	a "fine tagged tree decomposition" of width $k$ of $\gamma$---resp.\ of $\gamma'$.
    Let $b, b'$ be leaves of $T$ and $T'$ respectively, such that $b$ and $b'$ are "non-full bags"
	of the same cardinality,
	and let $Z = \bagmap(b) \cap \bagmap'(b')$. In particular, we have
    \[
        \bagmap(b) = \{x_1,\hdots,x_n\}\cup Z
        \;\text{ and }\;
        \bagmap'(b') = \{y_1,\hdots,y_n\}\cup Z,
    \]
	from some variables "st" the $x_i$'s are disjoint from the $y_i$'s.
	Assume moreover that $\vars(\gamma) \cap \vars(\gamma') \subseteq Z$.
    Then, for any conjunction $\delta$ of "atoms" the form:
	\begin{itemize}
		\item $x_i \atom{L} y_i$ for some $i \in \lBrack 1, n \rBrack$,
		\item $x_i \atom{L} z$ for some $i \in \lBrack 1, n \rBrack$ and $z\in Z$,
		\item $z \atom{L} y_i$ for some $i \in \lBrack 1, n \rBrack$ and $z\in Z$,
		\item $z \atom{L} z'$ for some $z, z'\in Z$,
	\end{itemize}
	the query $\gamma \land \gamma' \land \delta$
	has a "fine tagged tree decomposition"
    of "width" $k$ in which the length of the longest "non-branching path" is smaller
    than the sum of the longest "non-branching paths" of $T$ and of $T'$, plus $2n$.
\end{proposition}

The proof of \Cref{prop:connecting-tree-decompositions} is elementary and illustrated in \Cref{fig:connecting-tree-decompositions}.

\begin{proof}
   	We connect $T$ with $T'$ with $2n \leq 2k$ bags:
    start from $b_0 \defeq b$, which contains $\{x_1,\hdots,x_n\}\cup Z$. Then create
    the following bags:
    \begin{itemize}
        \item $\bagmap(b_1) \defeq \{x_1,x_2,\hdots,x_n\}\cup\{y_1\}\cup Z
            = \bagmap(b_0) \cup \{y_1\}$, 
        \item $\bagmap(b_2) \defeq \{x_2,\hdots,x_n\}\cup\{y_1\}\cup Z
            = \bagmap(b_1) \smallsetminus \{x_1\}$,
        \item $\bagmap(b_{2i-1}) \defeq \{x_i,\hdots,x_n\}\cup\{y_1,\hdots,y_i\}\cup Z
            = \bagmap(b_{2i-2}) \cup \{y_i\}$,
        \item $\bagmap(b_{2i}) \defeq \{x_{i+1},\hdots,x_n\}\cup\{y_1,\hdots,y_i\}\cup Z
            = \bagmap(b_{2i-1}) \smallsetminus \{x_i\}$
    \end{itemize}
    for $1 \leq i \leq n$, and observe that $\bagmap(b_{2n}) = \bagmap'(b')$.
    Then, "tag" every atom of $\delta$ in the first bag of
	$\langle b_1, \hdots, b_{2n-1} \rangle$ containing both variables of the "atom".
	Such a bag always exists:
	\begin{itemize}
		\item an "atom" of the form $x_i \atom{L} z$ is "tagged" in $b_0$;
		\item an "atom" of the form $z \atom{L} z'$ is "tagged" in $b_0$;
		\item an "atom" of the form $x_i \atom{L} y_i$ is "tagged" in $b_{2i-1}$;
		\item an "atom" of the form $z \atom{L} y_i$ is "tagged" in $b_{2i-1}$.
	\end{itemize}
	Observe that the decomposition obtained is indeed a "fine tagged tree decomposition":
	in particular, it satisfies that for each variable $t$, the set of all $b \in T$
	containing $t$ is a connected subtree of $T$, thanks to the assumption that
	$\vars(\gamma) \cap \vars(\gamma') \subseteq Z$.
\end{proof} 

\section{\AP{}Key Lemma: Maximal Under Approximations are Semantically Finite}
\label{sec:proof-key-lemma}

\AP We can now start to describe the constructions used to prove the "Key Lemma"~\ref{lemma:bound_size_refinements}. Given a fixed "C2RPQ" $\gamma$ and a fixed $k \geq 1$,
we call a ""trio"" any triple $(\alpha,\rho,\fun)$ such that 
$\alpha \in \Tw$,
$\rho \in \Refin(\gamma)$ and 
$\fun$ is a "strong onto homomorphism" from $\rho$ to $\alpha$. 
For clarity, we will denote such a "trio" by simply ``$\fun\colon \rho \surj \alpha$''. 
Using this terminology, in order to prove \Cref{lemma:bound_size_refinements}, it is sufficient (and necessary) to show that:
\begin{center}
	for every "trio" $\fun\colon \rho \surj \alpha$,
	there exists another "trio" $\fun'\colon \rho' \surj \alpha'$\\
	"st" $\alpha \contained \alpha'$ and $\rho'\in \Refin[\leq \l](\gamma)$.
\end{center}

\begin{remark}
	\AP\label{rk:key-lemma-tw1}
	Note that this section does not use the fact that $k \geq 2$. In particular,
	\Cref{lemma:bound_size_refinements} holds for $k=1$. However,
	\Cref{coro:equivalence_under_approx_homomorphism_twk} does not apply,
	and $\MUA{\gamma}{\Tw[1]}$ (which we are interested in)
	is not equivalent to $\MUAHom{\gamma}{\Tw[1]}$ (which is shown to be computable by \Cref{lemma:bound_size_refinements}). We discuss this case in further details in \Cref{sec:acyclic-queries}.
\end{remark}

\subsection{\AP{}Local Acyclicity}
Our first construction, which will ultimately allow us to bound the size of "atom refinements",
shows that we can assume "wlog" that they "induce" "acyclic paths" in a "fine tagged 
tree decomposition" of $\fun$.

\begin{restatable}{lem}{locallyacyclictreedec}
    \AP\label{lemma:locally_acyclic_treedec}
    \AP For any "trio" $\fun\colon \rho \surj \alpha$, there exists a "trio" $\fun'\colon \rho' \surj \alpha'$ and a "fine tagged tree decomposition" $(T', \bagmap', \tagmap')$ of "width" at most 
    $k$ of $\fun'$ such that
    $\alpha \contained \alpha'$, $\nbatoms{\rho'} \leq \nbatoms{\rho}$ and every "atom refinement" of
    $\rho'$ "induces" an "acyclic path" in the tree $T'$, in which case
    we say that $(T', \bagmap', \tagmap')$ is \AP\introinrestatable{locally acyclic}
	"wrt" $f'$.
\end{restatable}

\begin{figure}[tbp]
	\centering
	\includegraphics[width=\linewidth]{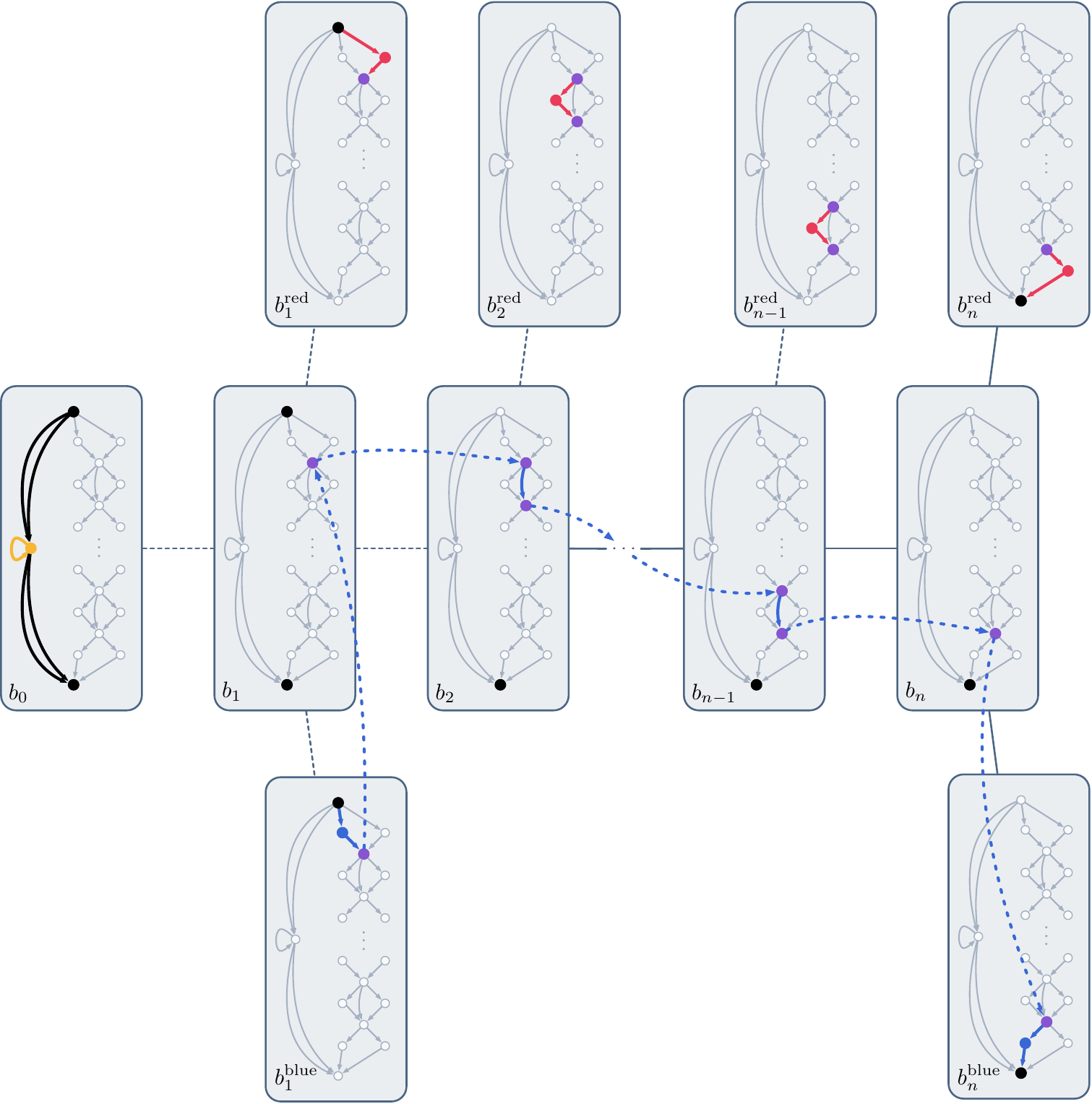}
	\caption{
		\AP\label{fig:path-induced}
        Changing the blue "atom refinement" of $\gamma$ (see \Cref{fig:trio-tree-dec})
		so that it induces an "acyclic path".
	}
\end{figure}

Note that the fact that $f'$ is a "trio" implies in particular that
$\rho'$ is a "refinement" of $\gamma$.
The construction behind \Cref{lemma:locally_acyclic_treedec} is
illustrated in \Cref{fig:path-induced}.
\begin{nota}
	\AP\label{nota:nice-tree-dec}
	When two "bags" are linked by a dashed
	edge (as in \Cref{fig:trio-tree-dec,fig:path-induced}), it means
	that there is another "bag" in between them, which is there
	to ensure the fact that the "decomposition@tree decomposition" is "fine".
	The vertices contained in this extra bag are exactly the intersection of
	the vertices contained by its two neighbours, and no atom is tagged inside.
\end{nota}

\begin{proof}[Informal proof of \Cref{lemma:locally_acyclic_treedec}]
    Start with a "trio"
    $\fun\colon \rho \surj \alpha$, and let $(T, \bagmap, \tagmap)$ be a "fine tagged tree decomposition" of
    $\fun$. Consider an "atom refinement"
    $\pi \defeq z_0 \atom{L_1} z_1 \atom{L_2} \cdots \atom{L_n} z_n$
    in $\rho$ of some atom $x \atom{L} y$ (with $z_0 \defeq x$ and $z_n \defeq y$),
    and assume that it "induces"
    a "cyclic path" in $T$--- see "eg" \Cref{fig:trio-tree-dec}. It means that some variables $z_i$ and $z_j$ are mapped by $\fun$
    to the same bag of $T$, somewhere along the path "induced" by $\pi$. It suffices then to "condense" $\rho$
    by replacing the "atoms" $z_i \atom{L_{i+1}} \cdots \atom{L_j} z_j$ by a single
    "atom" $z_i \atom{\contract{L_{i+1} \cdots L_j}} z_j$.
    We thus obtain a new "refinement" $\rho'$ of $\gamma$.
    Then define $\alpha'$ be simply adding an atom $\fun(z_i) \atom{L_{i+1} \cdots L_j} \fun(z_j)$.
    The definitions of $\fun'$ and $(T',\bagmap',\tagmap')$ are then straightforward---potentially,
    $\alpha'$ should be restricted to the image of $\fun'\colon \rho' \homto \alpha'$ so that $\fun'$ is still "strong onto" by using \Cref{fact:restriction_tagged_treedec}. Crucially, $\alpha \contained \alpha'$, and $\alpha'$ still has "tree-width" at most $k$ 
    since we picked $\fun(z_i)$ and $\fun(z_j)$ so that they belonged to the same bag of $T$: therefore, adding an "atom" between them is innocuous.
	We then iterate this construction for every "atom refinement".
\end{proof}

\Cref{fig:path-induced} shows the "fine tagged tree decomposition" $(T',\bagmap',\tagmap')$
obtained by applying the previous construction to the "decomposition@fine tagged tree decomposition"
$(T, \bagmap, \tagmap)$ of \Cref{fig:trio-tree-dec} for the blue "atom refinement",
followed by applying \Cref{fact:restriction_tagged_treedec}.
In \Cref{fig:trio-tree-dec}, the "induced path" was "leaving" the "bag" $b_2$ both
at the first and at the second purple vertex. This leads in \Cref{fig:path-induced}
to a new atom between these vertices. The same phenomenon happens to "bags" $b_3, \hdots, b_{n-1}$.
Lastly, note that because the "atoms" tagged in "bags"
$b^{\text{blue}}_2, \hdots, b^{\text{blue}}_{n-1}$ are not in the image of $f'$,
these bags were removed by \Cref{fact:restriction_tagged_treedec}.

\begin{proof}[Formal proof of \Cref{lemma:locally_acyclic_treedec}]
    Let $\pi$ be an "atom refinement" in $\rho$ that induce a "cyclic path" in $T$,
    say
    \[
      \pi = z_0 \atom{L_1} z_1 \atom{L_2} \dotsb \atom{L_{n-1}} z_{n-1} \atom{L_n} z_n.
    \]

    In order to build the "trio" $\fun'\colon \rho' \surj \alpha'$
    and a "fine tagged tree decomposition" $T'$ of $\fun'$ of width at most $k$,
    we will mainly use the fact that if two vertices $(u,v)$ of some graph $G$
    belong to the same bag of a "tree decomposition" $(T, \bagmap)$ of $G$, then
    $(T, \bagmap)$ is still also a "tree decomposition" of the graph obtained by adding an edge from $u$ to $v$.
    
    By definition, the "induced path" $\tagmappath{\pi} = \bigl({b_i \choose x_i}\bigr)_i$
	is of the form
    \[
        \tagmappath{\pi} =
        \textstyle{
        \Bigl\langle
            {b_{i_0} \choose \fun(z_0)},
			{b_{i_0+1} \choose \fun(z_1)},
			\hdots,
			{b_{i_1} \choose \fun(z_1)},
			{b_{i_1+1} \choose \fun(z_2)},
			\hdots,
			{b_{i_{n-1}} \choose \fun(z_{n-1})},
			{b_{i_{n-1}+1} \choose \fun(z_n)}
        \Bigr\rangle
        },
    \]
	where $i_0 \defeq 0$, and for each $l$, $b_{i_l} = b_{i_l+1}$.
    Since it is not "acyclic@@path", there exists $(j, j')$ such that
    $j + 2 \leq  j'$ and $b_j = b_{j'}$.
    Let $n(j)$ (resp.\ $n(j')$) denote the unique index such that
    $i_{n(j)-1} < j \leq i_{n(j)}$ (resp.\ $i_{n(j')-1} < j' \leq i_{n(j')}$).
    In particular, we have $\fun(z_{n(j)}) \in \bagmap(b_j)$
    and $\fun(z_{n(j')}) \in \bagmap(b_{j'})$.
    We claim that $n(j) < n(j')$---otherwise, we would have twice the same
	bag in a "link", which would contradict the fact that it is a simple path in $T$.
  
    We can then define
    \[
      \pi' \defeq t_0 \atom{L_1} \dotsb \atom{L_{n(j)}} t_{n(j)} \atom{K} t_{n(j')} \atom{L_{n(j')+1}} \dotsb \atom{L_n} t_n,
    \]
    where $K \defeq \contract{L_{n(j)+1}\cdots L_{n(j')}}$ (see \Cref{def:atom-contraction})
    and let $\rho'$ be the query obtained from $\rho$ by replacing $\pi$ with $\pi'$.
    Then, define $\alpha'$ to be the query obtained from $\alpha$ by adding an
    atom
    $\fun(z_{n(j)}) \atom{K} \fun(z_{n(j')})$,
    so that by construction, we have $\alpha \contained \alpha'$,
    that $\rho' \in \Refin(\gamma)$ with $\nbatoms{\rho'} \leq \nbatoms{\rho}$ and 
    $\fun$ induces a "homomorphism" $\fun'\colon \rho' \homto \alpha'$.
  
    We must then build a "tagged tree decomposition" $(T', \bagmap', \tagmap')$
    of $\fun'$.
	First, we restrict $\alpha'$ to be the image of
    $\fun'\colon \rho' \homto \alpha'$, in order to obtain a "strong onto homomorphism".
	Then, starting from the "tagged tree decomposition" $(T, \bagmap, \tagmap)$ of $\fun$, restrict $\tagmap$ to the atoms $\atoms{\rho'} \smallsetminus{\{z_{n(j)} \atom{K} z_{n(j')}\}} \subseteq \atoms{\rho}$,
    and tag the atom $z_{n(j)} \atom{K} z_{n(j')}$ to the bag 
    $b_j = b_{j'}$. This "tree decomposition" has the same "width" as $T$.
    Then, apply \Cref{fact:restriction_tagged_treedec} to get rid of potentially useless "bags".

    Observe then that the "path induced" by 
    $\pi'$ in $(T', \bagmap', \tagmap')$ is simply
    \[
        \tagmappathprime{\pi'} =
        \textstyle{
		\Bigl\langle
            {b_{i_0} \choose \fun(z_0)},
			{b_{i_0+1} \choose \fun(z_1)},
			\hdots,
			{b_{j} \choose \fun(z_{n(j)})},
            {b_{j'} \choose \fun(z_{n(j')})},
			\hdots,
			{b_{i_{n-1}} \choose \fun(z_{n-1})},
			{b_{i_{n-1}+1} \choose \fun(z_n)}
        \Bigr\rangle
        }
    \]
    and thus $\tagmappathprime{\pi'}$
    is strictly shorter than $\tagmappath{\pi}$ since $j+2 \leq j'$,
    by definition of these indices.
    Finally, observe that if $(T, \bagmap, \tagmap)$ is "fine" 
    then so is $(T', \bagmap', \tagmap')$. 
  
    Overall, we built $\fun'\colon \rho' \surj \alpha'$ together with a "fine tagged tree decomposition" $(T', \bagmap', \tagmap')$ of "width" at most $k$
    where $\alpha \contained \alpha'$ (by \Cref{fact:refinement-contained}),
    and $\rho' \in \Refin(\gamma)$ is such that $\nbatoms{\rho'} \leq \nbatoms{\rho}$, and
    for each atom of $\gamma$, the "refinement" of this atom in $\rho$ is exactly the
    same as the "refinement" of this atom in $\rho'$ except possibly for one atom,
    for which the path induced in $T'$ by its "refinement" in $\rho'$ is strictly shorter than the "path induced" in $T$ by its "refinement" in $\rho$.
    After iterating this construction as many times as needed, we obtain
    a "trio" as in the conclusion of \Cref{lemma:locally_acyclic_treedec}, which concludes our proof.
\end{proof}

\subsection{{\AP}Short Paths}

Ultimately, \Cref{lemma:locally_acyclic_treedec} will allow us to give a bound on the number of
leaves of a "fine tagged tree decomposition" of a "trio". The following claim---which is 
significantly more technical than the foregoing---will give us a bound on the height of
a "decomposition@fine tagged tree decomposition".

\begin{restatable}{lem}{shapedecomposition}
    \AP\label{lemma:shape-decomposition}
    Let $\fun\colon \rho \surj \alpha$ be a "trio" and $(T, \bagmap, \tagmap)$ be a "locally acyclic"
    "fine tagged tree decomposition" of "width" at most $k$ of $\fun$.
    Then there is a "trio" $\fun'\colon \rho' \surj \alpha'$ and a "fine tagged tree decomposition" $(T', \bagmap', \tagmap')$ of width at most $k$ of $\fun'$ such that:
    \begin{itemize}
        \item $\alpha \contained \alpha'$,
        \item $(T',\bagmap', \tagmap')$ is "locally acyclic" "wrt" $f'$, and
        \item the length of the longest "non-branching path" in $T$ is at most
        $\+O(\nbatoms{\gamma}\cdot (k+1)^{\nbatoms{\gamma}})$.
    \end{itemize}
\end{restatable}

\smallskip

To prove \Cref{lemma:shape-decomposition}, we will try to find, in a long "non-branching path",
some kind of shortcut. The piece of information that is relevant to finding this shortcut
is what we call the "profile" of a "bag".

\begin{figure}[tbp]
	\centering
	\includegraphics[scale=.48]{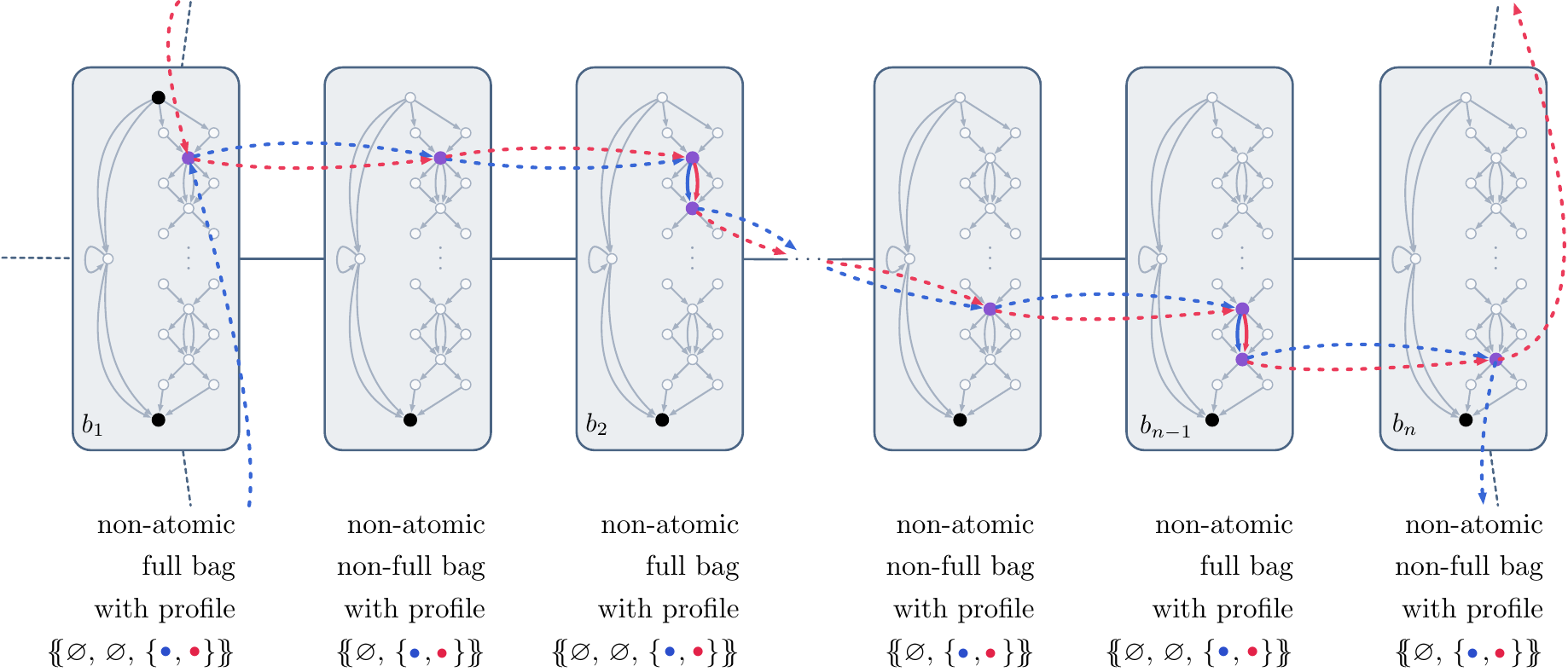}
	\caption{
		\AP\label{fig:trio-tree-dec-long-path}
        "Profiles" of the "bags" in the "non-branching path" between $b_1$ and $b_n$ in the "fine tagged tree decomposition" obtained from \Cref{fig:trio-tree-dec} after applying 
		\Cref{lemma:locally_acyclic_treedec,fact:restriction_tagged_treedec} to both the
		red and blue "atom refinements".
	}
\end{figure}

\begin{definition}[Types and Profiles]
    \AP Given a "trio" $\fun\colon \rho \surj \alpha$ and a "fine tagged tree decomposition"
    $(T, \bagmap, \tagmap)$ of $\fun$, for each "bag" $b$ of $T$, we say that:
    \begin{itemize}
        \itemAP $b$ is ``""atomic""'' if there is at least one atom $e \in \tagmap^{-1}[b]$ 
            and at least one variable $x$ of $e$ such that $x \in \vars(\gamma)$, "ie", the atom $e$ is not in the `middle' part of an "atom refinement";
        \itemAP otherwise, when $b$ is "non-atomic", we assign to each variable
            $z \in \bagmap(b) \subseteq \vertex{\alpha}$ a ""type""\phantomintro{\type}
            \begin{align*}
                \reintro*\type_z^b \defeq
                \bigl\{ x \atom{L} y \text{ "atom" of $\gamma$} \;\big\vert\; {} &
                    \text{the path "induced" by the "atom refinement"} \\
                    &\text{of $x \atom{L} y$ in $\rho$ "leaves" $b$ at $z$}
                \bigr\},
            \end{align*}
			where each "type" is potentially the empty set.
            Then the ""profile"" of $b$ is the multiset of the types of $z$
            when $z$ ranges over $\bagmap(b)$.
    \end{itemize}
\end{definition}

Note that $\rho$ and $\alpha$ can have arbitrarily more "atoms" than the original query
$\gamma$, and so the numbers of "bags" in $T$ can be arbitrarily high. However,
only few of them can be "atomic": an "atom refinement" of "atom" of $\gamma$
contains at most two "atoms" with a variable from $\gamma$---namely the
first and the last "atom" in the "refinement@atom refinement".
\begin{fact}
	\AP\label{fact:bound-atomic-bags}
	There is at most $2 \nbatoms{\gamma}$ "atomic bags" in $T$.
\end{fact}

Consider the "fine tree decomposition" of \Cref{fig:path-induced}, and now apply the construction of \Cref{lemma:locally_acyclic_treedec} to the red "atom refinement", followed by \Cref{fact:restriction_tagged_treedec}. We now obtain a "non-branching path"
between "bags" $b_1$ and $b_n$. We depict it in \Cref{fig:trio-tree-dec-long-path}:
the implicit "bags", hidden behind the dashed edges in \Cref{fig:path-induced} (see \Cref{nota:nice-tree-dec}), are made explicit in this new figure, and, moreover, the rest of the "fine tree decomposition" is not drawn.
Lastly, for each "bag", we indicate if it is "full" and if it is "atomic"; when it is not "atomic", 
we provide the "profile" of the "bag".

The rest of the proof consists in two parts.
First, we show that if two "non-atomic bags" $b$ and $b'$
occurring in some "non-branching path"
of $T$ have the same "profile", then we can essentially replace 
the path between $b$ and $b'$ by a path of constant length (\Cref{claim:shortening-paths}). And second, we show that in every sufficiently long "non-branching path"
we can find $b$ and $b'$ satisfying the aforementioned property: this
part simply relies on an enhanced ``pigeonhole principle'' (\Cref{fact:pigeon-hole}).

\begin{restatable}{clm}{shorteningpaths}
    \AP\label{claim:shortening-paths}
	Let $\fun\colon \rho \surj \alpha$ be a "trio", and consider a
	"fine tagged tree decomposition" of $f$ which is "locally acyclic".
    Suppose there are two "bags" $b$ and $b'$ such that:
	\begin{enumerate}
		\item they contain at most $k$ nodes ("ie", not "full" "bags"),
		\item they have the same "profile",
		\item there is a "non-branching path" in $T$ between these "bags", and
		\item no "bags" of the path between $b$ and $b'$ (both included) are "atomic".
	\end{enumerate}
	Then, there exists a "trio" $\fun'\colon \rho' \surj \alpha'$ and a
    "fine tagged tree decomposition" of $\fun'$ of "width" at most $k$
    that can be obtained by replacing the "non-branching path" between $b$ and $b'$
    in the "fine tagged tree decom\-position" of
    $\fun$ by another "non-branching path" with at most 
    $2k+1$ "bags", such that $\alpha \contained \alpha'$.
\end{restatable}
The proof of \Cref{claim:shortening-paths} relies on the definition of "profile", which
was specifically designed so that we can "condense" every "refinement" between $b$ and $b'$,
while preserving every needed property of the "trio".
We give first an informal and then a formal proof of \Cref{claim:shortening-paths}, which are 
illustrated in \Cref{fig:shortening-paths}.

\begin{figure}
    \centering%
	\begin{subfigure}{.49\linewidth}
		\centering
		\includegraphics{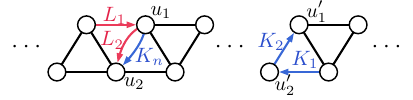}
		\caption{
			\AP\label{fig:shortening-paths-approx}
			An "approximation" $\alpha$.
		}
	\end{subfigure}
	\hfill
	\addtocounter{subfigure}{1}
	\begin{subfigure}{.49\linewidth}
		\centering
		\includegraphics{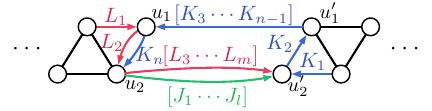}
		\caption{
			\AP\label{fig:shortening-paths-approx-result}
			The resulting "approximation" $\alpha'$.
		}
	\end{subfigure}\\[1em]
	\addtocounter{subfigure}{-2}
	\begin{subfigure}{\linewidth}
		\centering
		\includegraphics{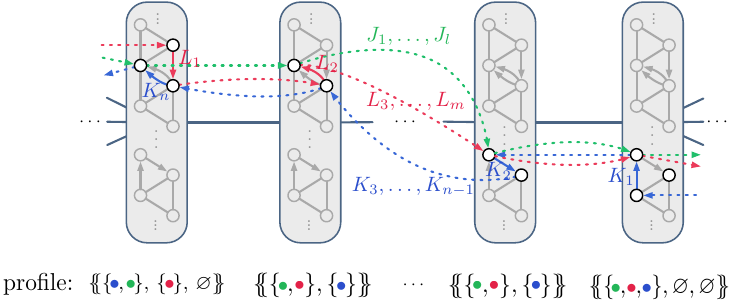}
		\caption{
			\AP\label{fig:shortening-paths-tree-dec}
			A "fine tagged tree decomposition" of $\fun\colon \rho \surj \alpha$
			together with the "path induced" by three "atom refinements"
			of $\rho$ (in red, green \& blue).
		}
	\end{subfigure}
    \caption{%
        \AP\label{fig:shortening-paths}%
		A "trio" $\fun\colon \rho \surj \alpha$---where $\rho$ and $\fun$ are implicit---
		together with one of its "fine tagged tree decomposition"
		(\Cref{fig:shortening-paths-approx,fig:shortening-paths-tree-dec}).
		There are two non-"full bags" in the path with the same "profile", and thus the query $\alpha$ can be simplified to $\alpha'$
		(see \Cref{fig:shortening-paths-approx-result})
		by applying "condensations" to the "atom refinements" involved.
    }
\end{figure}

\begin{proof}[Informal proof of \Cref{claim:shortening-paths}]
    \AP\label{proof-claim:shortening-paths}
    If $b$ and $b'$ have the same "profile", then in particular they have the same cardinality $m$,
    which is smaller or equal to $k$ by assumption.
    Let $\bagmap(b) = \{x_1,\hdots, x_m\}$ and $\bagmap(b') = \{y_1,\hdots,y_m\}$ be such that: 
    $\type^{b}_{x_i} = \type^{b'}_{y_i}$ for all $1 \leq i \leq m$.
	Note that the $x_i$'s don't need to be distinct from the $y_i$'s.
	Essentially, we can then "condense"
    every "atom refinement" in $\rho$ of some "atom" occurring in a set
    of the form $\type^{b}_{x_i} = \type^{b'}_{y_i}$ for some $i$.
	At this point, bags strictly comprised between
    $b$ and $b'$ are discarded, and so are variables of $\alpha$ that do not occur anywhere else.
    We are left with two halves of a "fine tagged tree decomposition" that we need to merge,
    which can easily be done by using \Cref{prop:connecting-tree-decompositions}.
    The construction makes use of some crucial ingredients to guarantee its correctness.
    \begin{itemize}
        \item First, an "atom" $y \atom{L} y'$ of $\gamma$ cannot occur
            in two different "types", allowing us to do the "condensation" of each "atom refinement" independently---this property is guaranteed by the fact that we started
			with a "locally acyclic" "fine tagged tree decomposition", so an "atom" of $\gamma$ cannot leave a given "bag" at two different variables, by
			\Cref{fact:acyclic-decomposition-leave-forever}.
        \item Second, this "condensation" forces us to add new "atoms" in $\alpha$ (to preserve the
            existence of a "homomorphism" from the "refinement" to the approximation)
            from some variables of $\bagmap(b)$ to some variables of $\bagmap(b')$,
            but we only add edges from $x_i$ to $y_i$, and never from $x_i$ to $y_j$ with
			$i \neq j$. This allows us to preserve the "tree-width" of the approximation
			by using \Cref{prop:connecting-tree-decompositions}.\qedhere
    \end{itemize}
\end{proof}

\begin{proof}[Formal proof of \Cref{claim:shortening-paths}]
    Let
    \[\bagmap(b) = \{x_1,\hdots, x_m\}
    \text{ and }
    \bagmap(b') =\{y_1,\hdots,y_m\}
    \]
    be as in the informal proof.
    Note that given an "atom" $x \atom{L} y$ of $\gamma$ and
	a "bag", there is at most
    one variable of $\alpha$ "st" $x \atom{L} y$ is in the "type" of this variable
	at this bag, by \Cref{fact:acyclic-decomposition-leave-forever}.

    For every "atom" $x \atom{L} y$ of $\gamma$, let 
    $\pi(x \atom{L} y) \defeq t_0 \atom{L_1} \hdots \atom{L_n} t_n$
    be its "refinement" in $\rho$.
    If $x \atom{L} y$ is not in some "type" of the "profile" of $b$
    (or equivalently, of $b'$), leave it as is.
    Otherwise, let $i$ (resp.\ $j$) be the unique
    index (by acyclicity) such that
    $\tagmappath{t_0 \atom{L_1} \hdots \atom{L_n} t_n}$ "leaves" $b$
    at $\fun(t_i)$ (resp.\ "leaves" $b'$ at $\fun(t_j)$).
    Define
    \[
		\pi'(x \atom{L} y) \defeq
        t_0 \atom{L_1} \hdots \atom{L_i} t_i \atom{\contract{L_{i+1}\cdots L_j}}
        t_j \atom{L_{j+1}} \hdots \atom{L_n} t_n
    \]
    when $i \leq j$ and otherwise the definition is symmetric.
    Then, let $\rho'$ be the "refinement" of $\gamma$ obtained by simultaneously
    substituting $\pi(x \atom{L} y)$ with $\pi'(x\atom{L} y)$  in $\rho$,
    for every "atom" $x\atom{L} y$ of $\gamma$.

    Then, let $\alpha'$ be the query obtained by first
    adding the atoms  
    \[\fun(t_i) \atom{\contract{L_{i+1}\cdots L_j}} \fun(t_j),\]
    and observe that $\fun\colon \rho \surj \alpha$
    induces a "homomorphism" $\fun'\colon \rho' \homto \alpha'$---in particular, note that
	because of assumption (4) of our claim, we could not have removed images of
	free variables of $\gamma$.
    Moreover, by construction, $\alpha \contained \alpha'$ (by \Cref{fact:refinement-contained}).
	As usual, we restrict $\alpha'$ to the image of $f'$ so that it
	becomes "strong onto", while preserving that fact that $\alpha \contained \alpha'$.
	Finally, we build a "tagged tree decomposition" $(T', \bagmap', \tagmap')$
    of $f'$ by applying \Cref{prop:connecting-tree-decompositions}; it can be applied because:
	\begin{itemize}
		\item by assumption $(1)$ and $(2)$ of the claim, both "bags" have the same cardinality $m \leq k$;
		\item the variables in common between the first and second half of the "decomposition@fine tagged tree decomposition" are necessarily included in $Z \defeq \bagmap(b) \cap \bagmap(b')$
		since we started from a "tree decomposition";
		\item we only add atoms from $x_i$ to $y_i$: depending on whether $x_i \in^? Z$,
		and whether $y_i \in^? Z$, we fall in one of the four types of atoms allowed
		by \Cref{prop:connecting-tree-decompositions}.
	\end{itemize}
	This concludes the proof of \Cref{claim:shortening-paths}.
\end{proof}

\begin{figure}[tbp]
	\centering
	\begin{subfigure}[c]{\linewidth}
		\centering
		\includegraphics[scale=.7]{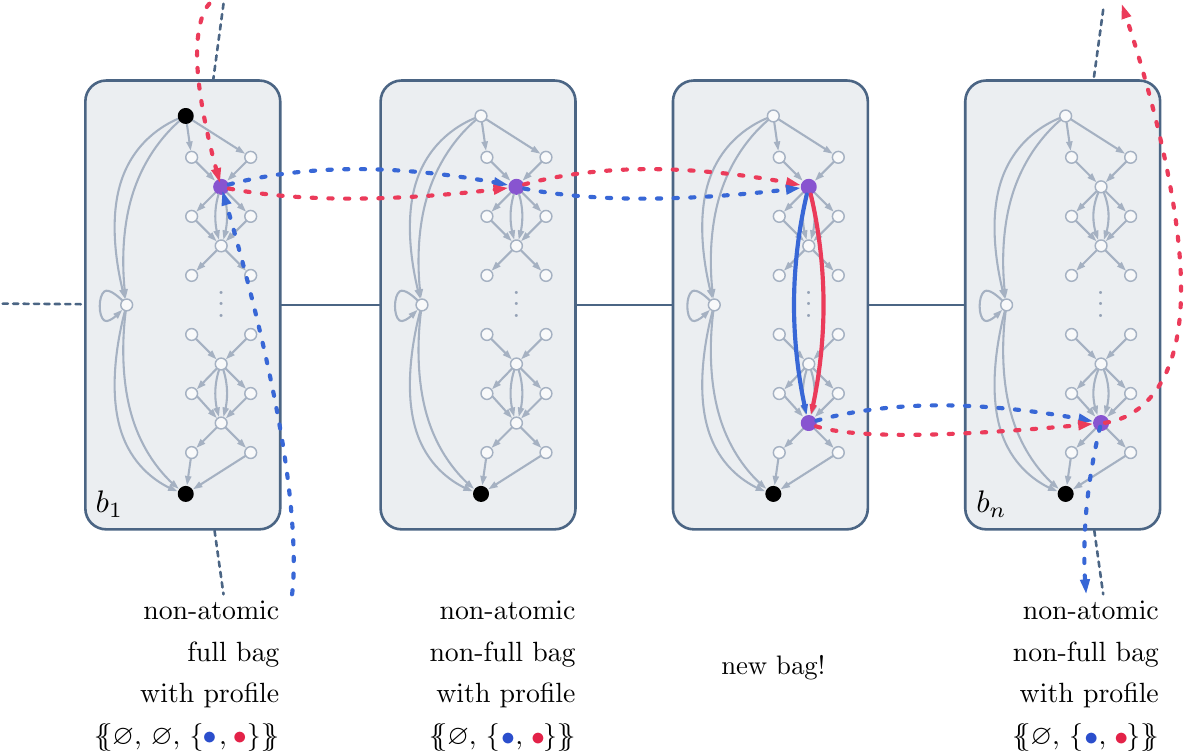}
		\smallskip
		\caption{
			\AP\label{fig:trio-tree-dec-long-path-shortened}
			"Non-branching path" in the ``new'' "fine tagged tree decomposition".
		}
	\end{subfigure}\\[1em]
	\begin{subfigure}[c]{.25\linewidth}
		\centering
		\includegraphics{trio-query.pdf}
		\caption{
			\AP\label{fig:trio-result-query}
			The original query $\gamma$ of "tree-width" 3.\\~
		}
	\end{subfigure}
	\hfill
	\begin{subfigure}[c]{.35\linewidth}
		\centering
		\includegraphics{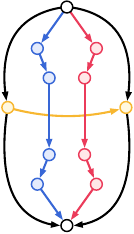}
		\caption{
			\AP\label{fig:trio-result-refinement}
			The ``new'' "refinement" $\rho'$ of $\gamma$.\\~
		}
	\end{subfigure}
	\hfill
	\begin{subfigure}[c]{.35\linewidth}
		\centering
		\includegraphics{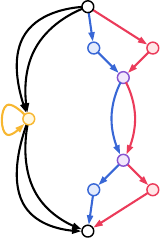}
		\caption{
			\AP\label{fig:trio-result-approx}
			The ``new'' "approximation" $\alpha'$ of "tree-width" 2.
		}
	\end{subfigure}
	\caption{
		\AP\label{fig:trio-result}
		"Trio" resulting from applying \Cref{claim:shortening-paths}
		between the second and last "bags" of \Cref{fig:trio-tree-dec-long-path}.
	}
\end{figure}
In \Cref{fig:trio-tree-dec-long-path-shortened}, we depict the "non-branching path" (the rest
of the "fine tree decomposition" is not depicted as it is left unchanged) obtained
by applying the construction used to prove \Cref{claim:shortening-paths}
between the second and last bag of \Cref{fig:trio-tree-dec-long-path}. Observe that a "non-branching path" of size $\+O(n)$ is replaced, by this procedure, by a path with three "bags".
Then, after applying \Cref{fact:restriction_tagged_treedec}, we obtain a "trio"
depicted in \Cref{fig:trio-result-query,fig:trio-result-refinement,fig:trio-result-approx}.

Before moving to the proof of \Cref{lemma:shape-decomposition}, we establish one last result.
\begin{fact}
    \AP\label{fact:pigeon-hole}
    Let $n, d, t \in \Nat$. Let $P$ be a set with at most $n$ elements,
    and $\tilde P$ be the disjoint union of $P$ and $\{\trap, \avoid\}$.
    For every natural number $m \geq 2(t+1)d(n+1) + 2t$, 
    for every sequence $(p_i)_{0 \leq i < m} \in \tilde P^m$ containing at most
    $t$ elements equal to $\trap$, if at most half of  the elements of the sequence
    are equal to $\avoid$, then there exists
    $i < i'$ such that $p_i = p_{i'} \neq \avoid$, $i' - i \geq d$ and
    $p_j \neq \trap$ for every $i \leq j \leq i'$.
\end{fact}

\begin{proof}
    First extract from $(x_i)_{0 \leq i < m}$ the subsequence of elements
    distinct from $\avoid$, of length at least $\lceil\frac{m}{2}\rceil \geq (t+1)d(n+1) + t$.
    Then extract from it contiguous subsequences that avoid
    the $\trap$ element. Since there is at most $t+1$ subsequences like this,
    one of them must have size at least $d(n+1)$. Denote by
    $(y_i)_{0 \leq i < d(n+1)}$ the prefix of such a subsequence.
    Applying the pigeon-hole principle
    to $(y_{i\cdot d})_{0 \leq i < n+1}$ yields the desired result.
\end{proof}

\begin{proof}[Proof of \Cref{lemma:shape-decomposition}]
    Let $\fun\colon \rho \surj \alpha$ be a "trio", and $(T, \bagmap, \tagmap)$ be a
    "locally acyclic" "fine tagged tree decomposition" of $\fun$. If there is a "non-branching path"
    $(b_i)_{0\leq i<m}$ in $T$ of length at least $m$, let $(p_i)_{0 \leq i<m}$ be
    the sequence defined by letting:
    \[
        p_i \defeq \begin{cases}
            \trap & \text{if $b_i$ is "atomic",} \\
            \avoid & \text{if $b_i$ contains $k+1$ variables,} \\
            \text{"profile" of } b_i & \text{otherwise.}
        \end{cases}    
    \]
    
    Observe that, by \Cref{fact:acyclic-decomposition-leave-forever},
	"profiles" can be seen encoded as partial functions from the set of "atoms" of $\gamma$ to $\lBrack 1,k\rBrack$---of course this encoding is not surjective---, so
	there are at most $(k+1)^{\nbatoms{\gamma}}$ different "profiles" on "bags" with at most $k$ variables.
    Applying \Cref{fact:pigeon-hole} for $n = (k+1)^{\nbatoms{\gamma}}$, $d = 2k+1$, $t = 2\nbatoms{\gamma}$ yields, under the assumption that
    \[
        m \geq m_0 \defeq 2(2\nbatoms{\gamma}+1)(2k+1)((k+1)^{\nbatoms{\gamma}}+1)
		+ 4\nbatoms{\gamma},
    \]
    the existence of indices $i < i'$ such that $i' - i \geq 2k+1$,
    and $b_i$ and $b_{i'}$ have the same "profile", contain at most $k$ variables,
    and every bag $b_j$ for $i \leq j \leq i'$ is "non-atomic"---note that the hypothesis
    of \Cref{fact:pigeon-hole} are satisfied since at most $t = 2\nbatoms{\gamma}$
    "bags" of $(b_i)_{0 \leq i < m}$ are "atomic" ("cf" \Cref{fact:bound-atomic-bags}), and assuming "wlog"~that
    no two consecutive "bags" of $(b_i)_{0 \leq i < m}$ are identical, since the "tagged tree decomposition" $(T, \bagmap, \tagmap)$ of "width" $k$ is "fine", at most half of the 
    "bags" contain $k+1$ variables.
    The assumption $i' - i \geq 2k+1$ means that the path from $b_i$ to $b_{i'}$ has
    length at least $2k+2$, and thus applying \Cref{claim:shortening-paths} will
    strictly shorten this path.
    Note that \Cref{claim:shortening-paths} preserves the "fineness" of the "tagged 
    tree decomposition", its "local acyclicity", and that the size of
    this "tree decomposition" is strictly smaller (in number of nodes) than
    the original "tree decomposition". By iteratively applying this construction,
    we obtain a "trio" $\fun'\colon \rho' \surj \alpha'$ together with a "locally acyclic"
    "fine tagged tree decomposition" $T'$
    of "width" at most $k$, such that $\alpha \contained \alpha'$ (by a variation of
	\Cref{fact:refinement-contained})
    and every "non-branching path" of $T'$ has length at most\footnote{Recall that $k$
	is fixed.} 
    $m_0 - 1 \in 
    \+O(\nbatoms{\gamma}\cdot (k+1)^{\nbatoms{\gamma}})$.
\end{proof}

\subsection{\AP{}Proof of Lemma~\ref{lemma:bound_size_refinements}}
\label{sec:proof-bound-size-refinements}

Finally, our main lemma follows from
\Cref{lemma:locally_acyclic_treedec,lemma:shape-decomposition}.

\begin{proof}[Proof of \Cref{lemma:bound_size_refinements}]
    In order to show $\MUA{\gamma}{\Tw} \contained \MUAHomBounded{\gamma}{\Tw}{\leq\l}$---the other
    "containment" being trivial---, pick a "trio" $\fun\colon \rho \surj \alpha$.
    Applying \Cref{lemma:locally_acyclic_treedec} and then \Cref{lemma:shape-decomposition}
    yields the existence of a "trio" $\fun'\colon \rho' \surj \alpha'$ together
    with a "fine tagged tree decomposition" $(T', \bagmap', \tagmap')$ of $\fun'$
    such that $\alpha \contained \alpha'$ and $(T', \bagmap', \tagmap')$ is "locally acyclic",
    and any "non-branching path" in $T'$ has length at most
    $\+O(\nbatoms{\gamma}\cdot (k+1)^{\nbatoms{\gamma}})$.

    Moreover, we can assume "wlog", by applying \Cref{fact:restriction_tagged_treedec},
    that every leaf of $T'$ is "tagged" by at least one "atom" of $\rho'$.
    The "local acyclicity" of $T'$ implies that if $b$ is a leaf of $T'$,
    and
    $
        \pi \defeq x \atom{L_1} t_1 \atom{L_2} \cdots \atom{L_{n-1}} t_{n-1} \atom{L_n} y
    $
    is an "atom refinement" in $\rho'$ of some "atom" $x \atom{L} y$ of $\gamma$,
    then if $b$ is tagged by one "atom" of $\pi$ this "atom" must either be $z_0 \atom{L_1} z_1$
    or $z_{n-1} \atom{L_n} z_n$ by "local acyclicity".
    The number of such "atoms" in $\rho'$ being bounded by $2\nbatoms{\gamma}$, we conclude that
    $T'$ has at most $2\nbatoms{\gamma}$ leaves.

    Then, observe that a tree with at most $p$ leaves and whose "non-branching paths"
    have length at most $q$ is of height at most\footnote{The length of a path being its number of nodes,
    and with the convention that the height of a single node is zero.} $p\cdot q - 1$. We
    conclude that the height of $T'$ is 
    $\+O(\nbatoms[2]{\gamma}\cdot (k+1)^{\nbatoms{\gamma}})$.
    Using again the "local acyclicity" of $T'$, observe that the "refinement length" of $\rho'$
    is at most twice the height of $T'$, and hence $\rho' \in \Refin[\leq \l](\gamma)$
    where $\l = \Theta(\nbatoms[2]{\gamma}\cdot (k+1)^{\nbatoms{\gamma}})$.
    In other words, $\alpha' \in \MUAHomBounded{\gamma}{\Tw}{\leq\l}$.
    Hence, we have shown that for all $\alpha \in \MUAHom{\gamma}{\Tw}$,
    there exists $\alpha' \in \MUAHomBounded{\gamma}{\Tw}{\leq\l}$ such that $\alpha \contained \alpha'$.
\end{proof}

This concludes \Cref{sec:proof-key-lemma} and the proof of the "Key Lemma".
The next four sections are independent of one another:
\begin{itemize}
	\item in \Cref{sec:sre}, we show that the "2ExpSpace" complexity
	of the "semantic tree-width $k$ problem" can be dropped down to
	"PiP2" under assumptions on the regular languages;
	\item in \Cref{sec:acyclic-queries,sec:semantic-path-width}, we adapt
	the proofs of this section to deal with
	"semantic tree-width" 1 and "semantic path-width" $k$, respectively.
	\item in \Cref{sec:lowerbound}, we prove an "ExpSpace" lower bound for the
	"semantic tree-width $k$ problem" and "semantic path-width $k$ problems";
\end{itemize}

\section{\AP{}Semantic Tree-Width for Simple Queries}
\label{sec:sre}
\AP A ""simple regular expression"", or \reintro{SRE}, is a regular expression the form $a^*$ for some letter $a \in \A$ or of the form $a_1 + \dotsb + a_m$ for some $a_1, \dotsc, a_m \in \A$. 

\AP Let \intro*{\UCRPQSRE} be the set of all "UCRPQ" whose languages are expressed via "SREs". Observe that {\UCRPQSRE} is "semantically equivalent" to the class of "UCRPQs" over the closure under concatenation of "simple regular expressions"
since $\gamma(x,y) = x \atom{e_1 \cdot e_2} y$ is equivalent to $\gamma'(x,y) = x \atom{e_1} z \land  z \atom{e_2} y$.
Moreover, \UCRPQSRE{} also corresponds to "UC2RPQ" whose languages are expressed via "SREs";
in other words adding two-wayness does not increase the expressivity of the class. 

One interest of {\UCRPQSRE} comes from the fact that it is used widely in practice, as recent studies on SPARQL query logs on Wikidata, DBpedia and other sources show that this kind of regular expressions cover a majority of the queries investigated, "eg", 75\% of
the ``property paths'' ("C2RPQ" "atoms") of the corpus of 1.5M queries of Bonifati, Martens and Timm \cite[Table 15]{BonifatiMT-vdlbj20}.
An additional interest comes from the fact that the "containment problem" for {\UCRPQSRE} is much better behaved than for general "UCRPQs", since it is in {\pitwo} \cite[Corollary 5.2]{FigueiraGKMNT20}, that is, just one level up the polynomial hierarchy compared to the "CQ" "containment problem", which is in "NP" \cite{DBLP:conf/stoc/ChandraM77}, and in sharp contrast with the costly \expspace-complete "CRPQ" "containment problem" \cite{CGLV00,Florescu:CRPQ}.

We devote this section to showing the following result.

\thmSemTwSREpitwo 

Observe that "simple regular expressions" are "closed under sublanguages". Hence, in the light of
\Cref{thm:closure-under-sublanguages}, the "maximal under-approximation" of a {\UCRPQSRE} query by infinitary unions of "CQs" of "tree-width" $k$ is always equivalent to a {\UCRPQSRE} query  of "tree-width" $k$. We will see how the construction of the "maximal under-approximation" of the previous section can be exploited to improve the complexity from "2ExpSpace" down to "PiTwo".

\subsection{\AP{}Summary Queries}
We will first show that the "maximal under-approximation" of "tree-width" $k$ of a "UC2RPQ" can be expressed as a union of polynomial sized ``summary'' queries. Each "summary query" represents a union of exponentially-bounded "C2RPQs" sharing some \emph{common structure}. 
"Summary queries" are normal "UC2RPQ" queries extended with some special kind of atoms, called ``"path-$l$ approximations"''. 
Intuitively, they represent a "maximal under-approximation" of "tree-width" $l$ of queries of the form $\bigwedge_i x_i \atom{L_i} y_i$ such that $x_i \neq y_j$ for all $i,j$. 
Path-$l$ approximations may require an exponential size when represented as "UC2RPQs".
\AP Formally, a ""path-$l$ approximation"" is a query of the form ``$\intro*\pathl(X, Y, \delta)$''
where: 
\begin{enumerate}
	\item $X$, $Y$, are two disjoint sets of variables of size at most $l$, 
	\item $\delta(\bar z)$ is a conjunction of "atoms" $\bigwedge_{1 \leq i \leq n} A_i$ where $\bar z$ contains all variables of $X \cup Y$,
	\item each $A_i$ is a "C2RPQ" atom of the form $x \atom{L} y$ or $y \atom{L} x$  such that $x \in X$, $y \in Y$, and $L_i$ is a regular language over $\A$.
\end{enumerate}
We give the semantics of $\pathl(X, Y, \delta)$ in terms of infinitary unions of "C2RPQs".
A query like the one before is defined to be equivalent to the (infinitary) union of all queries $\alpha(\bar z) \in \MUA{\delta}{\Pw[l]}$ such that
\begin{align}
	\text{$\alpha$ has a "path decomposition" of "width" $l$ where $X$ is the root and $Y$ is the leaf,}\AP\label{eq:pathl:rootleafppty}
\end{align}
that is, the root and leaf "bags" contain precisely the vertices of $X$ and $Y$, respectively.
 See \Cref{fig:l-path-example} for an example.
\begin{figure}
  \centering
  \includegraphics[width = \textwidth]{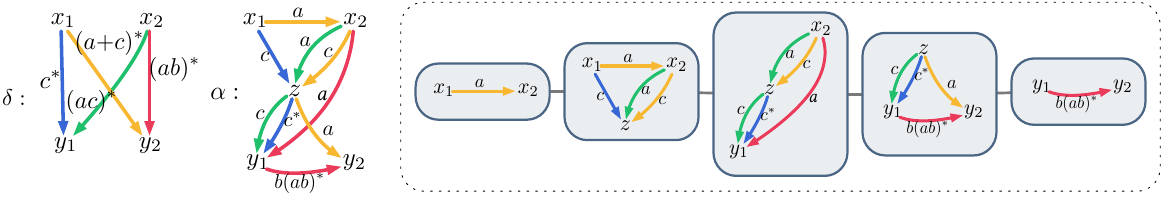}
  \caption{
	\AP\label{fig:l-path-example}
	Consider the "path-$l$ approximation" $\pathl(\set{x_1,x_2}, \set{y_1,y_2},\delta)$ where $l=2$ and $\delta$ is depicted on the left. Its semantics contains the "approximation" $\alpha(x_1,x_2,y_1,y_2) \in \MUA{\delta}{\Pw[l]}$ depicted in the middle because it has "path decomposition" of "width" $l$ verifying \eqref{eq:pathl:rootleafppty}, as shown on the right.
	}
\end{figure}

\AP We now simply define a ""$k$-summary query"" as a "C2RPQ" extended with "path-$l$ approximation" atoms for any $l\leq k$, with the expected semantics.
A \AP""refinement@@summary"" of a "$k$-summary query" is any "C2RPQ" obtained by
replacing "atoms" with "atom refinements", and 
each "path-$l$ approximation" $\pathl(X, Y, \delta)$ with any
$\alpha(\bar z) \in \MUA{\delta}{\Pw}$ verifying \eqref{eq:pathl:rootleafppty}. By definition,
a "database" satisfies a "$k$-summary query" if and only if it
satisfies one of its "refinements@@summary".

A \AP ""tree decomposition@@summary"" of a "$k$-summary query" $\gamma$ consists of
a pair $(T, \bagmap)$ with $\bagmap: \vertex{T} \to \pset{\vars(\gamma)}$
such that:
\begin{itemize}
	\item for every classical "atom" $x \atom{L} y$ in $\gamma$,
		there is a "bag" $b \in T$ such that $\{x,y\} \subseteq \bagmap(b)$;
	\item for every "path-$l$ approximation" $\pathl(X, Y, \delta)$ in $\gamma$,
		there are two adjacent "bags" $b, b' \in T$ such that $X \subseteq \bagmap(b)$
		and $Y \subseteq \bagmap(b')$.
\end{itemize}
The "width" is defined as usual. Then, by
\Cref{fact:refinement-tw}, we obtain the following upper bound.

\begin{fact}
	\AP\label{fact:tree-width-summary}
	For any $k\geq 2$, any "refinement@@summary" of a "$k$-summary query" with
	a "tree decomposition@@summary" of "width" at most $k$ is a "C2RPQ" of "tree-width" at most $k$.
\end{fact}

Lastly, a \AP""homomorphism@@summary"" from a "C2RPQ"
$\gamma(\bar z) = \bigwedge_{i} x_i \atom{L_i} y_i$
to a "summary query" $\delta(\bar z') =
\bigl( \bigwedge_{j} x'_j \atom{L'_j} y'_j \bigr)
\wedge \bigl( \bigwedge_{j'} \pathl(X_{j'}, Y_{j'}, \delta_{j'}) \bigr)$
consists of a mapping $\fun$ from variables of $\gamma$ to variables of $\delta$
such that $f(\bar z) = \bar z'$, and for each $i$,
there is an atom $f(x_i) \atom{L_i} f(y_i)$ in $\delta$.
Note that if there is a homomorphism from $\gamma(\bar z)$ to $\delta(\bar z')$,
then $\delta(\bar z') \contained \gamma(\bar z)$.

Let us fix $\+L$ to be any class  "closed under sublanguages". For every $\gamma \in \CtwoRPQ(\+L)$, we define \AP $\intro*\Qapp$ as the set of all "$k$-summary queries" $\alpha$ such that:
\begin{enumerate}[label=(\alph*)]
	\item $\alpha$ has a "fine tagged tree decomposition" $(T, \bagmap)$ of "width" at most $k$,
	\item there exists a "strong onto" "homomorphism@@summary" from a "refinement" $\rho$ of $\gamma$ to $\alpha$, 
	\item $T$ has at most $2\nbatoms{\gamma}$ leaves, and every
		"non-branching path" of $T$ consisting only of "non-atomic" "bags" must contain at most two "non-full bags".
\end{enumerate}
Note that since $\alpha$ is a homomorphic image a "refinement" of $\gamma$,
and since $\+L$ is "closed under sublanguages", then $\alpha$ has only $\+L$-labelled "atoms".

\begin{lemma}\AP\label{lemma:MUAk-union-of-polysized-summary-test-np}
    \AP Let $k \geq 2$. For every \textbf{finite} class $\+L$ "closed under sublanguages",
	and for every $\gamma \in \CtwoRPQ(\+L)$, 
	we have:
	\begin{enumerate}
		\itemAP \label{lemma:MUAk-union-of-polysized-summary-test-np:1} $\Qapp \semequiv \MUA{\gamma}{\Tw}$,
		\itemAP \label{lemma:MUAk-union-of-polysized-summary-test-np:2} $\Qapp$ is a union of polynomial-sized "$k$-summary queries" having only
			$\+L$-labelled "atoms", and
		\itemAP \label{lemma:MUAk-union-of-polysized-summary-test-np:3} one can test in "NP" if a "summary query" is part of this union.
	\end{enumerate}
\end{lemma}
  
\begin{proof}
	Point \eqref{lemma:MUAk-union-of-polysized-summary-test-np:2} follows directly from the definition: there are few branches in the decomposition, branches are short, and each "bag" cannot contain more than $(k+1)^2 \cdot |\+L|$ "atoms" labelled with $\+L$-languages.

	For point \eqref{lemma:MUAk-union-of-polysized-summary-test-np:3}, recall that one can check if a query has "tree-width" at most $k$ in linear time, "eg" using Bodlaender's algorithm \cite[Theorem 1.1]{bodlaender1996treewidth}.

	To prove \eqref{lemma:MUAk-union-of-polysized-summary-test-np:1}, notice first that 
	$\Qapp \contained \MUA{\gamma}{\Tw}$ as a consequence
	of \Cref{fact:tree-width-summary}.

	For the converse "containment", we use \Cref{coro:equivalence_under_approx_homomorphism_twk}
	and prove instead $\MUAHom{\gamma}{\Tw} \contained \Qapp$.
	Observe that, as corollary of the proof of \Cref{lemma:bound_size_refinements}, we can assume to have 
	$\MUAHom{\gamma}{\Tw}$ expressed as a union of $\CtwoRPQ(\+L)$ with a "fine tagged tree decomposition" of "width" $k$ with at most $2\nbatoms{\gamma}$ leaves, and hence it suffices to replace each "non-branching paths" having "non-atomic" bags with "path-$l$ approximations".

	Indeed, fix a "fine tagged tree decomposition" and a "trio"
	$\fun\colon \rho \surj \alpha$. Given a long "non-branching path" from "bag" $b_X$ with
	variables $X$ to a "bag" $b_Y$ with variables $Y$,
	such that $b_X$ and $b_Y$ are "non-full", and no "bag" in between is "atomic",
	define $X' \defeq X \setminus Y$ and $Y' \defeq Y \setminus X$.
	Consider the set $\+S$ of "atoms" $u \atom{L} v$ of $\gamma$, such that the "path induced"
	$\bigl( {b_i \choose z_i} \bigr)_i$ by the "refinement", say
	\[
		u = w_0 \atom{L_0} w_1 \atom{L_1} \cdots \atom{L_n} w_n	= v
	\]
	of $u \atom{L} v$ in $\rho$ 
	goes through $b_X$ at some variable of $X'$ and through $b_Y$ at some variable
	of $Y'$, in the sense that $b_i = b_X$ and $z_i \in X'$ for some $i$,
	and $b_j = b_Y$ and $z_j \in Y'$ for some $j$.
	There exist $i', j'$ such that $f(w_{i'}) = z_i$ and $f(w_{j'}) = z_j$,
	and "wlog" $i' < j'$.
	Now let $\alpha'$ be the query obtained from $\alpha$ by removing all "atoms" tagged in a bag between $b_X$ and $b_Y$, and add a "path-$l$ approximation" query
	\[
		\pathl(X', Y', \delta)
	\]
	where $\delta$ is the conjunction over $u \atom{L} v \in \+S$ of
	$w_{i'} \atom{\contract{L_{i'}\cdots L_{j'}}} w_{j'}$.
	Repeat this operation for every non-trivial "non-branching path" with
	"non-atomic" "bags". We obtain $\alpha'' \in \Qapp$ "st"
	$\alpha \contained \alpha''$, which concludes the proof
	that $\MUAHom{\gamma}{\Tw} \contained \Qapp$.
\end{proof}

\subsection{\AP{}Semantic Tree-Width Problem}
With the previous results in place, we now show that the semantic tree-width $k$ problem is in {\pitwo} for \UCRPQSRE, for every $k>1$.

\thmSemTwSREpitwo*
\begin{proof}
  It suffices to show the statement for any $\CRPQSRE~\gamma$. Remember that $\gamma$ is of "semantic tree-width" $k$ if, and only if, $\gamma \contained \Qapp$. The first ingredient to this proof is the fact that this "containment" has a polynomial counterexample property.
  
  \begin{claim}\AP\label{claim:poly-sized-counterexample-sre}
    If $\gamma \not\contained \MUAHom{\gamma}{\Tw} $ then there is a polynomial-sized "expansion" $\anexpansion$ of $\gamma$ such that $\anexpansion \not\contained \MUAHom{\gamma}{\Tw} $.
  \end{claim}
  \begin{proof}
    \newcommand{\Qalt}{\+U}
    \AP
    Let us call any atom with a language of the form $a^*$ a ""recursive atom"", and any other atom a ""non-recursive atom"".
    Let $n$ be the number of "non-recursive atoms" of $\gamma$. Hence, any refinement $\rho \in \Refin(\gamma)$ has $n$ atoms deriving from "non-recursive@@nrat" "atom refinements", all the remaining ones derive from "recursive@@rat" "atom refinements".

    We will work with the infinitary union of "conjunctive queries" $\Qalt\defeq\MUAHom{\gamma}{\Tw} \cap \textnormal{"CQ"}$. Note that $\Qalt=\set{ \alpha \in \textnormal{"CQ"} \mid \alpha \in \Tw \text{ and there is } \anexpansion \in \Exp(\gamma) \text{ "st" } \anexpansion \surj \alpha }$.
    It is easy to see that $\Qalt \semequiv \MUAHom{\gamma}{\Tw}$ as a consequence of \Cref{fact:refinement-tw}.
    By \Cref{prop:cont-char-exp-st}, we have $\gamma \not\contained \Qalt$ if, and only if, there is some "expansion" $\anexpansion$ of $\gamma$ such that $\anexpansion \not\contained \Qalt$. In turn, this happens if, and only if, there is no $\delta \in \Qalt$ such that $\delta \homto \anexpansion$.
    
    Take any such counterexample $\anexpansion$ of minimal size (in number of atoms). We show that for any "internal path" of $\anexpansion$ of the shape
    \[\pi = x_0 \atom{a} x_1 \atom{a} x_1 \dotsb x_{m-1}\atom{a} x_m,\]
    we have $m \leq n+1$. Hence, since $\anexpansion$ is an "expansion" of a {\CRPQSRE}, this means that the size of each \AP""atom expansion""---namely an "expansion" obtained from a query by only expanding one "atom"---of $\anexpansion$ is linearly bounded in the size of $\gamma$, and thus that $\anexpansion$ is quadratically bounded.

    By means of contradiction, if $m > n+1$ consider the "expansion" $\anexpansion'$ resulting from ``shrinking'' the path $\pi$ to a path $\pi'$ of length $n+1$. Hence, $\anexpansion'$ is smaller than $\anexpansion$, and since $\anexpansion$ was assumed to be minimal, $\anexpansion'$ cannot be a counterexample. Thus, there is some $\delta \in \Qalt$ such that $\fun_1:\delta \homto \anexpansion'$ for some "homomorphism" $\fun_1$. Further, by definition of $\Qalt$, we have $\fun_2 : \anexpansion'' \surj \delta$ for some $\anexpansion'' \in \Exp(\gamma)$. Consider then the composition $\anexpansion'' \surj \delta \homto \anexpansion'$ of $\fun_2$ with $\fun_1$ and let us call it $g: \anexpansion'' \homto \anexpansion'$. By definition of $n$ there must be at least one atom $x_i \atom{a} x_{i+1}$ of the shrunken path $\pi'$ of $\anexpansion'$ which either 
		(i) is not in the image of $\fun_1$, or 
		(ii) all its $g$-preimages proceed from "atoms" $z \atom{a} z'$ of $\anexpansion''$ which are in the "expansions" of "recursive atoms" of $\gamma$. 
	We show that, in both cases, we can replace $x_i \atom{a} x_{i+1}$ with a path of $a$'s of any arbitrary length $l>0$, obtaining a "conjunctive query" $\anexpansion'_{+l}$  which
	is---still---not a counterexample.  In the first case (i), we actually obtain that $\delta \homto \anexpansion'_{+l}$. In the second case (ii), we have to replace each atom $z \atom{a} z'$ in the $g$-preimage of $x_i \atom{a} x_{i+1}$ in $\anexpansion''$ by an $a$-path of length $l$, obtaining some "expansion" $\anexpansion''_{+l}$ of $\gamma$. We also replace each atom in the $\fun_1$-preimage of $x_i \atom{a} x_{i+1}$ by an $a$-path of length $l$ obtaining some $\delta_{+l}$ such that $\anexpansion''_{+l} \surj \delta_{+l} \homto \anexpansion'_{+l}$. Further, $\delta_{+l} \in \Tw$ since $\Tw$ is closed under "refinements" by \Cref{fact:refinement-tw}.
    In both cases this shows that $\anexpansion'_{+l}$ is \emph{not} a counterexample. 
	In particular, for $l=m-n-1$, we have $\anexpansion'_{+l}=\anexpansion$, and this would contradict the fact that $\anexpansion$ is a counterexample.
	Therefore, there exists a counterexample of polynomial (quadratic) size whenever $\gamma \not\contained \MUAHom{\gamma}{\Tw}$.
  \end{proof}
  The second ingredient is that testing whether a "CQ" is a counterexample is in "coNP".

  \begin{claim}\AP\label{claim:cq-in-Qapp-np}
    The problem of testing, given a "C2RPQ" $\gamma$ and a "CQ" $\xi$, whether $\xi \contained \Qapp$, is in "NP".
  \end{claim}
  \begin{proof}
        We first guess a polynomial-sized "$k$-summary query" $\delta_{\textup{zip}}$ and test in "NP" that it is part of $\Qapp$ by \Cref{lemma:MUAk-union-of-polysized-summary-test-np}. 
        Let us call $\Delta$ be the equivalent {\UCRPQSRE} query, given by \Cref{lemma:MUAk-union-of-polysized-summary-test-np} \textit{cum} \Cref{lemma:bound_size_refinements}.
        We have to check that there is some "expansion" $\delta$ of $\Delta$ such that there is a "homomorphism" $\delta \homto \xi$. 
        We first guess a valuation $\mu :  \vars(\delta_{\textup{zip}}) \to \vars(\xi)$.
        Now it remains to check that:
        \begin{enumerate}
          \item For every "CRPQ" atom $x \atom{a^*} y$ of $\delta_{\textup{zip}}$ there is an $a$-path in $\xi$ from $\mu(x)$ to $\mu(y)$.
          \item Every "path-$l$ approximation" $\pathl(X, Y, \bigwedge_{1 \leq i \leq n} A_i(x_i,y_i))$ of $\delta_{\textup{zip}}$ contains a "CQ" $\delta_{\textup{path}}$  admitting a "path decomposition" of "width" $l$ which starts with the bag $X$ and ends with $Y$. And further, there is a homomorphism $h: \delta_{\textup{path}}  \homto \xi$ which coincides with $\mu$ on variables $X \cup Y$.
        \end{enumerate}
        Observe that  these two properties hold true if, and only if, there is some "expansion" $\delta$ of $\Delta$ such that $\delta \homto \gamma$.
        It is clear the first point can be achieved in polynomial time (actually, in "NL") since it is a simple reachability query. 
        The second point can also be achieved in polynomial time (or in "NL"), since the "fine" "path decomposition" of "width" $l$ can be guessed on-the-fly using $l+1$ pointers to the variables of $\gamma$ ("cf" \Cref{lemma:evaluation-bounded-pathwidth}). An "NL" algorithm can advance down the "path decomposition" while simultaneously
        \begin{enumerate}
          \item guessing the "conjunctive query" $\delta_{\textup{path}}$ via its "fine" "path decomposition" of "width" $k$,
          \item checking that there is a partial "homomorphism" to $\gamma$ ("ie", a "homomorphism" from the subquery of $\gamma$ restricted to current "bag"'s variables to $\gamma$),
          \item ensuring that the "CQ" $\delta_{\textup{path}}$ being built is an element of $\pathl(X, Y, \bigwedge_{1 \leq i \leq n} A_i(x_i,y_i))$, which requires to also guess a homomorphism $\rho \homto \delta_{\textup{path}}$ from a refinement $\rho$ of $\bigwedge_{1 \leq i \leq n} A_i(x_i,y_i)$.
        \end{enumerate}
        Further, a simple test can ensure that the first and last bags of the decomposition  coincide with the guessed assignment $\mu$.
        Since the number of pointers (bounded by $k+1$) is fixed, this subroutine is in "NL", and hence in polynomial time. This yields an "NP" algorithm for testing $\xi \contained \Qapp$.
  \end{proof}

  As a consequence of the two claims, we obtain a \sigmatwo~algorithm for non-containment of $\gamma \contained \Qapp$: We first guess an "expansion" $\anexpansion$ of $\gamma$ of polynomial size, and we then test $\anexpansion \not\contained \Qapp$ in "coNP". This gives a {\pitwo} algorithm for the "semantic tree-width $k$ problem", which is correct by
  \Cref{lemma:MUAk-union-of-polysized-summary-test-np,claim:poly-sized-counterexample-sre}.
\end{proof}

\section{\AP{}Acyclic Queries: the Case $k=1$}
\label{sec:acyclic-queries}

Observe that in the previous sections we have treated the cases of "semantic tree-width" $k$ for every $k\geq 2$. However, the case $k=1$ remains rather elusive so far. 
While the "Key Lemma" holds for $k=1$, it proves the computability of an object that
is irrelevant to study "semantic tree-width" 1, see \Cref{rk:key-lemma-tw1}. 
The problem comes from \Cref{ex:counterex-tw1}, namely that
$\MUA{\gamma}{\Tw[1]} \not\semequiv \MUAHom{\gamma}{\Tw[1]}$.
This is the main obstacle why our approach does not directly yield an algorithm for the case $k=1$, which had been previously solved by Barceló, Romero and Vardi \cite{BarceloRV16}. 
However, as we argue in this section, a rather elegant modification on the notion of "tree-width" allows to use our approach as a unifying framework for both the case $k=1$ and the cases
$k \geq 2$. 
Concretely, we introduce a family of classes $\set{\ContrTw}_k$ such that 
$\MUA{\gamma}{\Tw} \semequiv \MUAHom{\gamma}{\ContrTw}$ for every $\gamma$ and $k$, and where 
$\MUAHom{\gamma}{\ContrTw[1]} \semequiv \MUAHomBounded{\gamma}{\ContrTw[1]}{\leq\textit{poly}(\size{\gamma})}$.
As a corollary, we reprove \cite[Theorem~6.1]{BarceloRV16}, namely that the "semantic tree-width $1$ problem" is \expspace-complete. Further, we can also solve the "one-way semantic tree-width $1$ problem", which is outside the scope of \cite{BarceloRV16}. 
Remember that for $k=1$, the "semantic tree-width@semantic tree-width $1$ problem" and "one-way semantic tree-width $1$ problems@one-way semantic tree-width $1$ problem" are two independent problems, since there are queries of "semantic tree-width" 1 but not of "one-way semantic tree-width" 1 ("cf" \Cref{rk:closure-under-sublanguages-k1}).

\subsection{\AP{}Contracted Tree-Width}
We next formally define the notion of ``"contracted tree-width"'', meaning the "tree-width" of the graph obtained by contracting paths (or directed paths) into edges. This altered notion of "tree-width" will allow us to seamlessly prove the case of $k=1$ for \Cref{thm:decidability-semtw}.

\AP Given a "C2RPQ" $\gamma$, an \AP""internal path"" is a sequence of atoms\footnote{We write
$x \symatom{\lambda} y$ to mean that there is either an "atom" $x \atom{\lambda} y$
\textbf{or} an "atom" $x \coatom{\lambda} y$.}
\[
	x_0 \symatom{L_1} x_1 \symatom{L_2} \cdots \symatom{L_{n-1}} x_{n-1} \symatom{L_n} x_n
\]
where each $x_i$ for $i \in \lBrack 1,n-1 \rBrack$ has total degree
exactly 2 and is existentially quantified.
\AP Its ""contraction@@path"" is defined as the edge
\[
	x_0 \atom{K_1 \cdot K_2 \cdots K_{n-1} \cdot K_n} x_n,
\]
where $K_i \defeq L_i$ if the "atom" between $x_i$ and $x_{i+1}$ is directed from
left to right, and $K_i \defeq L_i^{-}$ if the "atom" is directed from right to left.\footnote{Given a regular language $L$ over $\Aext$, we define a regular language $L^-$ over $\Aext$ by induction on regular expressions: $\emptyset^- \defeq \emptyset$, $(a)^- \defeq a^-$, $(a^-)^- \defeq a$, $(L_1\cdot L_2)^- \defeq L_2^- \cdot L_1^-$, $(L^*)^- \defeq (L^-)^*$ and $(L_1 + L_2)^- \defeq L_1^- + L_2^-$. Then, for any graph, there is a path from $x$ to $y$ labelled by
a word of $L$ "iff" there is a path from $y$ to $x$ labelled by a word of $L^-$.}

\AP Similarly, a ""one-way internal path"" is a sequence of atoms 
\[
	x_0 \atom{L_1} x_1 \atom{L_2} \cdots \atom{L_{n-1}} x_{n-1} \atom{L_n} x_n
\]
where each $x_i$ for $i \in \lBrack 1,n-1 \rBrack$ has exactly in-degree
and out-degree 1 in $\gamma$ and is existentially quantified.
\AP Its ""one-way contraction@@path"" is defined as the edge
\[
	x_0 \atom{L_1 \cdot L_2 \cdots L_{n-1} \cdot L_n} x_n.
\]

\AP A ""contraction"" (resp.\ ""one-way contraction"") of a "C2RPQ" is any query obtained
by iteratively replacing some "internal paths" (resp.\ "one-way internal paths") by
their "contraction@@path" (resp.\ "one-way contraction@@path"). By definition,
a query is always "equivalent" to any of its "contractions" or "one-way contractions".

\begin{definition}
	\AP
	Define the ""contracted tree-width"" (resp.\ ""one-way contracted tree-width"")
	of a "C2RPQ" as the minimum of the "tree-width" among its "contractions" (resp.\ of its "one-way contractions"). Let $\intro*\ContrTw$ and $\intro*\ContrTwOneWay$ be, respectively, the set of all "C2RPQs" of "contracted tree-width" at most $k$ and of "CRPQs" of "one-way contracted tree-width" at most $k$.
\end{definition}

For instance, the query
\begin{center}
	\small
	\begin{tikzcd}[column sep=small, row sep=small]
		&[-1em] x_0 \ar["K", rr] \ar["L" swap, ddr] & & x_1 \ar["M", ddl] \\[-1em]
		\gamma(x_0, x_1) \defeq & \\[-1em]
		& & y &
	\end{tikzcd}
\end{center}
has "contracted tree-width" one since the "internal path" $x_0 \atom{L} y \coatom{M} x_1$
can be "contracted" into $x_0 \atom{LM^{-}} x_1$. On the other hand,
its "one-way contracted tree-width" is two, since there is no non-trivial "one-way internal path"
as $x_1$ is an "output variable".

Note that, by definition:
\begin{itemize}
	\item the "contracted tree-width" is at most the "one-way contracted tree-width", which is
		in turn at most the "tree-width";
	\item for $k\geq 2$, these notions collapse (by \Cref{fact:refinement-tw});
	\item for $k=1$, both inequalities can be strict.
\end{itemize}
Moreover, for any $k\geq 1$,
"contracted tree-width" at most $k$ and "one-way contracted tree-width" 
at most $k$ are both closed under "refinements": if a query has "tree-width" at most $k$, so does any "refinement" thereof. In fact, the "CQs" of "contracted tree-width" $1$ precisely correspond to what in \cite[\S 5.2.1, p1358]{BarceloRV16} is known as ``pseudoacyclic graph databases''.

\begin{fact}
	\AP\label{fact:tw-equiv-to-ctw}
	Let $k \geq 1$. For any "CRPQ" $\gamma$, we have
	$\MUA{\gamma}{\TwOneWay} \semequiv \MUAHom{\gamma}{\ContrTwOneWay}$.\\
	Moreover, for any "C2RPQ" $\gamma$,
	$\MUA{\gamma}{\Tw} \semequiv \MUAHom{\gamma}{\ContrTw}$.
\end{fact}
\begin{proof}
	\begin{align*}
		\MUA{\gamma}{\Tw}
			& \semequiv \MUA{\gamma}{\ContrTw}
			\quad\text{since "contractions" preserve semantics,} \\
			& \semequiv \MUAHom{\gamma}{\ContrTw}
			\quad\text{by \Cref{obs:equivalence_under_approx_homomorphism}.}
	\end{align*}
	The same arguments work with "one-wayness".
\end{proof}

\subsection{\AP{}The Key Lemma for Contracted Tree-Width One}
We show next that "contracted tree-width" $1$ has all the needed properties for the analogue of "Key Lemma" for $k=1$ to hold.

\begin{lemma}
    \AP\label{lemma:bound_size_refinements_tw}
    \AP For any "CRPQ" $\gamma$, we have
    $\MUAHom{\gamma}{\ContrTwOneWay[1]} \semequiv \MUAHomBounded{\gamma}{\ContrTwOneWay[1]}{\leq\lOne}$, where
    $\intro*\lOne = \Theta(\nbatoms[3]{\gamma})$.
	Similarly, for a "C2RPQ" $\gamma$,
	$\MUAHom{\gamma}{\ContrTw[1]} \semequiv
		\MUAHomBounded{\gamma}{\ContrTw[1]}{\leq\lOne}$.
\end{lemma}

\begin{proof}
	Consider the proof of the "Key Lemma" (\Cref{lemma:bound_size_refinements}).
	We claim that:
	\begin{enumerate}
		\itemAP the constructions of
			\Cref{lemma:locally_acyclic_treedec,lemma:shape-decomposition}
			both preserve "contracted tree-width" at most 1 and "one-way contracted tree-width" at most 1;
			\label{eq:lemma:bound_size_refinements_tw:1}
		\itemAP the upper bound
			$\l \in \+O(\nbatoms[2]{\gamma} \cdot 2^{\nbatoms{\gamma}})$
			can be easily boiled down to $\lOne \in \+O(\nbatoms[3]{\gamma})$ in the special case of $k=1$.
			\label{eq:lemma:bound_size_refinements_tw:2}
	\end{enumerate}

	\smallskip

	\proofcase{\eqref{eq:lemma:bound_size_refinements_tw:1} Preservation of "contracted tree-width".}
	We claim that all constructions of \Cref{sec:proof-key-lemma} preserve "contracted tree-width" 
	at most 1. The setting is similar, except that now, a "trio" consists of a triple $\fun\colon \rho \to 
	\alpha$ where $\rho$ is a "refinement" of a fixed "C2RPQ" $\gamma$, and $\alpha$ is a "C2RPQ" of
	"\emph{contracted} tree-width@contracted tree-width" 1. We now apply the constructions not to a decomposition of $\alpha$ but to a "fine tagged tree decomposition" of a \emph{"contraction"} of $\alpha$ of "tree-width" 1.
	
	\begin{fact}
		Let $\gamma$ be a "C2RPQ", $\chi$ be a "contraction" of $\gamma$, and
		$(T, \bagmap, \tagmap)$ be a "fine tagged tree decomposition" of $\chi$ of "width" at most 1.
		Let $z, z' \in \bagmap(b)$ for some "bag" $b \in T$. Then $\gamma \land z \atom{\lambda} z'$
		still has "contracted tree-width" at most 1.
	\end{fact}
	As a consequence, the construction of \Cref{lemma:locally_acyclic_treedec}---which takes us from \Cref{fig:trio-tree-dec} to \Cref{fig:path-induced}---preserves "contracted tree-width" 1.
	Then, \Cref{prop:connecting-tree-decompositions}---illustrated in \Cref{fig:connecting-tree-decompositions}---can be trivially adapted to our setting as follows:
	\begin{fact}
		Let $\gamma$, $\gamma'$ be two "C2RPQ" with a disjoint set of variables.
		Let $z$ (resp.\ $z'$) be a variable of $\gamma$ (resp.\ $\gamma'$).
		If both $\gamma$ and $\gamma'$ have "contracted tree-width" at most 1, then so does
		$\gamma \land \gamma' \land z \atom{\lambda} z'$.
	\end{fact}
	As a consequence, the construction of \Cref{lemma:shape-decomposition}---which takes us from 
	\Cref{fig:path-induced} to \Cref{fig:trio-result}---preserves "contracted tree-width" 1. 

	\medskip

	\proofcase{\eqref{eq:lemma:bound_size_refinements_tw:2} Improved upper bound.} In the proof claiming that in a sufficiently
	long "non-branching path", we can always find two "non-full", "non-atomic bags" with the same 
	"profile" (see the proof of \Cref{lemma:shape-decomposition}), we obtain a bound of
	$\+O(2^{\nbatoms{\gamma}})$. We actually claim that it can be improved to obtain a polynomial 
	bound. This is because, for "width" 1, a "non-full" "bag" $b$ contains exactly 1 variable $z_b$.
	So, its "profile" consists  
	simply on a set of "atoms" of $\gamma$---namely the set of "atoms" whose "refinement" "induces a 
	path" which "leaves" the "bag" $b$ at $z_b$.
	But we claim that in a "non-branching path", not all of these $2^{\nbatoms{\gamma}}$
	"profiles" can occur at the same time. Indeed,
	in "tree decompositions", the set of "bags" containing a given variable must be connected.
	This property can be lifted to paths in "tagged tree decompositions" in the following way.
	\begin{fact}
		\AP\label{fact:paths-are-connected}
		Let $(T, \bagmap, \tagmap)$ be a "tagged tree decomposition"
		of some "homomorphism" $f\colon \rho \homto \alpha$.
		Let $\pi$ be a path in $\rho$. Assume that:
		\begin{itemize}
			\item the simple path in $T$ from $b$ to $b''$ goes through $b'$,
			\item there exists some variable $z$ of $\alpha$ such that $\tagmappath{\pi}$ 
			  	"leaves" $b$ at $z$, and
			\item there is no variable like that for the "bag" $b'$.
		\end{itemize}
		Then, there is no variable $z$ of $\alpha$ such that
		$\tagmappath{\pi}$ "leaves" $b''$ at $z$.
	\end{fact}
	\begin{proof}[Proof of \Cref{fact:paths-are-connected}]
		Fix a "tagged tree decomposition" $(T, \bagmap, \tagmap)$ of some "homomorphism" $f\colon \rho \homto \alpha$ and $\pi$ be a path in $\rho$. Let $b,b',b''$ be "bags" such that
		the simple path in $T$ from $b$ to $b''$ goes through $b'$.
		\AP Say that an "induced path" $\tagmappath\pi = ({b_i \choose z_i})_i$ ""visits"" a "bag" $b$
		if $b_i = b$ for some $i$. Note that this is equivalent to saying that there exists a variable
		$z$ of $\alpha$ "st" $\tagmappath\pi$ "leaves" $b$ at $z$.
		Hence, \Cref{fact:paths-are-connected} boils down to the following property:
		if $\tagmappath\pi$ "visits" both $b$ and $b''$, then it must also "visit" $b'$.
		This property holds because by construction, the sequence $(b_i)_i$---namely the projection of
		$\tagmappath\pi$ onto $T$---is a path in $T$, with some node repetition.
	\end{proof}
	As a consequence,
	if an "atom" occurs in a "bag", but not in a latter one, then it can never occur again.
	Hence, the number of "bags" of size $1$ in a "non-branching path" where each "bag"
	has a different "profile" must be at most $n = \nbatoms{\gamma}$. Hence, \Cref{fact:pigeon-hole}
	yields a bound of $\+O(\nbatoms[2]{\gamma})$.
	Finally, we can conclude like in \Cref{sec:proof-bound-size-refinements}: we obtain a tree
	with at most $\+O(\nbatoms{\gamma})$ leaves, and with "non-branching paths" of length
	at most $\+O(\nbatoms[2]{\gamma})$, so the tree has size at most $\lOne \in \+O(\nbatoms[3]{\gamma})$.
	By "local acyclicity", this concludes the proof of \Cref{lemma:bound_size_refinements_tw}.
	The case of "one-way contracted tree-width" is completely similar.
\end{proof}

\begin{lemma}
	\AP\label{lemma:characterisation-bounded-semantic-tw}
	Let $k \geq 1$.
	\begin{enumerate}
		\item Given a "UCRPQ" $\Gamma$, it has "one-way semantic tree-width" at most $1$
			"iff" $\Gamma \semequiv \MUAHomBounded{\Gamma}{\ContrTwOneWay[1]}{\leq\lOne}$;
		\item Given a "UC2RPQ" $\Gamma$, it has "semantic tree-width" at most $1$
			"iff" $\Gamma \semequiv \MUAHomBounded{\Gamma}{\ContrTw[1]}{\leq\lOne}$;
	\end{enumerate}
	where $\lOne \in \+O(\nbatoms[3]{\gamma})$.
\end{lemma}
\begin{proof}
	To prove the first point:
	\begin{itemize}
		\item if $\Gamma$ is "equivalent" to a "UCRPQ" $\Delta$ of "tree-width" at most $1$, then
			$\Delta \subseteq \MUA{\Gamma}{\Tw[1]}$ and by \Cref{fact:tw-equiv-to-ctw,lemma:bound_size_refinements_tw},
			$\Delta \contained \MUAHomBounded{\Gamma}{\ContrTwOneWay[1]}{\leq\lOne}$,
			and hence:
			\[
				\Gamma \semequiv \Delta
				\contained \MUAHomBounded{\Gamma}{\ContrTwOneWay[1]}{\leq\lOne}
				\contained \Gamma.
			\]
		\item If $\Gamma \semequiv \MUAHomBounded{\Gamma}{\ContrTwOneWay[1]}{\leq\lOne}$, 
			then $\Gamma$ is equivalent to a "UCRPQ" of "contracted tree-width" at
			most 1, and hence (by "contraction") it is equivalent a "UCRPQ" of "tree-width" at most 1. 
	\end{itemize}
	The second point can be proven similarly.
\end{proof}
\begin{corollary}[Upper bound of \Cref{thm:decidability-semtw} for $k=1$]
	\AP\label{cor:sem-tw-1-pb-exp-c}
	The "semantic tree-width $1$ problem" and "one-way semantic tree-width $1$ problem" are in \expspace.
\end{corollary}
\begin{proof}
	The fact that the "semantic tree-width $1$ problem" is in \expspace is actually the main result of \cite[Theorem~6.1]{BarceloRV16}, but we show how the upper bound follows as a direct corollary of \Cref{lemma:characterisation-bounded-semantic-tw} above.
	Since $\lOne \in \+O(\nbatoms[3]{\gamma})$,
	$\MUAHomBounded{\gamma}{\ContrTw[1]}{\leq\lOne}$ is an exponential union of polynomial sized "C2RPQs", and thus by \Cref{prop:bound-containment-pb} the "containment problem" $\Gamma \contained \MUAHomBounded{\Gamma}{\ContrTw[1]}{\leq\lOne}$ is in \expspace, and so is the "semantic tree-width $1$ problem" (since the converse "containment" $\MUAHomBounded{\Gamma}{\ContrTw[1]}{\leq\lOne} \contained \Gamma$ always holds, "cf" \Cref{rk:uctworpq}).
	The proof for "one-way semantic tree-width $1$ problem" is analogous.
\end{proof}

Lastly, note that we can derive from \Cref{lemma:bound_size_refinements_tw} a characterization of "semantic tree-width" 1 somewhat similar to \Cref{thm:closure-under-sublanguages}.

\begin{corollary}\AP
	\label{coro:charact-semantic-treewidth-1}
	\AP Assume that $\+L$ is "closed under sublanguages". 

	\noindent
    \proofcase{Two-way queries:} For any query $\Gamma \in \UCtwoRPQ(\+L)$, the following are equivalent:
    \begin{enumerate}
        \itemAP $\Gamma$ is "equivalent" to an "infinitary union" of "conjunctive queries"
            of "contracted tree-width" at most $1$;
        \itemAP $\Gamma$ has "semantic tree-width" at most $1$;
        \itemAP $\Gamma$ is "equivalent" to a $\UCtwoRPQ(\+L)$ of "contracted tree-width" at most $1$;
        \itemAP $\Gamma$ is "equivalent" to a $\UCtwoRPQ(\+L')$ of "tree-width" at most $1$,
			where $\+L'$ is the closure of $\+L$ under concatenation and inverses, "ie"
			$\+L'$ is the smallest class containing $\+L$ and such that  if $K, L \in \+L'$
			then $K\cdot L \in \+L'$ and $K^{-} \in \+L'$.
    \end{enumerate}
	
	\noindent
	\proofcase{One-way queries:} Similarly, if $\Gamma \in \UCRPQ(\+L)$, then the following are equivalent:
	\begin{enumerate}
        \itemAP $\Gamma$ is "equivalent" to an "infinitary union" of "conjunctive queries"
            of "one-way contracted tree-width" at most $1$;
        \itemAP $\Gamma$ has "one-way semantic tree-width" at most $1$;
        \itemAP $\Gamma$ is "equivalent" to a $\UCRPQ(\+L)$ of "one-way contracted tree-width" at most $1$;
        \itemAP $\Gamma$ is "equivalent" to a $\UCRPQ(\+L')$ of "tree-width" at most $1$,
			where $\+L'$ is the closure of $\+L$ under concatenation, "ie"
			$\+L'$ is the smallest class containing $\+L$ and such that if $K, L \in \+L'$
			then $K\cdot L \in \+L'$.
    \end{enumerate}
\end{corollary}

Note in particular how point (4) of each characterization reflects that a "UCRPQ" of
"semantic tree-width" 1 can have "one-way semantic tree-width" at least 2---as we showed in
\Cref{rk:closure-under-sublanguages-k1}. More generally, the differences between this
last corollary and \Cref{thm:closure-under-sublanguages} highlight the different combinatorial behaviour that "semantic tree-width" $k$ has, depending on whether $k=1$ or $k > 1$.

\begin{remark}
	Finally, note that results of \Cref{sec:sre,sec:acyclic-queries} can be joined in order to show that the "semantic tree-width 1 problems" are decidable in "PiP2" for "UC2RPQs" over the closure
	under concatenation and inverses of "SREs"  (resp. for "UCRPQs" over the closure under concatenation of "SREs").\footnote{While this the whole class of "UCRPQs" over "SREs" has the same expressivity as "UCRPQs" over the closure under concatenation of "SREs", this is not true if one adds the constraint of having "tree-width" at most 1, see \Cref{coro:charact-semantic-treewidth-1}.}
\end{remark}

\section{\AP{}Semantic Path-Width}
\label{sec:semantic-path-width}

In this section, we extend our results to "path-width". Our motivation lies in the fact that "UC2RPQs" of bounded "semantic path-width" admit a "paraNL"\footnote{This is the parametrized counterpart of non-deterministic logspace.} algorithm for the 
"evaluation problem"---see \Cref{thm:evaluation-bounded-pathwidth}---to be compared with "FPT" for bounded "semantic tree-width".

\subsection{\AP{}Path-Width of Queries}

Recall that for tree-width, for any $k \geq 2$, we proved that a "CRPQ" is equivalent to a
finite union of "C2RPQs" of "tree-width" at most $k$ "iff" it is equivalent to finite union of
"CRPQs" of "tree-width" at most $k$ (\Cref{thm:closure-under-sublanguages}). In other words, "two-way navigation" does not help to minimize further the "semantic tree-width" of a query that does not use "two-way navigation". This property does not hold for $k=1$
(\Cref{rk:closure-under-sublanguages-k1}). We show in \Cref{rk:path-width:oneway-vs-twoway} that it also does not hold for
"path-width", no matter the value of $k \geq 1$.

This motivates the following two definitions:
\begin{itemize}
	\itemAP the ""semantic path-width"" of a "UC2RPQ" is the minimal "path-width"
	of a "UC2RPQ" equivalent to it
	\itemAP the ""one-way semantic path-width"" of a "UCRPQ" is the minimal "path-width"
		of a "UCRPQ" equivalent to it.
\end{itemize}
For a given "UCRPQ", the two natural numbers are well-defined, and the former is always less or equal to
the letter.
The \AP ""semantic path-width $k$ problems"" ask, given a "UCRPQ" (resp.\ "UC2RPQ"), if it has
"semantic path-width" (resp.\ "one-way semantic path-width") at most $k$.

In this section, we first show that the "semantic path-width $k$ problems"
are decidable (\Cref{thm:decidability-sempw}), and then 
after showing that evaluation of "UC2RPQs" of bounded "path-width" is "NL" (\Cref{lemma:evaluation-bounded-pathwidth}) we deduce that for the "evaluation problem" for
"UC2RPQs" of bounded "semantic path-width" (in particular, this captures the
case of "UCRPQs" of bounded "one-way semantic path-width") is in
"para-NL" when parametrized in the size of the query (\Cref{thm:evaluation-bounded-pathwidth}).

\subsection{\AP{}Deciding Bounded Semantic Path-Width}

The key (implicit) ingredient in the proof of \Cref{thm:closure-under-sublanguages,cor:mua-exists-effective} is that "tree-width" at most $k$ is closed under "expansions" (\Cref{fact:refinement-tw}). Unfortunately, this property fails for "path-width".

\begin{figure}[tbp]
	\centering

	\begin{subfigure}[b]{.45\textwidth}
		\centering
		\includegraphics[width=.85\textwidth]{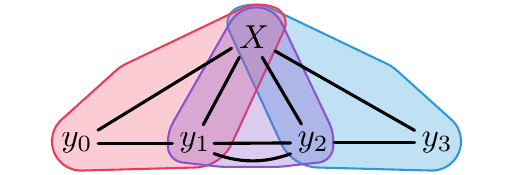}
		\caption{
			\AP\label{subfig:pw-not-closed-original}
			Graph $\+G_k$ with a "path decomposition" of "width" $k$, with three "bags".
		}
	\end{subfigure}
	\hspace{1.5em}
	\begin{subfigure}[b]{.45\textwidth}
		\centering
		\includegraphics[width=.85\textwidth]{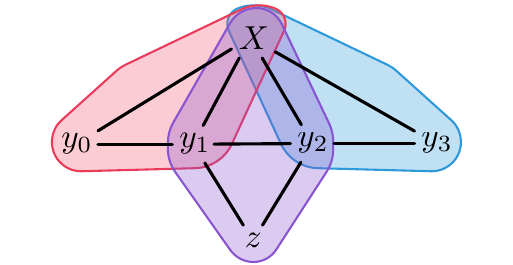}
		\caption{
			\AP\label{subfig:pw-not-closed-path-dec-horizontal}
			"Expansion" $\+G'_k$ of $\+G_k$ with a "path decomposition" of "width" $k+1$, with three "bags".
		}
	\end{subfigure}

	\bigskip

	\begin{subfigure}[b]{.45\textwidth}
		\centering
		\includegraphics[width=.85\textwidth]{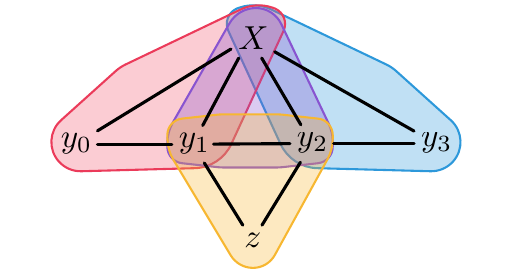}
		\caption{
			\AP\label{subfig:pw-not-closed-tree-dec}
			"Tree decomposition" of $\+G'_k$ of "width" $k$, with four "bags".
		}
	\end{subfigure}
	\hspace{1.5em}
	\begin{subfigure}[b]{.45\textwidth}
		\centering
		\includegraphics[width=.85\textwidth]{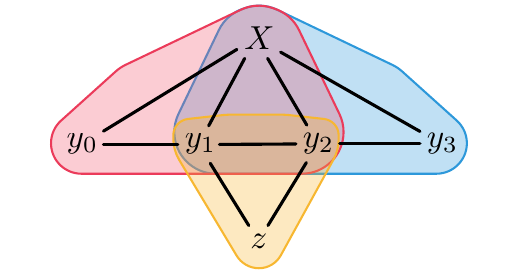}
		\caption{
			\AP\label{subfig:pw-not-closed-path-dec-vertical}
			"Path decomposition" of $\+G'_k$ of "width" $k+1$, with three "bags".
		}
	\end{subfigure}
	\caption{
		\AP\label{fig:pw-not-closed}
		The class of multigraphs with "path-width" at most $k \geq 1$ is not closed under "expansions":
		illustration of a multigraph $\+G_k$ of "path-width" $k$ (\Cref{subfig:pw-not-closed-original})
		whose "expansion" $\+G'_k$ has "path-width" $k+1$ (\Cref{subfig:pw-not-closed-path-dec-horizontal,subfig:pw-not-closed-path-dec-vertical})---but
		"tree-width" $k$ (\Cref{subfig:pw-not-closed-tree-dec}). Set $X$ represents a $(k-1)$-clique.
	}
\end{figure}
\begin{restatable}{fact}{pathwidthnotclosed}
	\AP\label{fact:pw-not-closed}
	For each $k \geq 1$, the class of graphs of "path-width" at most $k$ is not
	closed under "expansions".
\end{restatable}

The counterexample is illustrated in \Cref{fig:pw-not-closed}. A formal proof can be found
in \Cref{apdx-sec:path-width-not-closed-refinements}.
\begin{restatable}{rem}{onewayVsTwowayPw}
	\AP\label{rk:path-width:oneway-vs-twoway}
	Contrary to the case of "semantic tree-width", for every $k$ there are "CRPQs" which are of "semantic path-width" $k$ but not of "one-way semantic path-width" $k$.
\end{restatable}
\begin{proof}
	Indeed, let
	\begin{align*}
		\gamma_k(\bar x, \bar y) \defeq\hspace{1em}
		\bigl( & \bigwedge_{1 \leq i < j \leq k-1} x_i \atom{a} x_j \bigr)
		\;\land\; \bigl(\bigwedge_{1 \leq i \leq k-1} \bigwedge_{0 \leq j \leq 3} x_i \atom{b} y_i \bigr) \\
		\land\; \bigl( & \bigwedge_{0 \leq j < 3} y_i \atom{c} y_{i+1}\bigr)
		\;\land\; y_1 \atom{d} z \;\land\; y_2 \atom{e} z,
	\end{align*}
	whose underlying graph corresponds to \Cref{subfig:pw-not-closed-path-dec-horizontal}.
	Observe that it is a "core" and that only $z$ is existentially quantified.
	Then in $\gamma_k(\bar x, \bar y)$, one can replace the two "atoms"
	$y_1 \atom{d} z \;\land\; y_2 \atom{e} z$ by $y_1 \atom{d\, e^{-}} y_2$, while preserving the semantics. The underlying graph of this new query being \Cref{subfig:pw-not-closed-original}, it shows that $\gamma_k$ has "semantic path-width" $k$.

	Finally, we claim that $\gamma_k$ has "one-way semantic path-width" $k+1$.
	The upper bound follows from \Cref{subfig:pw-not-closed-path-dec-vertical}.
	For the lower bound, consider a "UCRPQ" $\Delta_k(\bar x, \bar y)$ such that
	$\gamma_k \semequiv \Delta_k$. Since $\gamma_k$ is a "CQ", the "equivalence" implies
	that there exists an "expansion" $\xi$ of a "CRPQ" of $\Delta_k$ such $\gamma_k$ and
	$\xi$ are homomorphically equivalent. Since $\gamma_k$ is a "core", it follows that
	$\xi$ contains it as a subgraph. Hence, the underlying directed multigraph
	of the "CRPQ" in $\Delta_k$ from which $\xi$ originated
	must contain a "one-way contraction" of $\xi$ as a subgraph.
	But the only "one-way contraction" of $\xi$ is itself, and so it follows that
	at least one "CRPQ" in $\Delta_k$ contains the underlying graph of $\xi$ as a subgraph.
	Therefore, $\Delta_k$ has "path-width" at least $k+1$, which concludes the proof
	that the "one-way semantic path-width" of $\gamma_k$ is at least (and hence exactly) $k+1$.
\end{proof}

As done for "contracted tree-width", we define "contracted path-width".
\begin{definition}
	\AP
	Define the ""contracted path-width"" (resp.\ ""one-way contracted path-width"")
	of a "C2RPQ" as the minimum of the "path-width" among its "contractions" (resp.\ of its "one-way contractions"). Let $\intro*\ContrPw$ and $\intro*\ContrPwOneWay$ be, respectively, the set of all "C2RPQs" of "contracted path-width" at most $k$ and of "CRPQs" of "one-way contracted path-width" at most $k$.
\end{definition}
The statements and proofs of this section are analogous to the ones of \Cref{sec:acyclic-queries} in the context of "contracted tree-width" 1. We keep the order and structure to make this correspondence evident.

Again, by definition, "contracted path-width" at most $k$ and "one-way contracted path-width" 
at most $k$ are both closed under "refinements": if a query has width at most $k$, so does any "refinement" thereof.

\begin{fact}
	\AP\label{fact:pw-equiv-to-cpw}
	Let $k \geq 1$. For any "CRPQ" $\gamma$, we have
	$\MUA{\gamma}{\PwOneWay} \semequiv \MUAHom{\gamma}{\ContrPwOneWay}$.\\
	Moreover, for any "C2RPQ" $\gamma$,
	$\MUA{\gamma}{\Pw} \semequiv \MUAHom{\gamma}{\ContrPw}$.
\end{fact}

\begin{proof}
	\begin{align*}
		\MUA{\gamma}{\PwOneWay}
			& \semequiv \MUA{\gamma}{\ContrPwOneWay}
			\quad\text{since "contractions" preserves the semantics,} \\
			& \semequiv \MUAHom{\gamma}{\ContrPwOneWay}
			\quad\text{by \Cref{obs:equivalence_under_approx_homomorphism}.}
	\end{align*}
	The same arguments work for "C2RPQs".
\end{proof}

\begin{lemma}
    \AP\label{lemma:bound_size_refinements_pw}
    \AP For $k \geq 1$ and "CRPQ" $\gamma$, we have
    $\MUAHom{\gamma}{\ContrPwOneWay} \semequiv \MUAHomBounded{\gamma}{\ContrPwOneWay}{\leq\l}$, where
    $\l = \lbound{k}{\gamma}$.
	Similarly, for a "C2RPQ" $\gamma$,
	$\MUAHom{\gamma}{\ContrPw} \semequiv \MUAHomBounded{\gamma}{\ContrPw}{\leq\l}$.
\end{lemma}

\begin{proof}
	Consider the proof of the "Key Lemma" (\Cref{lemma:bound_size_refinements}):
	both constructions (\Cref{lemma:locally_acyclic_treedec,lemma:shape-decomposition}) preserve the "contracted path-width" of $\alpha$ if the operations are applied
	to a suitable "path decomposition" of a "contraction" of $\alpha$ of width $k$.
\end{proof}

\begin{lemma}
	\AP\label{lemma:characterisation-bounded-semantic-pw}
	Let $k \geq 1$.
	\begin{enumerate}
		\item Given "UCRPQ" $\Gamma$, it has "one-way semantic path-width" at most $k$
			"iff" $\Gamma \semequiv \MUAHomBounded{\Gamma}{\ContrPwOneWay}{\leq\l}$;
		\item Given a "UC2RPQ" $\Gamma$, it has "semantic path-width" at most $k$
			"iff" $\Gamma \semequiv \MUAHomBounded{\Gamma}{\ContrPw}{\leq\l}$.
	\end{enumerate}
\end{lemma}

\begin{proof}
	To prove the first point:
	\begin{itemize}
		\item if $\gamma$ is "equivalent" to a "UCRPQ" $\Delta$ of "path-width" at most $k$, then
			$\Delta \subseteq \MUA{\gamma}{\Pw}$ and by \Cref{fact:pw-equiv-to-cpw,lemma:bound_size_refinements_pw},
			$\Delta \contained \MUAHomBounded{\gamma}{\ContrPwOneWay}{\leq\l}$,
			and hence:
			\[
				\gamma \semequiv \Delta
				\contained \MUAHomBounded{\gamma}{\ContrPwOneWay}{\leq\l}
				\contained \gamma.
			\]
		\item If $\gamma \semequiv \MUAHomBounded{\gamma}{\ContrPwOneWay}{\leq\l}$, 
			then $\gamma$ is equivalent to a "UCRPQ" of "contracted path-width" at
			most $k$, and hence it is equivalent a "UCRPQ" of "path-width" at most $k$. 
	\end{itemize}
	The second point can be proven similarly.
\end{proof}

We can now prove the main theorem.
\decidabilitySemPw

\begin{proof}
	The upper bounds follow from \Cref{lemma:characterisation-bounded-semantic-pw}.
	The lower bounds will be shown in \Cref{lemma:lowerbound}. Lastly, to prove the "ExpSpace" upper bound
	for $k=1$, we can apply the same trick as in \Cref{cor:sem-tw-1-pb-exp-c}.
\end{proof}

Similarly to \Cref{coro:charact-semantic-treewidth-1}, we can derive from
\Cref{lemma:bound_size_refinements_pw} a characterization of "semantic path-width" at most $k$.
\begin{corollary}
	\AP\label{coro:charact-semantic-pathwidth-k}
	\AP Assume that $\+L$ is "closed under sublanguages", and let $k \geq 1$. 

	\noindent
    \proofcase{Two-way queries:}
    For any query $\Gamma \in \UCtwoRPQ(\+L)$, the following are equivalent:
    \begin{enumerate}
        \itemAP $\Gamma$ is "equivalent" to an "infinitary union" of "conjunctive queries"
            of "contracted path-width" at most $k$;
        \itemAP $\Gamma$ has "semantic path-width" at most $k$;
        \itemAP $\Gamma$ is "equivalent" to a $\UCtwoRPQ(\+L)$ of "contracted path-width" at most $k$;
        \itemAP $\Gamma$ is "equivalent" to a $\UCtwoRPQ(\+L')$ of "path-width" at most $k$,
			where $\+L'$ is the closure of $\+L$ under concatenation and inverses, "ie"
			$\+L'$ is the smallest class containing $\+L$ and such that  if $K, L \in \+L'$
			then $K\cdot L \in \+L'$ and $K^{-} \in \+L'$.
    \end{enumerate}

	\noindent
    \proofcase{One-way queries:}
	Similarly, if $\Gamma \in \UCRPQ(\+L)$, then the following are equivalent:
	\begin{enumerate}
        \itemAP $\Gamma$ is "equivalent" to an "infinitary union" of "conjunctive queries"
            of "one-way contracted path-width" at most $k$;
        \itemAP $\Gamma$ has "one-way semantic path-width" at most $k$;
        \itemAP $\Gamma$ is "equivalent" to a $\UCRPQ(\+L)$ of "one-way contracted path-width" at most $k$;
        \itemAP $\Gamma$ is "equivalent" to a $\UCRPQ(\+L')$ of "path-width" at most $k$,
			where $\+L'$ is the closure of $\+L$ under concatenation, "ie"
			$\+L'$ is the smallest class containing $\+L$ and such that if $K, L \in \+L'$
			then $K\cdot L \in \+L'$.
    \end{enumerate}
\end{corollary}

\subsection{\AP{}Evaluation of Queries of Bounded Semantic Path-Width}

In this section, we show that, as a consequence of \Cref{thm:decidability-sempw}, we can obtain an
efficient algorithm for the "evaluation problem".
\paraNLEvalBoundedSemPathWidth

The class "para-NL" was introduced in \cite[Definition, p. 123]{CCDF97Advice} under the
name ``uniform "NL" + advice''. It was renamed "para-NL" in
\cite[Definition 1, p. 294]{FG03Describing}. For the sake of simplicity,
instead of either of those definitions, we use the characterization of "para-NL"
proven in \cite[Theorem 4, p. 296]{FG03Describing}.

We define \AP""para-NL"" as the class of parametrized languages $L \subseteq \Sigma^* \times \Nat$
for which there is a Turing machine $M$ "st"
\begin{center}\begin{tikzpicture}[>={Classical TikZ Rightarrow}]
	\node (first) {%
		$M \text{ accepts } \langle \textcolor{cRed}{x},\, \textcolor{cBlue}{k}\rangle%
		\quad\text{"iff"}\quad
		\langle$%
	};
	\node[inode, cRed, right=-3.5pt of first.base east, anchor=base west](toX){$x$};
	\node[inode, right=0pt of toX.base east, anchor=base west](comma){$,\,$};
	\node[inode, cBlue, right=0pt of comma.base east, anchor=base west](toK){$k$};    
	\node[inode, right=-.5pt of toK.base east, anchor=base west]{$\rangle \in L,$};
	\node[cRed, below left= .3cm and .2cm of toX](fromX){input};
	\node[cBlue, below right= .3cm and .2cm of toK](fromK){parameter};
	\draw[->, cRed] (fromX) to[out=60, in=-120] (toX.south);
	\draw[->, cBlue] (fromK) to[out=120, in=-60] (toK.south);
\end{tikzpicture}\end{center}
and, moreover, $M$ runs in non-deterministic space
$f(\textcolor{cBlue}{k}) + \+O(\log(|\textcolor{cRed}{x}|))$,
where $f\colon \Nat \to \Nat$ is a computable function.
A typical example of "para-NL" problem is the model-checking problem for first-order logic on
finite structures, when parametrized by the maximum degree of the structure
\cite[Example 6]{FG03Describing}. 

To show \Cref{thm:evaluation-bounded-pathwidth}, we first focus on the evaluation of queries of bounded "path-width".
\begin{lemma}
	\AP\label{lemma:evaluation-bounded-pathwidth}
	For each $k \geq 1$, the "evaluation problem", restricted to "UC2RPQs" of
	"path-width" at most $k$, is "NL"-complete.
\end{lemma}

\begin{figure}
	\centering
	\begin{subfigure}{.45\linewidth}
		\centering
		\includegraphics[scale=.7]{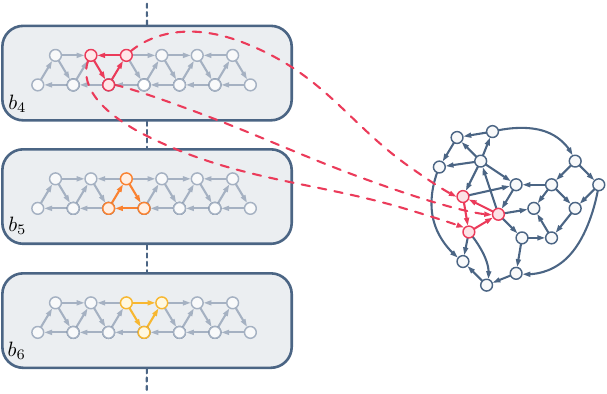}
		\caption{
			Partial homomorphism computed at the fourth step of the algorithm.
		}
	\end{subfigure}
	\hfill
	\begin{subfigure}{.45\linewidth}
		\centering
		\includegraphics[scale=.7]{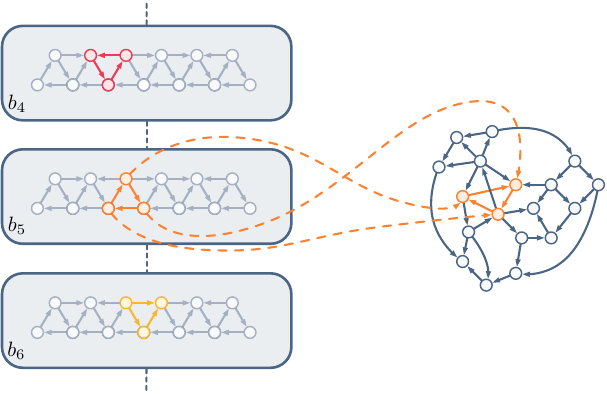}
		\caption{
			Partial homomorphism computed at the fifth step of the algorithm.
		}
	\end{subfigure}
	\caption{
		\AP\label{fig:eval-pw-nl}
		Algorithm to evaluate "C2RPQs" of "path-width" at most $k$ (here $k=2$)
		in non-deterministic logspace. Each subfigure represents: (i) on the left-hand side, a
		"path decomposition" of a "C2RPQ", (ii) on the right-hand side, a "graph database",
		and (iii) a partial "homomorphism" from the "C2RPQ" to the "database".
	}
\end{figure}

\begin{proof}
	\proofcase{Lower bound.} "NL"-hardness direected follows from the "NL"-hardness of
	the reachability problem in directed graphs, see "eg" \cite[Theorem~4.18, p.~89]{AB09Computational}.

	\proofcase{Upper bound, first part: with the "path decomposition".}
	First, we assume that a "tagged path decomposition" of width at most $k$ of the query of
	is also provided as part of the input. Moreover, we assume "wlog" that the input is "C2RPQ"---the extension to "UC2RPQ" being straightforward. So, we are given as input:
	\begin{itemize}
		\item a database together with a tuple of nodes $G(\bar u)$,
		\item a "C2RPQ" $\gamma(\bar x)$, and
		\item a "tagged path decomposition" $\langle T, \bagmap, \tagmap \rangle$ of "width" at 	
			most $k$ of $\gamma(\bar x)$.
	\end{itemize}
	The algorithm, illustrated in \Cref{fig:eval-pw-nl}, maintains a partial homomorphism
	$f\colon X \pto G$
	of these variables onto $G$.
	We scan the "bags" of the decomposition from left to right.
	\begin{itemize}
		\item Initially---before even scanning the first "bag"---$f$ is the map with empty domain.
		\item Then, when scanning the $i$-th bag $b_i$,
			we start by restricting $f$ to variables of $\dom(f)\cap b_i$.
			Then, we extend $f$ so that it is defined on the whole "bag" $b_i$.
			For every variable $y$ in $b_i \smallsetminus\dom(f)$:
			\begin{itemize}
				\item if it belongs to $\bar x$, say $y = x_i$, 
					we let $f(x_i) \defeq u_i$;
				\item otherwise, we non-deterministically guess the value of $f(y)$.
			\end{itemize}
			We then check, for every "atom" $x \atom{L} y$ of $\gamma(\bar x)$ which is "tagged" in 
			the current "bag" $b_i$ if there is a path from $f(x)$ to $f(y)$ labelled by a word of $L$ in $G$. If not, we reject. 
	\end{itemize}
	If the algorithm manages to scan the whole bag decomposition without rejecting, it accepts.

	Completeness of the algorithm is trivial. Correctness follows from the fact that if
	a variable occurs in "bags" $b_i$ and $b_k$ with $i \leq k$, then it must also belong to every
	"bag" $b_j$ for $j \in \lBrack i,k \rBrack$. As a consequence, a variable $x$ is assigned 
	exactly one value $f(x)$ during the whole process. 

	Concerning the space complexity:
	\begin{itemize}
		\item By construction, at the $i$-th step of the algorithm, $f$ is defined exactly on $b_i$, so on at most $k+1$ variables.
		So, $f$ can be stored in space $(k+1)\log(|G|)$;
		\item we need a counter with $\log(|T|)$ bits to scan through the "tagged path decomposition",
		\item each atomic check---checking if there is an $L$-labelled path from $f(x)$ to $f(y)$---can be done in non-deterministic space $c\cdot (\log(|G|) + \log(|\+A_L|))$,
		where $\+A_L$ is an NFA for $L$, using a straightforward adaptation of the "NL" algorithm for graph connectivity; note that these atomic checks are
		independent of one another, so we can reuse this space.
	\end{itemize}
	Overall, the algorithm runs in non-deterministic space
	$\+O(k\log(|G|) + \log(|T|)) = \+O(\log(|G|) + \log(|T|))$,
	which is logarithmic in the size of the input.

	\smallskip

	\proofcase{Upper bound, second part: without the "path decomposition".}
	Then, we claim that the original problem---when the "tagged tree decomposition" is not part of 
	the input---also lies in "NL". This is because one can compute, from $\gamma$, a
	"path decomposition" in (deterministic) logarithmic space by\footnote{This result is an 
	adaption of a similar statement for "tree-width" \cite[Theorem I.1, p. 143]{EJT10Logspace}. Note that the promise that the query has bounded "path-width"---in fact bounded "tree-width" suffices---in a crucial assumption of  \cite[Theorem I.1, p. 143]{EJT10Logspace}.} 
	\cite[Theorem 1.3, p. 2]{KM10Computing}. Then, a "path decomposition" can be turned into a 
	"tagged path decomposition" in (deterministic) logarithmic space by tagging an "atom"
	$x \atom{L} y$ in the first bag containing both $x$ and $y$.
	The conclusion follows since functions computable in non-deterministic logarithmic space
	are closed under composition \cite[Lemma 4.17, p. 88]{AB09Computational}. 
\end{proof}

We can now conclude with \Cref{thm:evaluation-bounded-pathwidth}, namely that the "evaluation problem" for "UC2RPQs" of
"semantic path-width" $k$ is in "para-NL".
\begin{proof}[Proof of \Cref{thm:evaluation-bounded-pathwidth}.]
	Given a "UC2RPQ"
	$\Gamma(\bar x)$ of "semantic path-width" at most $k$ and a database $G(\bar u)$, we
	first compute $\MUAHomBounded{\Gamma}{\ContrPwOneWay}{\leq\l}$---where $\l = \lbound{k}{\Gamma}$---, which is "equivalent" to
	$\Gamma$ by \Cref{lemma:characterisation-bounded-semantic-pw}.
	Then, we use \Cref{lemma:evaluation-bounded-pathwidth} to evaluate each
	$\delta(\bar x) \in \MUAHomBounded{\Gamma}{\ContrPwOneWay}{\leq\l}$ on $G(\bar u)$.
	If one of the queries accepts, we accept. Otherwise, we reject.

	The non-deterministic space needed by the algorithm is:
	\begin{itemize}
		\item $\+O(\l)$ bits to enumerate and store $\delta(\bar x)$, where $\l = \lbound{k}{\Gamma} $ 
		\item $\+O(\log{|G|} + \log{\size{\delta}}) \subseteq \+O(\log{|G|} + \log{|\l|} + \log{\size{\Gamma}})$ to evaluate
		$\delta(\bar x)$ on $G(\bar u)$, by \Cref{lemma:evaluation-bounded-pathwidth}
		and since $\size{\delta} \leq \size{\Gamma}\cdot \l$.
	\end{itemize}
	Overall, we use non-deterministic space $f(\size{\Gamma}) + \+O(\log(|G|))$
	where $f$ is a single exponential, which concludes the proof.
\end{proof}

\section{\AP{}Lower Bounds for Deciding Semantic Tree-Width and Path-Width}
\label{sec:lowerbound}

An "ExpSpace" lower bound follows by a straightforward adaptation from the
"ExpSpace" lower bound for the case $k=1$ \cite[Proposition 6.2]{BarceloRV16}.
\begin{restatable}[Lower bound of \Cref{thm:decidability-semtw}]{lem}{lowerbound}
    \AP\label{lemma:lowerbound}
    For every $k\geq 1$, the following problems are "ExpSpace"-hard, even if restricted to "Boolean CRPQs":
	\begin{itemize}
		\item the "semantic tree-width $k$ problem";
		\item the "one-way semantic tree-width $k$ problem";
		\item the "semantic path-width $k$ problem";
		\item the "one-way semantic path-width $k$ problem".
	\end{itemize}
\end{restatable}

\AP We say that a "C2RPQ" is ""connected"" when its underlying undirected graph is connected.
We first give a small useful fact.

\begin{fact}[{Implicit in \cite[Proof of Proposition 6.2]{BarceloRV16}}]
    \AP\label{fact:connectedness}\leavevmode
    \begin{enumerate}
        \item Let $G, G'$ be two "databases" and $\delta$ be a "connected" "Boolean" "C2RPQ".
            If the disjoint union $G \AP\intro*\disjointUnion G'$ ("ie", the union assuming $\vertex{G}$ and $\vertex{G'}$ are disjoint) "satisfies" $\delta$, then either $G$ "satisfies" $\delta$ or
            $G'$ "satisfies" $\delta$.
        \item Let $\gamma,\delta$ be "Boolean" "C2RPQs". If $\delta$ is "connected"
            and $\gamma \contained \delta$, then there exists a "subquery" $\gamma'$ of $\gamma$,
            obtained as connected component of $\gamma$, such that
            $\gamma' \contained \delta$. 
    \end{enumerate}
\end{fact}

Notice first that if $\gamma$ and $\delta$ are "CQs" then the proof of
\Cref{fact:connectedness} follows directly from the equivalence of $\gamma \contained \delta$ (resp.\ $G$ "satisfies" $\gamma$) and the existence of a "homomorphism" from $\delta$ to $\gamma$ (resp.\ $\gamma$ to $G$).

\begin{proof}
	We first prove the second point.
    Write $\gamma \defeq \gamma_1 \land \hdots \land \gamma_n$ where $\gamma_1,\hdots,\gamma_n$
    are connected components of $\gamma$, and assume by contradiction that
    for all $i$, $\gamma_i \not\contained \delta$. Then there exists a database $G_i$
    such that $G_i$ "satisfies" $\gamma_i$ but not $\delta$.
    Consider the disjoint union $G = G_1 \disjointUnion \hdots \disjointUnion G_n$.

    On the one hand, since the $\gamma_i$'s have disjoint variables
    and $G_i$ "satisfies" $\gamma_i$ for each $i$, then $G$ "satisfies" $\gamma$.
    On the other hand, $G$ cannot "satisfy" $\delta$:
    if there was a "homomorphism" from $\delta$ to $G$,
    since $\delta$ is connected, there would exist an index $i$
    such that $\delta$ is mapped on $G_i$, which would contradict the
    fact that $G_i$ does not "satisfy" $\delta$.
    Hence, $G$ does not "satisfy" $\delta$, which contradicts the containment
    $\gamma \contained \delta$.

	To prove the first point, we simply apply the second one, by letting $\gamma$
	be the conjunction of the canonical "CQ" associated with $G$ and $G'$---which
	is in fact the canonical "CQ" associated with $G\disjointUnion G'$. From the assumption
    that $G \disjointUnion G'$ "satisfies" $\delta$ it follows that $\gamma \contained \delta$
    and so, by the first point, there is either a homomorphism from $\delta$
    to $G$ or from $\delta$ to $G'$.
\end{proof}

We can then prove \Cref{lemma:lowerbound}.

\begin{proof}[Proof of \Cref{lemma:lowerbound}]
    Fix $k \geq 1$. We focus on "semantic tree-width", but the exact same reduction works for the other three problems.
    We introduce an intermediate problem, called 
    the \AP""asymmetric containment problem for tree-width $k$"":
    given two "Boolean" "CRPQs" $\gamma$ and $\gamma'$,
    where $\gamma$ has "tree-width" $k$, $\gamma'$ is "connected"  
    and does \emph{not} have "semantic tree-width" $k$,
    it asks whether $\gamma \contained \gamma'$.
    The proof of the lemma then contains two parts:
    \begin{enumerate}
        \item first, we reduce the "asymmetric containment problem for tree-width $k$"
            to the "semantic tree-width $k$ problem",
        \item then, we prove that the "asymmetric containment problem for tree-width $k$" is
            "ExpSpace"-hard.
    \end{enumerate}
    \proofcase{(1):} We reduce the instance $(\gamma, \gamma')$
    of the "asymmetric containment problem for tree-width $k$"
    to the instance $\gamma \land \gamma'$ of the "semantic tree-width $k$ problem". We simply have 
    to check that $\gamma \contained \gamma'$ if and only if 
    $\gamma \land \gamma'$ has "semantic tree-width" $k$. The left-to-right
    implication is straightforward since $\gamma \contained \gamma'$
    implies that $\gamma \land \gamma' \semequiv \gamma$ and $\gamma$ was assumed to
    have "tree-width" $k$.
    For the converse implication, if $\gamma \land \gamma' \semequiv \delta$
    where $\delta$ is a "UC2RPQ" of "tree-width" $k$ then write
    $\delta = \bigvee_{i=1}^n \delta_i$ where the $\delta_i$'s are "C2RPQs"
    and let $\delta_{i,1},\hdots, \delta_{i,k_i}$ be the connected components
    of $\delta_i$.
    
    Since for each $i$ we have
    $\delta_i \contained \delta \semequiv \gamma \land \gamma' \contained \gamma'$, by 
    \Cref{fact:connectedness}, there exists $j_i$ such that
    $\delta_{i,j_i} \contained \gamma'$.
    Let $\delta' \defeq \bigvee_{i=1}^n \delta_{i,j_i}$ so that, by construction
    $\delta' \contained \gamma'$. However, note that $\delta'$ has "tree-width" at most $k$
    but $\gamma'$ was assumed not to have "semantic tree-width" $k$,
    hence $\delta' \strcontained \gamma'$,
    so there exists $G'$ such that:
    \begin{equation}
        \AP\label{eq:subquery-reduction}
        G' \text{ "satisfies" } \gamma'
        \quad\text{ and }\quad
        G' \text{ does not "satisfy" } \delta'.
    \end{equation}
    
    We now prove that $\gamma \contained \gamma'$. Let $G$ be a
    "database" "satisfying" $\gamma$. Then the disjoint union $G \disjointUnion G'$ "satisfies" $\gamma \land \gamma'$
    since $G$ "satisfies" $\gamma$, $G'$ "satisfies" $\gamma'$ and $\gamma$ and $\gamma'$
    are Boolean so we can assume "wlog" that they have disjoint variables.
    As a consequence, $G \disjointUnion G'$ "satisfies" $\delta$ and hence $\delta'$,
    so there exists $i$ such that $G \disjointUnion G'$ "satisfies" $\delta_{i,j_i}$.
    Since $\delta_{i,j_i}$ is "connected", either $G$ "satisfies"
    $\delta_{i,j_i}$ or $G'$ "satisfies" $\delta_{i,j_i}$.
    By \Cref{eq:subquery-reduction}, the latter cannot hold,
    so $G$ "satisfies" $\delta_{i,j_i}$ and hence $\gamma'$.

    Therefore, we have shown that for each "database" $G$ that "satisfies" 
    $\gamma$, then $G$ "satisfies" $\gamma'$, "ie", $\gamma \contained \gamma'$. Overall,
    $\gamma \land \gamma'$ has "semantic tree-width" $k$ if and only if
    $\gamma \contained \gamma'$.

    \medskip

    \proofcase{(2):}
    We now show that the "asymmetric containment problem for tree-width $k$" is
    "ExpSpace"-hard.
    It was shown in \cite[Lemma 8]{Figueira20} that
    the "containment" of "CRPQs" was still "ExpSpace"-hard
    when restricted to inputs of the form:
    \begin{center}
        \begin{tikzcd}[column sep=scriptsize]
            \gamma() = \qvar \rar["K"] & \qvar & \contained^{?} & 
            \qvar \rar["L_1", bend left=40] \rar["\raisebox{6pt}{\small\vdots}", phantom] \rar["L_p" below, bend right=40] & \qvar &[-2.1em] = \delta(),
        \end{tikzcd}
    \end{center}
    where $K,L_1,\hdots,L_p$ are regular languages over $\A$.
    We reduce it to the following problem:
    \begin{center}
        \begin{tikzcd}[column sep=scriptsize, row sep=tiny]
            & & & & & & \qvar \ar[dr, "\#"] \ar[dd, "\#" above left] & \\
            \gamma'() \defeq \qvar \rar["K"] & \qvar \arrow["\#" above, loop above] & \contained^{?} & 
            \qvar \rar["L_1", bend left=40] \rar["\raisebox{6pt}{\small\vdots}", phantom] \rar["L_p" below, bend right=40] & \qvar \rar["\A^*"] & \qvar \ar[rr, "\#" below right] \ar[ur, "\#"] \ar[dr, "\#" below left] & & \qvar = \delta'(). \\
            & & & & & & \qvar \ar[ur, "\#" below right]
        \end{tikzcd}
    \end{center}
    where the right-hand side of $\delta'$ is a directed $(k+2)$-clique
    and where $\#$ is a new symbol, i.e. $\# \not\in \A$.

    We claim that $\gamma \contained \delta$ if and only if $\gamma' \contained \delta'$.
    The forward implication is direct and the converse implication
    simply relies on the fact that $\# \not\in \A$.\footnote{Indeed, the only possible "homomorphisms" from "expansions" of $\delta'$ to "expansions" of $\gamma'$ are the ones sending the "expansions" of "atoms" containing $L_1, \dotsc, L_p$ inside the "expansion" of the "atom" on $K$.}
    Then, observe that $\gamma'$ has "tree-width" $1 \leq k$,
    and that $\delta'$ is connected but do not have "semantic tree-width"
    at most $k$.
    
    To prove the last point, consider a "UC2RPQ" $\Delta''$ that is
    equivalent to $\delta'$. Pick any "expansion" $\xi'_1$ of $\delta'$.
    Since $\Delta'' \contained \delta'$, there exists
    an "expansion" $\xi''$ of $\Delta''$ such that there is a "homomorphism"
    from $\xi'_1$ to $\xi''$. Dually, since $\delta' \contained \Delta''$,
    there exists an "expansion" $\xi'_2$ of $\delta'$ such that there
    is a "homomorphism" from $\xi'' \to \xi'_2$. Overall, we have
    "homomorphisms" $\xi'_1 \to \xi'' \to \xi'_2$.
    Since $\xi'_1$ and $\xi'_2$ are both "expansions" of $\delta'$,
    they contain a $\#$-labelled directed $(k+2)$-clique,
    and the $\#$-letter appears nowhere else.
    Should the "homomorphism" $\xi'_1 \to \xi''$ not be injective,
    $\xi''$ would contain a $\#$-labelled self-loop,
    and hence, the "homomorphism" $\xi'' \to \xi'_2$ would
    yield a $\#$-self loop in $\xi'_2$, which does not exist!
    Hence, the "homomorphism" from $\xi'_1$ to $\xi''$
    is injective on the $(k+2)$-clique.
    As a result, $\xi''$ contains a $(k+2)$-clique
    and has "tree-width" at least $k+1$.
    We conclude that $\Delta''$ has "tree-width" at least $k+1$ by \Cref{fact:refinement-tw}, provided that $k \geq 2$.

    Hence, we have shown that $\gamma \contained \delta$ if and only if $\gamma' \contained \delta'$
    where $\gamma'$ has "tree-width" at most $k$, where $\delta'$ is "connected" and
    has "semantic tree-width" at least $k+1$. Since our reduction
    can be implemented in polynomial time, we conclude that the problems of
    \Cref{lemma:lowerbound} are "ExpSpace"-hard.
\end{proof}

\section{\AP{}Discussion}
\label{sec:discussion}

\subsection{\AP{}Complexity}
\label{sec:discussion-complexity}
We have studied the definability and approximation of "UC2RPQ" queries by queries of bounded "tree-width" and shown that the "maximal under-approximation" in terms of an infinitary union of "conjunctive queries" of "tree-width" $k$ can be always effectively expressed as a "UC2RPQ"
of "tree-width" $k$ (\Cref{cor:mua-exists-effective}). 
However, while the "semantic tree-width $1$ problem" is shown to be \expspace-complete (which was also established in \cite[Theorem 6.1, Proposition 6.2]{BarceloRV16}), we have left a gap between our lower and upper bounds in \Cref{thm:decidability-semtw} for every $k>1$.

\begin{qu}
    For $k > 1$, is the "semantic tree-width $k$ problem" \expspace-complete?
\end{qu}
A related question is whether the "containment problem" between a "C2RPQ" and a "summary query" is in \expspace. Should this be the case, then the "semantic tree-width $k$ problem" would be in \expspace.
We also point out that since every "path-$l$ approximation" can be expressed by a polynomial "UC2RPQ" of tree-width $2k$---this is the same idea as in \cite[Lemma IV.13]{DBLP:conf/lics/0001BV17}---, one can produce, for every "UC2RPQ" $\Delta$ a union $\Gamma$ of poly-sized "C2RPQ" of tree-width $2k$ such that $ \MUA{\Delta}{\Tw} \contained \Gamma \contained \Delta$. This implies that the following ``promise'' problem\footnote{In reference to ``promise constraint satisfaction problems''
\cite[Definition 2.3]{bg2021pcsp}.} is decidable in \expspace: given a "UC2RPQ" $\Gamma$, answer `yes' if $\Gamma$ is of "semantic tree-width" $2k$, and answer `no' if $\Gamma$ is not of "semantic tree-width" $k$. The fact that $\MUA{\Delta}{\Tw}$ can be approximated by an exponential query of tree-width $2k+1$ can also be seen as a corollary of the proof of \cite[Theorem~V.1]{DBLP:conf/lics/0001BV17}.

We also do not know whether the {\pitwo} bound on the "semantic tree-width $k$ problem" for {\UCRPQSRE} has a matching lower bound. The known lower bound for the {\UCRPQSRE} "containment problem" \cite[Theorem~5.1]{FigueiraGKMNT20} does not seem to be useful to be employed in a reduction in this context, since it necessitates queries of arbitrary high "tree-width".

\subsection{\AP{}Characterization of Tractability}
\label{sec:charact-tractability}

Our result implies that for each $k$
the "evaluation problem" for "UC2RPQs" $\Gamma$
of "semantic tree-width" $k$ is fixed-parameter tractable---or \AP""FPT""---taking
the query as parameter, "ie", in time $\+O(|G|^{c} \cdot f(|\Gamma|))$
for a computable function $f$ and constant $c$, where $G$ is the "database" given as input.
While this was a known fact \cite[Corollary IV.12]{DBLP:conf/lics/0001BV17}, the dependence on the database was $c=2k+1$.  Our results show that the dependence can be improved to $c=k+1$, similarly to \cite[Theorem 6.3]{BarceloRV16} for the case $k=1$.
It has been further shown by Feier, Gogacz and Murlak that the "evaluation" can be done with a single-exponential $f$ \cite[Theorem~22]{FGM24}.

In a similar vein, our results show that the "evaluation problem" for "UC2RPQs"
of "semantic path-width" $k$ is in "para-NL". It is unknown whether the semantic bounded width properties characterize all "FPT" and "para-NL" classes.

\begin{qu}\AP\label{qu:paraNLtractability}
    Does every recursively enumerable class of "CRPQs" with "para-NL" evaluation have bounded "semantic path-width"?
\end{qu}

\begin{qu}[Also mentioned in {\cite[\S IV-(4)]{DBLP:conf/lics/0001BV17}}]\AP\label{qu:FPTtractability}
    Does every recursively enumerable class of "CRPQs" with "FPT" evaluation have bounded "semantic tree-width"?
\end{qu}

Note that the classes of bounded "contracted path-width" or "contracted tree-width" are not counterexamples to \Cref{qu:paraNLtractability,qu:FPTtractability}, since
the "path-width" is upper-bounded by one plus the "contracted path-width", and lower-bounded by the "contracted path-width"---and similarly for "tree-width"---and so a width is bounded "iff" its "contracted" variant is bounded.

In the case of "CQs", the answer is `yes' to \Cref{qu:FPTtractability} \cite[Theorem~1]{Grohe07}
under standard complexity-theoretic hypotheses ($W[1] \neq $ "FPT"). For \Cref{qu:paraNLtractability}, the answer is still `yes' \cite[Theorem 3.1]{ChenM13} conditional to a less standard assumption\footnote{By \cite[Theorems 3.1 \& 4.3]{ChenM13}, if the class has bounded "semantic path-width", then the problem is in $\textsf{Path} \subseteq $ "paraNL"; by \cite[Theorems 3.1 \& 5.5]{ChenM13}, if the class does not have bounded "semantic path-width", then the problem is \textsf{Tree}-hard.} (no $\textsf{Tree}$-hard problem is in "paraNL").

However, attempting at answering these questions for "CRPQs" is considerably more challenging. In particular, one important technical difficulty is that a class of "CRPQs" with unbounded "tree-width" may contain queries with no "expansions" which are "maximal" in the sense of "containment". That is, for every $k$, for every query $\gamma$ of "semantic tree-width" $>k$ and "expansion" $\xi$ of "semantic tree-width" $>k$ there may be another "expansion" $\xi'$ such that $\xi' \homto \xi$ ("ie", such that $\xi \contained \xi'$).  
In fact, for classes of "CRPQs" avoiding such problematic behavior, Question~\ref{qu:FPTtractability} can be positively answered. We next show why.

\AP
Let us call a "UC2RPQ" ""finitely-redundant"" if there is no infinite chain $\xi_1(\bar x) \strcontained \xi_2(\bar x) \strcontained \dotsb$ among its "expansions". See \Cref{fig:example-infinite-chain} for a non-example.
\begin{figure}
    \centering%
    \includegraphics[width=.7\linewidth]{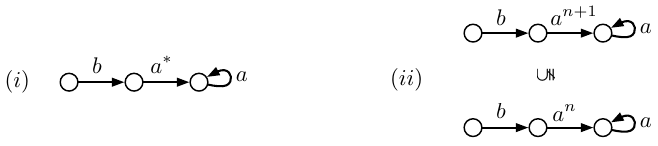}
    \caption{%
        \AP\label{fig:example-infinite-chain}%
        $(i)$ A simple non-"finitely-redundant" Boolean "CRPQ". $(ii)$ For every $n$, there is a "homomorphism" from the $(n+1)$-"expansion" to the $n$-"expansion", but no "homomorphism" in the converse direction.
    }
\end{figure}
\AP
Observe that the classes of "CQs" and "UCQs" are "finitely-redundant", and also the class of "CRPQs" with ``""no directed cycles""'', meaning no directed cycle in its underlying directed graph and no empty word $\epsilon$ in the "atom" languages.

\begin{lem}
    The class of "CRPQs" with "no directed cycles" is "finitely-redundant".
\end{lem}
\begin{proof}
    By means of contradiction, let $\gamma$ have "no directed cycles" and suppose there is an infinite chain $\xi_1(\bar x) \strcontained \xi_2(\bar x) \strcontained \dotsb$ of "expansions" of $\gamma$. Hence, there must be an "atom" "expansion" which grows arbitrarily in the chain. Take $\xi_i$ such that it contains an "atom" "expansion" of size bigger than $\xi_1$. Since such "atom" "expansion" is a directed path (as we are dealing with one-way "CRPQs"), the fact that $\xi_i \homto \xi_1$ implies that there is some cycle in $\xi_1$.
    Since $\gamma$ cannot contain the empty word in the "atom" languages, this is in contradiction with the hypothesis that there are "no directed cycles" in $\gamma$.
\end{proof}

We next show that, restricted to classes of "finitely-redundant" "UC2RPQ", we can obtain a characterization of "evaluation" in "FPT".
\thmtractabilityfinred
\begin{proof}
    \proofcase{Left-to-right} By contraposition, we show that if $\+C$ has unbounded "tree-width", then its "evaluation problem" is \wone-hard via an "FPT"-reduction from the 
    \AP
    ""parameterized clique problem"". This is the problem of, given a parameter $k$ and a simple graph $G$, whether $G$ contains a $k$-clique. We do this by a simple adaptation of the proof of Grohe \cite[Theorem~4.1]{Grohe07} for the case of "CQs".

    Given an instance $\langle G,k \rangle$ of the "parameterized clique problem", the idea is to first search for a query $\gamma \in \+C$ of ``sufficiently large'' "semantic tree-width". 
    
    \AP
    Let us call an "expansion" $\xi$ of a "UC2RPQ" to be ""maximal"" if there is no other "expansion" $\xi'$ such that $\xi \strcontained \xi'$.
    \AP
    The ""core"" of a "CQ" is the result of repeatedly removing any atom which results in an equivalent query. It is unique up to isomorphism (see, "eg", \cite{DBLP:conf/stoc/ChandraM77}), and a "CQ" has "semantic tree-width" $k$ if{f} its "core" has "tree-width" $k$ \cite[Theorem 12]{DBLP:conf/cp/DalmauKV02}. We say that a "CQ" is `a core' if it is isomorphic to its "core".

    \begin{proposition}
        \label{prop:big-expansion}
        For any $k \geq 3$, if a "finitely-redundant" "C2RPQ" has "semantic tree-width" $\geq k$, then there is a "maximal" "expansion" thereof of "semantic tree-width" $\geq k$.
    \end{proposition}

    \begin{proof}
        Let $\gamma$ be a "finitely-redundant" "C2RPQ".
        Consider the "infinitary UCQ" \[\Xi \defeq \{\core(\xi) \mid \xi \text{ is a "maximal" "expansion" of } \gamma \}.\]
        Since $\gamma$ is "finitely-redundant", we have $\gamma \semequiv \Xi$.
        We prove the fact by contraposition.
        If all "maximal" "expansions" of $\gamma$ have "semantic tree-width" $\leq k-1$,
        then all "CQs" of $\Xi$ have "tree-width" $\leq k-1$, and so
        by the implication $\itemClosureInfCQ \Rightarrow \itemClosureUCRPQSimple$ of \Cref{thm:closure-under-sublanguages}, query $\gamma$ has "semantic tree-width" at most $k-1$.
        Note that for \Cref{thm:closure-under-sublanguages} to apply, we need $k-1 > 1$
        "ie" $k \geq 3$.
    \end{proof}

    \begin{proposition}
        The set of all "maximal" "expansions" of queries from $\+C$ is recursively enumerable.
    \end{proposition}

    \begin{proof}
        \AP We first show that given an "expansion" $\xi$ of some "C2RPQ" $\gamma$, it is 
        decidable whether $\xi$ is "maximal". This follows from the following claim:
        there exists an "expansion" $\xi'$ of $\gamma$ "st" $\xi' \homto \xi$ and
        $\xi \nothomto \xi'$ "iff"
        there exists such an "expansion" whose "atom expansions" have length at most
        $2^m \cdot |\xi|\cdot |\A|^{2|\xi|}$ where $|\xi|$ is the number of variables of $\xi$ and
        $m$ is the greatest number of states of an NFA labelling an "atom" of $\gamma$. 
        Decidability of "maximality" clearly follows from this claim: it suffices to check if 
        $\xi' \homto \xi$ implies $\xi \homto \xi'$ for all ``small'' $\xi'$.

        To prove the claim, let $\xi'$ be an "expansion" of $\gamma$, and assume that there is 
        a homomorphism $f\colon \xi' \to \xi$ and that $\xi \nothomto \xi'$. Consider an "atom expansion"
        \[
            \pi' = x_0 \atom{a_1} x_1 \atom{a_2} \cdots \atom{a_{n-1}} x_{n-1} \atom{a_n} x_n
        \]
        of $\xi'$, and let $\+A$ denote the NFA associated with the "atom".
        For any index $i \in \lBrack 0,n\rBrack$ which is neither among the $|\xi|$ first positions
        nor the $|\xi|$ last positions, define its type $\tau_i$ as
        the word $a_{i-|\xi|}-1 \cdots a_i a_{i+1} \cdots a_{i+|\xi|}$ of length $2|\xi|$---note that $\tau_i$ uniquely describes the ball of radius $|\xi|$ centred at $x_i$ in $\xi$.
        Consider the function which maps index $i \in \lBrack |\xi|, n-|\xi|+1\rBrack$
        to the pair $\langle f(x_i), Q_i , \tau_i \rangle$, where $Q_i$ is the set of states $q$
        of $\+A$ which admit a path from an initial state to $q$ labelled by $a_1\cdots a_i$.
        If $n \geq |\xi|\cdot 2^{|\+A|}\cdot |\A|^{2|\xi|} + 2|\xi|$ then by the pigeon-hole principle,
        there exists $i,j \in \lBrack |\xi|, n-|\xi|+1\rBrack$ "st" $i<j$, $f(x_i) = f(x_j)$,
        $Q_i = Q_j$ and $\tau_i = \tau_j$. Letting
        \[
            \pi'' = x_0 \atom{a_1} x_1 \atom{a_2} \cdots \atom{a_{i-1}} x_{i-1} \atom{a_i} x_i = x_j
            \atom{a_{j+1}} x_{j+1} \atom{a_{j+2}} \cdots \atom{a_{n-1}} x_{n-1} \atom{a_n} x_n,
        \]
        consider the query $\xi''$ obtained from $\xi'$ 
        by replacing $\pi'$ with $\pi''$.
        Since $Q_i = Q_j$, $\xi''$ is still an "expansion" of $\gamma$.
        Moreover, $f(x_i) = f(x_j)$ implies that there is a "homomorphism" from $\xi''$ to $\xi$.
        Lastly, it there was a "homomorphism" from $\xi$ to $\xi''$,
        then this homomorphism should contain $x_i$ in its image---otherwise there would clearly be a "homomorphism" from $\xi$ to $\xi'$. Note that the image of this "homomorphism"
        is included in the ball of $\xi''$ centered at $x_i = x_j$ of radius $|\xi|$.
        But since $\tau_i = \tau_j$ this ball is equal to the ball of $\xi'$ centered at $x_i$
        (or equivalently at $x_j$) of radius $|\xi|$, and so we found a "homomorphism" from
        $\xi$ to $\xi'$, which is not possible. Hence, there cannot be any "homomorphism" from
        $\xi$ to $\xi''$, which concludes the proof.

        Finally, to enumerate all "maximal" "expansions" of queries from $\+C$,
        it suffices to enumerate all "expansions" of queries from $\+C$---which is doable since $\+C$ is recursively enumerable---and only keep those which are "maximal", using the previous algorithm.
    \end{proof}

    We proceed with the reduction.
    For any value of $k$ which is big enough, we enumerate all "maximal" "expansions" of $\+C$ until we find one such "expansion" $\xi$ whose "core" contains a $K \times K$ grid as a minor, for $K = {k \choose 2}$.
    We know that this must happen by \Cref{prop:big-expansion} and the Excluded Minor Theorem \cite{RobertsonS86}, stating that there exists a function $f : \Nat \to \Nat$ such that for every $n \in \Nat$ every graph of "tree-width" at least $f(n)$ contains a $(n \times n)$-grid as a minor.
    Once we get hold of such a "maximal" "expansion" $\xi$, we proceed as in \cite[proof of Theorem~4.1]{Grohe07} to produce, in polynomial time, a "graph database" $G_\xi$ such that:
    \begin{enumerate}
        \item there is a "homomorphism" $G_\xi \homto \xi$, and
        \item $G_\xi$ "satisfies" $\xi$ if, and only if, $G$ has a clique of size $k$.
    \end{enumerate}
Now consider the "UC2RPQ" $\Gamma \in \+C$ of which $\xi$ is an expansion, and observe that if $G_\xi$ "satisfies" $\Gamma$, then we must have that $G_\xi$ also "satisfies" $\xi$, by the fact that $G_\xi \homto \xi$ and $\xi$ is "maximal". Hence, the following are equivalent:
\begin{itemize}
    \item $G_\xi$ "satisfies" $\Gamma$, 
    \item $G_\xi$ "satisfies" $\xi$,
    \item $G$ contains a $k$-clique.
\end{itemize}
This finishes the "FPT"-reduction.

\medskip

\proofcase{Right-to-left} This direction does not need any of the hypotheses (neither "finite-redundancy",  $\wone \neq $ "FPT", nor r.e.), by \Cref{coro:fpt-eval-bounded-semtreewidth}.
\end{proof}

\subsection{\AP{}Larger Classes}
\label{sec:discussion-larger-classes}

A natural and simple approach to extend the expressive power of "CRPQs" is to close the queries by transitive closure. That is, given a binary "CRPQ" $\gamma(x,y)$ we can consider "CRPQ" over the extended alphabet $\A \cup \set \gamma$, where the label $\gamma$ is interpreted as the binary relation defined by $\gamma(x,y)$. \AP This is the principle behind ""Regular Queries"" \cite{DBLP:journals/mst/ReutterRV17}. The notion of "tree-width" can be easily lifted to this class, and classes of bounded "tree-width" still have a polynomial-time "evaluation problem". However, this class has not yet been studied in the context of the "semantic tree-width". It is not known if the "semantic tree-width $k$ problem" is decidable, nor whether classes of bounded "semantic tree-width" have an "FPT" "evaluation problem".
\begin{qu}
    Is the "semantic tree-width $k$ problem" for "Regular Queries" decidable?
\end{qu}

\subsection{\AP{}Different Notions}
\label{sec:discussion-different-notions}

\begin{table}[bht]
    \begin{tabular}{ccc}
        \toprule
        Query class & Membership problem & "Evaluation problem" \\ \midrule
        "path-width" $\leq k$ & 
        "L"-c {\footnotesize\cite[Theorem 1.3, p. 2]{KM10Computing}}
        & "NL"-c {\footnotesize(\Cref{lemma:evaluation-bounded-pathwidth})}  \\
        "sem. path-w." $\leq k$ & "2ExpSpace" \& "ExpSpace"-h & "paraNL" {\footnotesize(\Cref{thm:evaluation-bounded-pathwidth})} \\ 
        & {\footnotesize(\Cref{thm:decidability-sempw})} \\ \midrule
        "tree-width" $\leq k$ & "L"-c {\footnotesize \cite[Lemma 1.4]{EJT10Logspace}} & "P" {\footnotesize(Folklore)\footnotemark} \\
        "sem. tree-w." $\leq k$ & "2ExpSpace" \& "ExpSpace"-h\footnotemark & "FPT" {\footnotesize\cite[Corollary V.2]{DBLP:conf/lics/0001BV17}}\footnotemark \\
        & {\footnotesize(\Cref{thm:decidability-semtw})} & "NP"-c {\footnotesize\cite[Theorem V.3]{DBLP:conf/lics/0001BV17}}\\ \bottomrule
    \end{tabular}
    \caption{
        \label{table:tractability-classes-queries}
        Complexity of the membership and "evaluation problem" for
        some classes of "UC2RPQs" studied in this paper, where $k \geq 1$ is fixed.
        The same results hold for the "contracted" variants.
        The abbreviation ``-c'' (resp. ``-h'') stands for ``-complete'' (resp. ``-hard'').
    }
\end{table}
\addtocounter{footnote}{-2}
\footnotetext{Originally proven by Chekuri \& Rajaraman \cite[Theorem 3]{CHEKURI2000211} for CQs. The generalization to "UC2RPQs" is trivial, see "eg" \Cref{prop:crpq-bound-tree-width-upper-bound}
or \cite[Theorem IV.3]{DBLP:conf/lics/0001BV17}.}
\stepcounter{footnote}
\footnotetext{See also \cite[Theorem 6.1]{BarceloRV16} for $k=1$.}
\stepcounter{footnote}
\footnotetext{See also \Cref{coro:fpt-eval-bounded-semtreewidth} and \cite[Theorem~22]{FGM24}.}

"CRPQs" of small "tree-width" or "path-width" enjoy a tractable "evaluation problem", see \Cref{table:tractability-classes-queries}. However, it must be noticed that "containment" between "tree-width" $k$ or "path-width" $k$ queries is still very hard: \expspace-complete (even for $k=1$) \cite{CGLV00}. The more restrictive measure of ``"bridge-width"'' \cite{Figueira20} has been proposed as a more robust measure, which results in classes of queries which are well-behaved both for "evaluation" (since "bridge-width" $k$ implies "tree-width" $\leq k$) and for "containment" (since "containment" of bounded-"bridge-width" classes is in \pspace).
It is not hard to see that "bridge-width" is closed under "refinements", and thus that this notion is amenable to our approach ("cf" \Cref{obs:equivalence_under_approx_homomorphism}).
\begin{qu}
    \AP\label{qu:semantic-bridge-width}
    Is the problem of whether a "UC2RPQ" is equivalent to a "UC2RPQ" of "bridge-width" at most $k$ decidable?
\end{qu}

\clearpage\appendix
\section{\AP{}Polynomial-Time Evaluation of Queries of Bounded Tree-Width}
\label{apdx-sec:prop:crpq-bound-tree-width-upper-bound}

\crpqboundtwupperbound*
\begin{proof}
    \AP
    \knowledgenewrobustcmd{\sjoin}{\cmdkl{\ltimes}}
    \knowledgenewrobustcmd{\costjoin}{\cmdkl{c_{sj}}}
    A ""semi-join"" is a "CQ" of the form $q(\bar x) = R(\bar x) \land S(\bar y)$, noted $R(\bar x) \intro*\sjoin S(\bar y)$, where $\bar x$ and $\bar y$ may contain common variables and constants. \AP ""Yannakakis algorithm"" \cite{Yannakakis81} allows to evaluate any Boolean "acyclic"%
	\footnote{%
	An \AP""acyclic"" "CQ" over arbitrary relational structures is one which admits a "tree decomposition" whose set of "bags" are the sets of variables of its atoms (also known as ``generalized hyper tree-width 1'' or ``$\alpha$-acyclicity'' of its underlying hypergraph).
	} 
	"conjunctive query" $q$ on an $n$-tuple relational database in $\+O(|q| \cdot \costjoin(n))$, where $\intro*\costjoin(n)$ is the cost of performing a "semi-join": $\+O(n \cdot \log(n))$ for Turing machines, or $\+O(n)$ for a RAM model.\footnote{
		We use Random Access Machines (RAMs) with domain $\Nat$, in which we assume that the RAM's memory is initialized to 0. 
		For every fixed dimension $d \geq 1$  we have available an unbounded number of $d$-ary arrays $A$ such that for given $(n_1, \dotsc, n_d) \in \Nat^d$ the entry $A[n_1, \dotsc, n_d]$ at position $(n_1, \dotsc, n_d)$ can be accessed in constant time.
		To compute $R(\bar x \, \bar y) \sjoin S(\bar y \, \bar z)$, where $\bar y$ are the only common variables between the atoms and $n_R$ and $n_S$ are the number of tuples in $R$ and $S$ respectively, we may assume that constants from the relations are encoded as numbers in binary, of size $\+O(\log(n_R + n_S))$. 
		We first encode the relation $R$ projected onto $\bar y$, as an array $A$ of the dimension of $\bar y$ in $\+O(n_R)$, by putting a `1' in $A$ for each tuple of $R$. 
		Then, for each tuple of $S$ we test if the $\bar y$-projection belongs to $A$, if so we put a `2' in the array $A$, this is $\+O(n_S)$. 
		Finally, for each tuple of $R$ we output the tuple if its projection onto $\bar y$ has a `2' in $A$. This last step takes $\+O(n_R)$.
	}

    We reduce the "evaluation problem" for "C2RPQs" of "tree-width" $k$ to the "evaluation problem" of Boolean acyclic "CQs". First, if we receive $\Gamma (\bar x)$, $G$ and $\bar v$ as input, we replace the "free variables" $\bar x$ of $\Gamma$ with the "graph database" nodes $\bar v$ as constants, to obtain a Boolean "C2RPQ" (with constants). Let us therefore assume that $\Gamma$ is Boolean.
	Take any $\gamma \in \Gamma$ and a "tagged tree decomposition" $(T, \bagmap, \tagmap)$ of "width" at most $k$ of $\gamma$.\footnote{By a "tagged tree decomposition" of $\gamma$ we mean one for the identity $\gamma \homto \gamma$. The definition of "tagged tree decompositions" can be found in \Cref{sec:treedec}.}
	It is easy to see that we can assume that the decomposition is of linear size in the number of "atoms" $\nbatoms{\gamma}$.%
	\footnote{
		Indeed, observe that the decomposition resulting from contracting any edge so that one "bag" is contained or equal to the other is of linear size and preserves the "width" (see, "eg", \cite{DBLP:journals/jcss/GottlobLS02}).
	}
	 For every "bag" $v$ of the decomposition containing $\bagmap(v)=\set{x_1, \dotsc, x_t}$ consider the following relation $R_v$ consisting of all tuples of "database" nodes $(v_1, \dotsc, v_t)$ such that for every "atom" $A=x_i \atom{L} x_j$ in $\tagmap^{-1}(v)$ we have that $(v_i, v_j)$ "satisfies" $A$.

    It then follows that 
    \begin{enumerate}
        \item  each $R_v$ is of size $\+O(|G|^{k+1})$, 
        \item the "Boolean" "CQ" $q = \bigwedge_{v \in \vertex{T}} R_v(\bagmap(v))$ is "acyclic" and of size linear in $\nbatoms{\gamma}$, and 
        \item $q$ is "equivalent" to $\gamma$.
    \end{enumerate}
    Hence, "Yannakakis algorithm" yields a complexity of $\+O(\size{\gamma} \cdot \costjoin(|G|^{k+1}))$. That is, $\+O(\size{\gamma} \cdot |G|^{k+1} \cdot \log{|G|})$, or $\+O(\size{\gamma} \cdot |G|^{k+1})$ in a RAM model.

    We are only left with the cost of computing the relation $R_v$ for any given $v \in \vertex{T}$. Let $\+A_v= \set{A_1(y_1,z_1), \dotsc, A_s(y_s,z_s)}$ be the set of all "atoms" of $\tagmap^{-1}(v)$. Let $\bagmap(v)=\set{x_1, \dotsc, x_t}$.
    Compute first the relation $S_i = \set{ (u,u') \in \vertex{G}^2 : (u,u') \text{ "satisfies" } A_i(y_i,z_i)}$ for every $i$; this can be done in $\+O(|G| \cdot |L_i \times \Gpm|)$, where $L_i \times \Gpm$ is the product of the NFA for the regular language $L_i$ of $A_i$ and the expanded "database" $\Gpm$.\footnote{That is, $L_i \times \Gpm$ is a  "database" having pairs $(q,v)$ where $q$ is a state for the NFA $\+A_i$ of $L_i$ and $v \in \vertex{\Gpm}$, and we have an edge $(q,v) \xrightarrow{a} (q',v')$ in $L_i \times \Gpm$ if  $q \xrightarrow{a} q'$ and $v \xrightarrow{a} v'$ are in $\+A_i$ and $\Gpm$ respectively.} Compute also the relation $U = \vertex{G}^t$ in $\+O(|G|^{t})$ (hence in $\+O(|G|^{k+1})$). Finally, we compute $R_v$ by performing $s$ nested "semi-joins" 
    \[
        ((U(x_1, \dotsc, x_t) \sjoin S_1(y_1,z_1)) \sjoin S_2(y_2,z_2) \dotsb )\sjoin S_s(y_s,z_s) 
    \]
    in $\+O(s \cdot |G|^{k+1} \cdot \costjoin(|G|^{k+1}))$, that is, in $\+O(s \cdot |G|^{k+1} \cdot \log(|G|))$ or $\+O(s \cdot |G|^{k+1})$ in a RAM model. Repeating this for every "bag" we can compute all the $R_v$'s in $\+O(\size{\gamma} \cdot |G|^{k+1} \cdot \log(|G|))$ (observe that we iterate only once on each "atom", since we are using a "``tagged'' decomposition@tagged tree decomposition") or $\+O(\size{\gamma} \cdot |G|^{k+1})$ in a RAM model.
\end{proof}

\section{\AP{}Alternative Upper Bound for Containment of UC2RPQs}
\label{apdx-sec:alternative-upper-bound-containment}

In \Cref{sec:maximal-under-approximations}, in order to prove the "2ExpSpace" upper bound to the 
"semantic tree-width $k$ problem" (\Cref{lem:sem-tw-in-twoexp}), we proved an upper bound on containment
of "UC2RPQs" (\Cref{prop:bound-containment-pb}) by relying on the notion of "bridge-width".
In this section, we give a slightly different bound, which is more elementary (in the sense that it 
does not rely on "bridge-width") and still yields a "2ExpSpace" upper bound to the "semantic 
tree-width $k$ problem".

\begin{proposition}
	\AP\label{prop:bound-containment-pb-alt}
	The "containment problem" $\Gamma \contained \Delta$ between two "UC2RPQs" can be solved in non-deterministic space $2^{c\cdot \size{\Gamma}} + p_{\Delta} \cdot 2^{c\cdot m_\Delta}$, where $m_\Delta$ is the size of the greatest disjunct of $\Delta$, namely $m_\Delta = \max{\{\size{\delta_\Delta} \mid \delta \in \Delta\}}$, $p_\Delta$
	is the number of disjuncts of $\Delta$, and $c$ is a constant.
\end{proposition}

\begin{proof}[Proof sketch]
    The proposition can be shown by close inspection of the standard "containment problem" for "UC2RPQs" \cite[Theorem 5]{CGLV00}: the "containment problem" is reduced, in this instance, to checking the inclusion between NFAs of the form\footnote{$\+A_\Gamma$ and $\+A_\delta$ are denoted $A_1$ and $A_3$, respectively, in \cite{CGLV00}.}
	\begin{equation}
		\AP\label{eq:containment-to-inclusion-regexp}
		\+A_\Gamma \subseteq^? \bigcup_{\delta \in \Delta} \+A_\delta,
	\end{equation}
	where $A_\Gamma$ is a regular expression which is exponential in $\size{\Gamma}$,
	and $\+A_\delta$ has size exponential in $\size{\delta} \leq m_\Delta$.
	Should \eqref{eq:containment-to-inclusion-regexp} not hold, there must exist a counterexample of size at most 
	\[
		2^{|\+A_\Gamma|}\times \prod_{\delta \in \Delta} 2^{|\+A_\delta|}
	\]
	Letting $p_\Delta$ be the number
	of queries in $\Delta$, we get that the logarithm of the expression above---representing the
	size of the non-deterministic space needed by the algorithm---is upper bounded by
	\[
		c_0 \Bigl(|\+A_\Gamma| + \sum_{\delta \in \Delta} |\+A_\delta|\Bigr)
		\mathrel{\underset{\tiny\textrm{eventually}}{\leq}}
		2^{c\cdot \size{\Gamma}} + p_{\Delta} \cdot 2^{c\cdot n_\Delta},
	\]
	for some constants $c_0$ and $c$.
\end{proof}

\section{\AP{}Path-Width is not Closed under Refinements}
\label{apdx-sec:path-width-not-closed-refinements}

\pathwidthnotclosed*

\begin{proof}
	Let $X$ be a set of $k-1$ variables. Consider the undirected multigraph $\+G_k$ whose set of nodes is
	$X \cup \{y_0, y_1, y_2, y_3\}$ with the following edge set:
	\begin{itemize}
		\item each $X \cup \{y_i\}$ ($i \in \{0,1,2,3\}$) is a clique,
		\item there is an edge from $y_i$ to $y_{i+1}$  for $i \in \{0,1,2\}$, and
		\item there is a second edge from $y_1$ to $y_2$.
	\end{itemize}
	By definition, this graph has "path-width" exactly $k$: it is as least $k$ since it contains
	a $(k+1)$-clique---namely $X \cup \{y_i,y_{i+1}\}$---and, moreover the following sequence of bags---"cf" \Cref{subfig:pw-not-closed-original}---defines a "path decomposition" of $\+G_k$ of "width" $k$:
	\[
		\langle X \cup\{y_0,y_1\},\; X \cup\{y_1,y_2\},\; X \cup\{y_2,y_3\} \rangle.
	\]

	Let $\+G'_k$ be the graph obtained by "refining" the second edge from $y_1$ to $y_2$, into
	two edges $\langle y_1, z \rangle$ and $\langle z, y_2 \rangle$, where $z$ is a new variable---see \Cref{subfig:pw-not-closed-path-dec-horizontal,subfig:pw-not-closed-path-dec-vertical}.
	We claim that $\+G'_k$ has path-width at least $k+1$. Indeed, let
	$\langle T, \bagmap, \tagmap \rangle$ be a "path decomposition" of $\+G'_k$.

	Note that $X \cup \{y_0, y_1\}$, $X \cup \{y_1, y_2\}$, $X \cup \{y_2, y_3\}$
	and $\{z,y_1,y_2\}$ are cliques, so there must be "bags" of $\langle T, \bagmap, \tagmap \rangle$ containing each of them. Let $b_{0,1}$, $b_{1,2}$, $b_{2,3}$ and $b_Z$ denote these bags---note that they do not have to be distinct.

	\begin{enumerate}
		\item If $b_Z$ appears in $T$ between at least two "bags" among $b_{0,1}$, $b_{1,2}$
		and $b_{2,3}$ (as in \Cref{subfig:pw-not-closed-path-dec-horizontal}), since $X \subseteq \bagmap(b_{i,j})$ for all $(i,j)$, then $X \subseteq \bagmap(b_z)$.
		Hence $X \cup \{z,y_1,y_2\} \subseteq \bagmap(b_z)$
		and so $b_z$ has $k+2$ elements.
		\item Otherwise, "wlog" $b_z$ appears strictly before all three bags
		$b_{0,1}$, $b_{1,2}$ and $b_{2,3}$, as in \Cref{subfig:pw-not-closed-path-dec-vertical}.
		We consider the way $b_{0,1}$, $b_{1,2}$ and $b_{2,3}$ are ordered in the "path decomposition". If $b_{0,1}$ or $b_{2,3}$ appears first, then 
		they are located between $b_z$ and $b_{1,2}$, which both contain $\{y_1,y_2\}$,
		and so this bag must also contain $\{y_1,y_2\}$, and so it has size at least $k+2$.
		Otherwise, if $b_{1,2}$ appears first, depending on the relative ordering
		of $b_{0,1}$ and $b_{2,3}$, we either get that $y_2 \in \bagmap(b_{0,1})$
		or that $y_1 \in \bagmap(b_{2,3})$. In both cases, we have a bag with at least
		$k+2$ elements.
	\end{enumerate}

	In all cases, the "path decomposition" has width at least $k+1$, showing that
	$\+G'_k$ has "path-width" at least\footnote{In fact it has "path-width" exactly $k+1$.} $k+1$.
\end{proof}

\clearpage
\section*{Acknowledgment}
\noindent The work is supported by ANR QUID, grant ANR-18-CE40-0031.
\bibliographystyle{alphaurl}
\bibliography{long,bibli-crpq}

\end{document}